\documentclass[11pt]{article}
\usepackage{amsmath}
\usepackage{amssymb}
\usepackage{amsfonts}
\usepackage{epsfig}
\usepackage{mathrsfs}
\usepackage{graphicx}
\usepackage[scriptsize]{subfigure}

\newtheorem{theorem}{Theorem}
\newtheorem{proposition}{Proposition}
\newtheorem{definition}{Definition}
\newtheorem{lemma}{Lemma}
\newtheorem{corollary}{Corollary}

\newenvironment{proof}[1][Proof]{\noindent \textbf{#1.} }{\qedsymbol}
\newcommand{\qedsymbol}{\hspace{\fill}\rule{1.5ex}{1.5ex}}
\input{tcilatex}
\def\b0{\mbox{\boldmath $0$}}

\def\baselinestretch{1.28}
\begin{document}

\title{\vspace{-2.1cm} {\Huge Optimal Linear Precoding Strategies for
Wideband Non-Cooperative Systems based on Game Theory-Part I: Nash Equilibria%
}}
\author{Gesualdo Scutari$^{1}$, Daniel P. Palomar$^{2}$, and Sergio
Barbarossa$^{1}$ \\
%EndAName
{\small E-mail: $\{$aldo.scutari,sergio$\}$@infocom.uniroma1.it,
palomar@ust.hk}\\
$^{1\text{ }}${\small Dpt. INFOCOM, Univ. of Rome \textquotedblleft La
Sapienza\textquotedblright, Via Eudossiana 18, 00184 Rome, Italy} \\
$^{2}$ {\small Dpt. of Electronic and Computer Eng., Hong Kong Univ. of
Science and Technology, Kowloon Hong Kong.}}
\date{{\small Submitted to IEEE \textit{Transactions on Signal Processing},
September 22, 2005.}\\
{\small Revised March 14, 2007. Accepted June 5, 2007.\thanks{%
This work was supported by the SURFACE project funded by the European
Community under Contract IST-4-027187-STP-SURFACE.}}\vspace{-0.9cm}}
\maketitle
\begin{abstract}
In this two-parts paper we propose a decentralized strategy, based on a
game-theoretic formulation, to find out the optimal precoding/multiplexing
matrices for a multipoint-to-multipoint communication system composed of a
set of wideband links sharing the same physical resources, i.e., time and
bandwidth. We assume, as optimality criterion, the achievement of a Nash
equilibrium and consider two alternative optimization problems: 1) the
competitive maximization of mutual information on each link, given
constraints on the transmit power and on the spectral mask imposed by the
radio spectrum regulatory bodies; and 2) the competitive maximization of the
transmission rate, using finite order constellations, under the same
constraints as above, plus a constraint on the average error probability. In
Part I of the paper, we start by showing that the solution set of both
noncooperative games is always nonempty and contains only pure strategies.
Then, we prove that the optimal precoding/multiplexing scheme for both games
leads to a channel diagonalizing structure, so that both matrix-valued
problems can be recast in a simpler unified vector power control game, with
no performance penalty. Thus, we study this simpler game and derive
sufficient conditions ensuring the uniqueness of the Nash equilibrium.
Interestingly, although derived under stronger constraints, incorporating
for example spectral mask constraints, our uniqueness conditions have
broader validity than previously known conditions. Finally, we assess the
goodness of the proposed decentralized strategy by comparing its performance
with the performance of a Pareto-optimal centralized scheme. To reach the
Nash equilibria of the game, in Part II, we propose alternative distributed
algorithms, along with their convergence conditions.
\end{abstract}

%Submitted to IEEE \textit{Transactions on Signal Processing}, May 31, 2006.}
\vspace{-.5cm}

%\vspace{-1cm}

%\date{May 31, 2006}

\section{Introduction and Motivation}

In this two-parts paper, we address the problem of finding the optimal
precoding/multiplexing strategy for a multiuser system composed of a set of $%
Q$ noncooperative wideband links, sharing the same physical resources, e.g.,
time and bandwidth. No multiplexing strategy is imposed a priori so that, in
principle, each user interferes with each other. Moreover, to avoid
excessive signaling and the need of coordination among users, we
assume that encoding/decoding on each link is performed independently of the
other links. Furthermore, no interference cancellation techniques are used
and thus multiuser interference is treated as additive, albeit colored, noise. We
consider block transmissions, as a general framework encompassing most
current schemes like, e.g., CDMA or OFDM systems (it is also a
capacity-lossless strategy for sufficiently large block length \cite{Hir88,
Ral98}). Thus, each source transmits a coded vector
\begin{equation}
\mathbf{x}_{q}=\mathbf{F}_{q}\mathbf{s}_{q},  \label{x_q-block}
\end{equation}%
where $\mathbf{s}_{q}$ is the $N\times 1$ information symbol vector and $%
\mathbf{F}_{q}$ is the $N\times N$ precoding matrix. Denoting with $\mathbf{H%
}_{rq}$ the channel matrix between source $r$ and destination $q$, the
sampled baseband block received by the $q$-th destination is (dropping the
block index)\footnote{%
For brevity of notation, we denote as source (destination) $q$ the source
(destination) of link $q$.}%
\begin{equation}
\mathbf{y}_{q}=\mathbf{H}_{qq}\mathbf{x}_{q}+\sum_{r\neq q=1}^{Q}\mathbf{H}%
_{rq}\mathbf{x}_{r}+\mathbf{w}_{q},  \label{vector I/O}
\end{equation}%
where $\mathbf{w}_{q}$ is a zero-mean circularly symmetric complex Gaussian
white noise vector with covariance matrix $\sigma _{q}^{2}\mathbf{I}$.\footnote{We consider only white noise for simplicity, but the extension to colored noise is straightforward along well-known guidelines.} The
second term on the right-hand side of (\ref{vector I/O}) represents the
Multi-User Interference (MUI) received by the $q$-th destination and caused
by the other active links. Treating MUI as additive noise, the estimated
symbol vector at the $q$-th receiver is
\begin{equation}
\widehat{\mathbf{s}}_{q}=D\left[ \mathbf{G}_{q}^{H}\mathbf{y}_{q}\right] ,
\label{Rx_block}
\end{equation}%
where $\mathbf{G}_{q}^{H}$ is the $N\times N$ receive matrix (linear
equalizer) and $D[\cdot ]$ denotes the decision operator that decides which
symbol vector has been transmitted.

The above system model is sufficiently general to incorporate many cases of
practical interest, such as: i) digital subscriber lines, where the matrices
$(\mathbf{F}_{q})_{q=1}^{Q}$ incorporate DFT precoding and power allocation,
whereas the MUI is mainly caused by near-end cross talk \cite{Starr-Cioffi
Book}; ii) cellular radio, where the matrices $(\mathbf{F}_{q})_{q=1}^{Q}$
contain the user codes within a given cell, whereas the MUI is essentially
intercell interference \cite{Goldsmith-Wicker}; iii) ad hoc wireless
networks, where there is no central unit assigning the coding/multiplexing
strategy to the users \cite{Akyildiz-Wang}. The I/O model in (\ref%
{vector I/O}) is particularly appropriate for studying \textit{cognitive
radio} systems \cite{Haykin}, where each user is allowed to re-use portions
of the already assigned spectrum in an adaptive way, depending on the
interference generated by other users. Many recent works have shown that
considerable performance gain can be achieved by exploiting some kind of
information at the transmitter side, either in single-user \cite{Ral98},
\cite{Barbarossa}-\cite{Palomar-Barbarossa} or in multiple access or
broadcast scenarios (see, e.g. \cite{Goldsmith-MIMO}). Here, we extend this
idea to the system described above assuming that each destination has
perfect knowledge of the channel from its source (but not of the channels
from the interfering sources) and of the interference covariance matrix.

Within this setup, the system design consists on finding the optimal matrix
set $(\mathbf{F}_{q},\mathbf{G}_{q})_{q=1}^{Q}$ according to some
performance measure. In this paper we focus on the following two
optimization problems: P.1) the maximization of mutual information on each
link, given constraints on the transmit power and on the spectral radiation
mask; and P.2) the maximization of the transmission rate on each link, using
finite order constellations, under the same constraints as above plus a
constraint on the average (uncoded) error probability. The spectral mask
constraints are useful to impose radiation limits over licensed bands, where
it is possible to transmit but only with a spectral density below a
specified value. Problem P.2 is motivated by the practical need of using
discrete constellations, as opposed to Gaussian distributed symbols.

Both problems P.1 and P.2 are multi-objective optimization problems \cite%
{Miettinen}, as the (information/ transmission) rate achieved in each link
constitutes a different single objective. Thus, in principle, the
optimization of the transceivers %$(\mathbf{F}_{q},\mathbf{G}_{q})_{q=1}^{Q}$
requires a centralized computation (see, e.g., \cite{Cendrillon sub04,
Yu_DSL} for a special case of problem P.1, with diagonal transmissions and
no spectral mask constraints). This would entail a high complexity, a heavy
signaling burden, and the need for coordination among the users. Conversely,
our interest is focused on finding distributed algorithms to compute $(\mathbf{F}_{q},\mathbf{G}_{q})_{q=1}^{Q}$ with no
centralized control. To achieve this goal, we formulate the system design within
a game theory framework. More specifically, we cast both
problems P.1 and P.2 as strategic noncooperative (matrix-valued) games,
where every link is a player that competes against the others by choosing
its transceiver pair $(\mathbf{F}_{q},\mathbf{G}_{q})$ to maximize its own
objective (payoff) function. This converts the original multi-objective
optimization problem into a set of mutually coupled competitive
single-objective optimization problems (the mutual coupling is precisely
what makes the problem hard to solve). Within this perspective, we thus
adopt, as optimality criterion, the achievement of a Nash equilibrium, i.e.,
the users' strategy profile where every player is unilaterally optimum, in
the sense that no player is willing to change its own strategy as this would
cause a performance loss \cite{Osborne}-\cite{Rosen}. This criterion is
certainly useful to devise decentralized coding strategies. However, the
game theoretical formulation poses some fundamental questions: 1) Under
which conditions does a NE exist and is unique? 2) What is the
performance penalty resulting from the use of a decentralized strategy as
opposed to the Pareto-optimal centralized approach? 3) How can the Nash
equilibria be reached in a totally distributed way? 4) What can be said
about the convergence conditions of distributed algorithms? In Part I of this
two-part paper, we provide an answer to questions 1) and 2). The answer to
questions 3) and 4) is given in Part II.

Because of the inherently competitive nature of a multi-user system, it is
not surprising that game theory has been already adopted to solve many
problems in communications. Current works in the field can be divided in two
large classes, according to the kind of games dealt with: \emph{scalar} and
\emph{vector }power control games. In scalar games, each user has only one
degree of freedom to optimize, typically the transmit power or rate, and the
solution has been provided in a very elegant framework, exploiting the
theory of the so called \textit{standard} functions \cite{Yates-Jsac}-\cite%
{Saraydar-Mandayam-singlecell}. The vector games are clearly more
complicated, as each user has several degrees of freedom to optimize, like
user codes or power allocation across frequency bins, and the approach based
on the \textquotedblleft standard\textquotedblright\ formulation of \cite%
{Yates-Jsac}-\cite{Sung-Leung} is no longer valid. \ A vector power control
game was proposed in \cite{Yu} to maximize the information rates (under
constraints on the transmit power) of two users in a DSL system, modeled as
a frequency-selective Gaussian interference channel. The problem was
extended to an arbitrary number of users in
%\cite{ChungISIT03, Yamashitay-Luo,
%Scutari-Barbarossa-ICASSP, Scutari-Barbarossa-SPAWC03, Luo-Pang}.
\cite{ChungISIT03}-\cite{Luo-Pang}. Vector power control problem in
flat-fading Gaussian interference channels was addressed in \cite{Tse}.

The original contributions of this paper with respect to the current
literature on vector games \cite{Yu}-\cite{Tse} are listed next. We consider
two alternative \emph{matrix-valued }games, whereas in \cite{Yu}-\cite%
{Scutari-Barbarossa-SPAWC03}, \cite{Tse} the authors studied a \emph{vector }%
power control game which can be obtained from P.1 as a special case, when the
diagonal transmission is imposed a priori and there are no spectral mask constraints. Problem P.2, at the best of the authors' knowledge,
is totally new. The matrix nature of the players' strategies and the
presence of spectral mask constraints make the analysis of both games P.1
and P.2 complicated and none of the results in \cite{Yu}-\cite{Tse} can be
successfully applied. Our first contribution is to show that the solution
set of both games is always nonempty and contains only pure (i.e.,
deterministic) strategies. More important, we prove that the diagonal
transmission from each user through the channel eigenmodes (i.e., the
frequency bins) is optimal, irrespective of the channel state, power budget,
spectral mask constraints, and interference levels. This result yields a
strong simplification of the original optimization, as it converts both
complicated \emph{matrix-valued} problems P.1 and P.2 into a simpler unified
\emph{vector} power control game, with no performance penalty.
Interestingly, such a simpler vector game contains, as a special case, the
game studied in \cite{Yu}-\cite{Scutari-Barbarossa-SPAWC03}, when the users
are assumed to transmit with the same (transmit) power and no spectral mask
constraints are imposed. The second important contribution of the paper is
to provide sufficient conditions for the uniqueness of the NE of our vector
power control game that have broader validity than those given in \cite{Yu}-%
\cite{Scutari-Barbarossa-SPAWC03}, \cite{Tse} (without mask constraints)
and, more recently, in \cite{Luo-Pang} (including mask constraints). Our
uniqueness condition, besides being valid in a broader context than those
given in \cite{Yu}-\cite{Tse}, exhibits also an interesting behavior not
deducible from the cited papers: It is satisfied as soon as the interlink
distance exceeds a critical value, almost \emph{irrespective} of the channel
frequency response. Finally, to assess the performance of the proposed
game-theoretic approach, we compare the Nash equilibria of the game with the
Pareto-optimal centralized solutions to the corresponding multi-objective
optimization. We also show how to modify the original game in order to make
the Nash equilibria of the modified game to coincide with the Pareto-optimal
solutions. Not surprisingly, the Nash equilibria of the modified game can be
reached at the price of a significant increase of signaling and coordination
among the users.

The paper is organized as follows. % After reviewing, in Section \ref%
%{preliminaries}, the basic concepts of game theory,
In Section \ref{System Model Section}, the optimization problems P.1 and P.2
are formulated as strategic noncooperative games. Section \ref{Unified
formulation} proves the optimality of the diagonal transmission and in
Section \ref{Sec:Existence-Uniqueness} the conditions for the existence and
uniqueness of the NE are derived. Section \ref{Sec: OFDA/CDMA interpretation}
gives a physical interpretation of the NE, with particular emphasis on the
way each user allocates power across the available subchannels. Section \ref%
{Sec:How_good_NE} assesses the goodness of the NE by comparing the
performance of the decentralized game-theoretic approach with the
centralized Pareto-optimal solution. Numerical results are given in Section %
\ref{Sec:Numerical Results}. Finally, in Section \ref{Conclusions}, the
conclusions are drawn. Part of this work already appeared in \cite%
{Scutari-Barbarossa-ICASSP, Scutari-Barbarossa-SPAWC03, Scutari_Thesis,
Scutari-Tech-Rep}. \vspace{-0.2cm} \vspace{-0.2cm}

\section{System Model and Problem Formulation\label{System Model Section}}

In this section we clarify the assumptions and constraints underlying the
model (\ref{vector I/O}) and we formulate the optimization problem addressed
in this paper explicitly. \vspace{-0.4cm}

\subsection{System model}

Given the I/O system in (\ref{vector I/O}), we make the following
assumptions:

\noindent \textbf{A.1 }Neither user coordination nor interference
cancellation is allowed; consequently encoding/decoding on each link is
performed independently of the other links. Hence, the overall system in (%
\ref{vector I/O}) is modeled as a \emph{vector} Gaussian interference
channel \cite{Cover}, where MUI is treated as additive colored noise;

\noindent \textbf{A.2} Each channel is modeled as a FIR filter of maximum
order $L_{h}$ and it is assumed to change sufficiently slowly to be
considered fixed during the whole transmission, so that the information
theoretical results are meaningful;

\noindent \textbf{A.3} In the case of frequency selective channels, with
maximum channel order $L_{h}$, a cyclic prefix of length $L\geq L_{h}$ is
incorporated on each transmitted block $\mathbf{x}_{q}$ in (\ref{x_q-block});

\noindent\textbf{A.4} A (quasi-) block synchronization among the users is
assumed, so that all streams are parsed into blocks of equal length, having
the same temporization, within an uncertainty at most equal to the cyclic
prefix length;

\noindent \textbf{A.5} The channel from each source to its own destination
is known to the intended receiver, but not to the other terminals; an error-free estimate of MUI covariance matrix is supposed to be available at each receiver.
%receiver is also assumed to estimate with no errors the covariance matrix of
%the disturbance coming from other links.
Based on this information, each
destination computes the optimal precoding matrix for its own link and
transmits it back to its transmitter through a low (error-free) bit rate
feedback channel.\footnote{%
In practice, both estimation and feedback are inevitably affected by errors.
This scenario can be studied by extending our formulation to games with
partial information \cite{Osborne, Aubin-book}, but this goes beyond the
scope of the present paper.}

Assumption \textbf{A.1} is motivated by the need of finding solutions,
possibly sub-optimal, but that can be obtained through simple distributed
algorithms, that require no extra signaling among the
users. This assumption is well motivated in many practical scenarios, where
additional limitations such as decoder complexity, delay constraints, etc.,
may preclude the use of interference cancellation techniques. Assumption
\textbf{A.3} entails a rate loss by a factor $N/(N+L)$, but it
facilitates symbol recovery. For practical systems, $N$ is sufficiently
large with respect to $L$, so that the loss due to CP insertion is
negligible. Observe that, thanks to the CP insertion, each matrix $%
\mathbf{H}_{rq}$ in (\ref{vector I/O}) resulting after having discarded the
guard interval at the receiver, is a Toeplitz circulant matrix. Thus, $%
\mathbf{H}_{rq}$ is diagonalized as $\mathbf{H}_{rq}=\mathbf{WD}_{rq}\mathbf{%
W}^{H}$, with $\mathbf{W\in
%TCIMACRO{\U{2102} }%
%BeginExpansion
\mathbb{C}
%EndExpansion
}^{N\times N}$ denoting the normalized IFFT matrix, i.e., $\left[ \mathbf{W}%
\right] _{ij}\triangleq e^{j2\pi (i-1)(j-1)/N}/\sqrt{N}$ for $i,j=1,\ldots
,N $ and $\mathbf{D}_{rq}$ is a $N\times N$ diagonal matrix, where $\left[
\mathbf{D}_{rq}\right] _{kk}%
%TCIMACRO{\TeXButton{def}{\triangleq}}%
%BeginExpansion
\triangleq%
%EndExpansion
\bar{H}_{rq}(k)/\sqrt{d_{rq}^{\gamma }}$ is the frequency-response of the
channel between source $r$ and destination $q$, including the path-loss $%
d_{rq}^{\gamma }$\ with exponent $\gamma $ and normalized fading $\bar{H}%
_{rq}(k),$ with $d_{rq}$ denoting the distance between transmitter $r$
and receiver $q.$

The physical constraints required by the applications are:

\noindent \textbf{Co.1} Maximum transmit power for each transmitter, i.e.,
\begin{equation}
E\left\{ \left\Vert \mathbf{x}_{q}\right\Vert _{2}^{2}\right\} =\frac{1}{N}%
%TCIMACRO{\TeXButton{Tr}{\mathsf{Tr}}}%
%BeginExpansion
\mathsf{Tr}%
%EndExpansion
\left( \mathbf{F}_{q}\mathbf{F}_{q}^{H}\right) \leq P_{q},
\label{power-constraint}
\end{equation}%
where $P_{q}$ is power in units of energy per transmitted symbol, and the
symbols are assumed to be, without loss of generality (w.l.o.g.), zero-mean
unit energy uncorrelated symbols, i.e., $E\left\{ \mathbf{s}_{q}(n)\mathbf{s}%
_{q}^{H}(n)\right\} =\mathbf{I}$. Note that different symbols may be drawn
from different constellations.

\noindent \textbf{Co.2} Spectral mask constraint, i.e.,
\begin{equation}
E\left\{ \left\vert [\mathbf{W}^{H}\mathbf{F}_{q}\mathbf{s}%
_{q}]_{k}\right\vert ^{2}\right\} =\left[ \mathbf{W}^{H}\mathbf{F}_{q}%
\mathbf{F}_{q}^{H}\mathbf{W}\right] _{kk}\leq \overline{p}_{q}^{\max
}(k),\quad \forall k\in \{1,\ldots ,N\},  \label{mask_on_F_q_}
\end{equation}%
where $\overline{p}_{q}^{\max }(k)$ represents the maximum power user $q$
is allowed to allocate on the $k$-th frequency bin.
\footnote{%
Observe that if $(1/N)\dsum_{k}\overline{p}_{q}^{\max }(k)\leq P_{q},$ we
obtain the trivial solution $[\mathbf{W}^{H}\mathbf{F}_{q}\mathbf{F}_{q}^{H}%
\mathbf{W}]_{kk}=\overline{p}_{q}^{\max }(k),$ $\forall k.$} Constraints in (%
\ref{mask_on_F_q_}) are imposed by radio spectrum regulations and attempt to
limit the amounts of interference generated by each transmitter over some
specified frequency bands.

\noindent \textbf{Co.3} Maximum tolerable (uncoded) symbol error rate (SER)
on each link, i.e.,\footnote{%
Given the symbol error probability $P_{e},$ the Bit Error Rate (BER) $P_{b}$
can be approximately obtained from $P_{e}$ (using a Gray encoding to map the
bits into the constellation points) as $P_{b}=P_{e}/\log _{2}(\left\vert
\mathcal{C}\right\vert ),$ where $\log _{2}(\left\vert \mathcal{C}%
\right\vert )$ is the number of bits per symbol, and $\left\vert \mathcal{C}%
\right\vert $ is the constellation size.}
\begin{equation}
P_{e,q}(k)%
%TCIMACRO{\TeXButton{def}{\triangleq}}%
%BeginExpansion
\triangleq%
%EndExpansion
\limfunc{Prob}\{\hat{s}_{q}(k)\neq s_{q}(k)\}\leq P_{e,q}^{\star },\quad
\forall k\in \{1,\ldots ,N\},  \label{Pe-constraint}
\end{equation}%
where $\hat{s}_{q}(k)$ is the $k$-th entry of $\mathbf{\hat{s}}_{q}$ given
in (\ref{Rx_block}). Another alternative approach to guarantee the required
quality of service (QoS) of the system is to impose an upper bound
constraint on the global average BER of each link, defined as $%
(1/N)\sum\nolimits_{k=1}^{N}P_{e,q}(k)$. Interestingly, in \cite%
{Dani-Ottersten} it was proved that equal BER constraints on each subchannel
as given in (\ref{Pe-constraint}), provide essentially the same performance
of those obtained imposing a global average BER constraint, as the average
BER is strongly dominated by the minimum of the BERs on the individual
subchannels. Thus, for the rest of the paper we consider BER constraints as
in (\ref{Pe-constraint}).\vspace{-0.3cm}

\subsection{Problem Formulation: Optimal Transceivers Design based on Game
Theory \label{Sec:GT formulation}}

In this section we formulate the design of the transceiver pairs $(\mathbf{F}_{q},%
\mathbf{G}_{q})_{q=1}^{Q}$ of system (\ref{vector I/O}) within the framework
of game theory, using as optimality criterion \ the concept of NE \cite%
{Osborne}-\cite{Rosen}. We consider two classes of payoff functions, as
detailed next.\vspace{-0.2cm}

\subsubsection{Competitive maximization of mutual information}

In this section we focus on the fundamental (theoretic) limits of system (%
\ref{vector I/O}), under \textbf{A}.\textbf{1}-\textbf{A}.\textbf{5, }and
consider the competitive maximization of information rate of each link,
given constraints \textbf{Co.1} and \textbf{Co.2. }Using \textbf{A}.\textbf{1},
the achievable information rate for user $q$ is computed as the maximum
mutual information between the transmitted block $\mathbf{x}_{q}$ and the
received block $\mathbf{y}_{q}$, \textit{assuming the other received signals
as additive (colored) noise}. It is straightforward to see that a (pure or
mixed strategy) NE is obtained if each user transmits using Gaussian
signaling, with a proper precoder $\mathbf{F}_{q}$. In fact, for each user,
given that all other users use Gaussian codebooks, the codebook that
maximizes mutual information is also Gaussian \cite{Cover}. Hence, given
\textbf{A}.\textbf{5}, the mutual information for the $q$-th user is \cite%
{Cover}%
\begin{equation}
{\limfunc{I}\nolimits_{q}}(\mathbf{F}_{q},\mathbf{F}_{-q})=\dfrac{1}{N}\log
\left( \left\vert \mathbf{I}+\mathbf{F}_{q}^{H}\mathbf{H}_{qq}^{H}\mathbf{R}%
_{\mathbf{-}q}^{-1}\mathbf{H}_{qq}\mathbf{F}_{q}\right\vert \right)
\label{R_qq}
\end{equation}%
where $\mathbf{R}_{-q}%
%TCIMACRO{\TeXButton{def}{\triangleq}}%
%BeginExpansion
\triangleq%
%EndExpansion
\sigma _{q}^{2}\mathbf{I}+\sum\limits_{r\neq q\hfill =1\hfill }^{Q}\mathbf{H}%
_{rq}\mathbf{F}_{r}\mathbf{F}_{r}^{H}\mathbf{H}_{rq}^{H}$ is the
interference plus noise covariance matrix, observed by user $q$, and $%
\mathbf{F}_{-q}%
%TCIMACRO{\TeXButton{def}{\triangleq}}%
%BeginExpansion
\triangleq%
%EndExpansion
\left( \mathbf{F}_{r}\right) _{r\neq q=1}^{Q}$ is the set of all the
precoding matrices, except the $q$-th one. Observe that, for each link, we
can always assume that the receiver is composed of an MMSE stage followed by some
other stage, since the MMSE is capacity-lossless. Thus, w.l.o.g., we assume
in the following that\footnote{%
It is straightforward to verify that the MMSE receiver in (\ref{MMSE_q}) is
capacity-lossless by checking that, for each $q,$ the mutual information
(for a given set of $(\mathbf{F}_{q})_{q=1}^{Q}$) after the equalizer $%
\mathbf{G}_{q},$ $\log (|\mathbf{I}+\mathbf{F}_{q}^{H}\mathbf{H}_{qq}^{H}%
\mathbf{G}_{q}(\mathbf{G}_{q}^{H}\mathbf{R}_{\mathbf{-}q}\mathbf{G}_{q})^{-1}%
\mathbf{G}_{q}^{H}\mathbf{H}_{qq}\mathbf{F}_{q}|)$ is equal to (\ref{R_qq}).}%
\begin{equation}
\mathbf{G}_{q}=\mathbf{R}_{-q}^{-1}\mathbf{H}_{qq}\mathbf{F}_{q}(\mathbf{I}+%
\mathbf{F}_{q}^{H}\mathbf{H}_{qq}^{H}\mathbf{R}_{-q}^{-1}\mathbf{HF}%
_{q})^{-1},\quad \forall q\in \{1,\ldots ,Q\}.  \label{MMSE_q}
\end{equation}

Hence, the strategy of each player reduces to finding the optimal precoding $%
\mathbf{F}_{q}$ that maximizes ${\limfunc{I}\nolimits_{q}}(\mathbf{F}_{q},%
\mathbf{F}_{-q})$ in (\ref{R_qq}), under constraints \textbf{Co.1} and
\textbf{Co.2}. \ Stated in mathematical terms, we have the following
strategic noncooperative game
\begin{equation}
\left(
%TCIMACRO{\TeXButton{G}{\mathscr{G}}}%
%BeginExpansion
\mathscr{G}%
%EndExpansion
_{1}\right) :\qquad \qquad \qquad
\begin{array}{ll}
\limfunc{maximize}\limits_{\mathbf{F}_{q}} & {\limfunc{I}\nolimits_{q}}(%
\mathbf{F}_{q},\mathbf{F}_{-q}) \\
\limfunc{subject}\limfunc{to} & \mathbf{F}_{q}\in {\mathscr{F}}_{q},%
\end{array}%
\qquad \forall q\in \Omega ,\qquad \qquad \qquad  \label{Rate-matrix-game}
\end{equation}%
where $\Omega
%TCIMACRO{\TeXButton{def}{\triangleq}}%
%BeginExpansion
\triangleq%
%EndExpansion
\{1,\ldots ,Q\}$ is the set of players (i.e., the links), ${\limfunc{I}%
\nolimits_{q}}(\mathbf{F}_{q},\mathbf{F}_{-q})$ is the payoff function of
player $q,$ given in (\ref{R_qq}), and ${\mathscr{F}}_{q}$ is the set of
admissible strategies (the precoding matrices) of player $q$, defined as
\begin{equation}
{\mathscr{F}}_{q}%
%TCIMACRO{\TeXButton{def}{\triangleq}}%
%BeginExpansion
\triangleq%
%EndExpansion
\left\{ \mathbf{F}_{q}\in \mathcal{\ \mathbb{C}}^{N\times N}:\frac{1}{N}%
%TCIMACRO{\TeXButton{Tr}{\mathsf{Tr}}}%
%BeginExpansion
\mathsf{Tr}%
%EndExpansion
\left( \mathbf{F}_{q}\mathbf{F}_{q}^{H}\right) \leq P_{q},\quad \left[
\mathbf{W}^{H}\mathbf{F}_{q}\mathbf{F}_{q}^{H}\mathbf{W}\right] _{kk}\leq
\overline{p}_{q}^{\max }(k),\quad \forall k=1,\ldots ,N\right\} .
\label{P_q_Q_q}
\end{equation}

The solutions to (\ref{Rate-matrix-game}) are the well-known Nash
equilibria, which are formally defined as follows.

\begin{definition}
\label{NE def} A (pure) strategy profile $\mathbf{F}^{\star }=\left( \mathbf{%
F}_{q}^{\ast }\right) _{q\in \Omega }\in {\mathscr{F}}_{1}\times \ldots
\times {\mathscr{F}}_{Q}$ \ is a NE of game ${%
%TCIMACRO{\TeXButton{G}{\mathscr{G}}}%
%BeginExpansion
\mathscr{G}%
%EndExpansion
}_{1}$ if
\begin{equation}
{\limfunc{I}\nolimits_{q}}(\mathbf{F}_{q}^{\star },\mathbf{F}_{-q}^{\star
})\geq {\limfunc{I}\nolimits_{q}}(\mathbf{F}_{q},\mathbf{F}_{-q}^{\star }),\
\text{\ \ }\forall \mathbf{F}_{q}\in {\mathscr{F}}_{q},\text{ }\forall q\in
\Omega .  \label{pure-NE}
\end{equation}
\end{definition}

The definition of NE as given in (\ref{pure-NE}) can be generalized to
contain \emph{mixed} strategies \cite{Osborne}, i.e., the possibility of
choosing a randomization over a set of pure strategies (the randomizations
of different players are independent). Hence, the mixed extension of the
strategic game ${%
%TCIMACRO{\TeXButton{G}{\mathscr{G}}}%
%BeginExpansion
\mathscr{G}%
%EndExpansion
}_{1}$ is given by ${%
%TCIMACRO{\TeXButton{G_bar}{\overline{\mathscr{G}}}}%
%BeginExpansion
\overline{\mathscr{G}}%
%EndExpansion
}_{1}=\left\{ \Omega ,\{{\overline{\mathscr{F}}}_{q}\}_{q\in \Omega
},\left\{ \overline{{\limfunc{I}}}_{q}\right\} _{q\in \Omega }\right\} ,$
where ${\overline{\mathscr{F}}}_{q}$ denotes the set of the probability
distributions over the set ${\mathscr{F}}_{q}$ of pure strategies. In game ${%
%TCIMACRO{\TeXButton{G_bar}{\overline{\mathscr{G}}}}%
%BeginExpansion
\overline{\mathscr{G}}%
%EndExpansion
}_{1}$, the strategy profile, for each player $q,$ is the probability
density function $f_{\mathbf{F}_{q}}(\mathbf{F}_{q})$ defined on ${%
\mathscr{F}}_{q}$ and the payoff function $\overline{{\limfunc{I}}}_{q}=%
\mathrm{E}_{f_{\mathbf{F}_{q}}}\mathrm{E}_{f_{\mathbf{F}_{-q}}}\{{\limfunc{I}%
\nolimits_{q}}\}$ is the expectation of ${\limfunc{I}\nolimits_{q}}$ defined
in (\ref{R_qq}) taken over the mixed strategies of all the players$.$ A
mixed strategy NE of a strategic game is defined as a NE of its mixed
extension \cite{Osborne}.

Observe that for the payoff functions defined in (\ref{R_qq}), we can indeed
limit ourselves to adopt pure strategies w.l.o.g., as we did in (\ref%
{Rate-matrix-game}). Too see why, consider the mixed extension ${%
%TCIMACRO{\TeXButton{G_bar}{\overline{\mathscr{G}}}}%
%BeginExpansion
\overline{\mathscr{G}}%
%EndExpansion
}_{1}$ of ${%
%TCIMACRO{\TeXButton{G}{\mathscr{G}}}%
%BeginExpansion
\mathscr{G}%
%EndExpansion
}_{1}$ in $($\ref{Rate-matrix-game}$)$. For any player $q$, we have
\begin{equation}
\mathrm{E}_{f_{\mathbf{F}_{q}}}\mathrm{E}_{f_{\mathbf{F}_{-q}}}\left\{ {%
\limfunc{I}\nolimits_{q}}\left( \mathbf{F}_{q},\mathbf{F}_{-q}\right)
\right\} \leq \mathrm{E}_{f_{\mathbf{F}_{-q}}}\left\{ {\limfunc{I}%
\nolimits_{q}}\left( \mathrm{E}_{f_{\mathbf{F}_{q}}}\left\{ \mathbf{F}%
_{q}\right\} ,\mathbf{F}_{-q}\right) \right\} ,\quad \forall \mathbf{F}%
_{-q}:f_{\mathbf{F}_{-q}}(\mathbf{F}_{-q})\in {\overline{\mathscr{F}}}_{-q},
\label{mixed_extension}
\end{equation}%
where ${\overline{\mathscr{F}}}_{-q}%
%TCIMACRO{\TeXButton{def}{\triangleq}}%
%BeginExpansion
\triangleq%
%EndExpansion
{\overline{\mathscr{F}}}_{1}\times \ldots \times {\overline{\mathscr{F}}}%
_{q-1}\times {\overline{\mathscr{F}}}_{q+1}\times \ldots \times {\overline{%
\mathscr{F}}}_{Q}$. The inequality in (\ref{mixed_extension}) follows from
the concavity of the function ${\limfunc{I}\nolimits_{q}}(\mathbf{F}_{q},%
\mathbf{F}_{-q})$ in $\mathbf{F}_{q}\mathbf{F}_{q}^{H}$ \cite{Boyd} and from
Jensen's inequality \cite{Cover}. Since the equality is reached if and only
if $\mathbf{F}_{q}$ reduces to a pure strategy (because of the strict
concavity of ${\limfunc{I}\nolimits_{q}}(\mathbf{F}_{q},\mathbf{F}_{-q})$ in
$\mathbf{F}_{q}\mathbf{F}_{q}^{H}$), whatever the strategies of the other
players are, every NE of the game is achieved using pure strategies.%
\footnote{%
This result was obtained independently in \cite{Tse}-\cite{Scutari-Tech-Rep}.%
}%\vspace{-0.2cm}

\subsubsection{Competitive maximization of transmission rates}

The optimality criterion chosen in the previous section requires the use of ideal Gaussian codebooks with a proper covariance matrix. In practice,
Gaussian codes are substituted with simple (suboptimal) finite order signal
constellations, such as Quadrature Amplitude Modulation (QAM) or Pulse
Amplitude Modulation (PAM), and practical (yet suboptimal) coding schemes.
Hence, in this section, we focus on the more practical case where the
information bits are mapped onto constellations of finite size (with
possibly different cardinality), and consider the optimization of the
transceivers $(\mathbf{F}_{q},\mathbf{G}_{q})_{q\in \Omega }$, in order to
maximize the transmission rate on each link, under constraints \textbf{Co.1}
$\div $ \textbf{Co.3}.

Given the signal model in (\ref{vector I/O}), where now each vector $\mathbf{%
s}_{q}%
%TCIMACRO{\TeXButton{def}{\triangleq}}%
%BeginExpansion
\triangleq%
%EndExpansion
(s_{q}(k))_{k=1}^{N}$ is drawn from a set of finite-constellations $(%
\mathcal{C}_{k,q})_{k=1}^{N}$ , i.e., $s_{q}(k)\in \mathcal{C}_{k,q},$ the
transmission rate of each link is simply the number of
transmitted bits per symbol, i.e.,\vspace{-0.3cm}
\begin{equation}
r_{q}=\sum\limits_{k=1}^{N}\log _{2}(\left\vert \mathcal{C}_{k,q}\right\vert
),  \label{Tx-rate}
\end{equation}%
where $\left\vert \mathcal{C}_{k,q}\right\vert $ denotes the size of
constellation $\mathcal{C}_{k,q}.$ The (uncoded) average error probability
of the $q$-th link on the $k$-th substream, as defined in (\ref%
{Pe-constraint}), under the Gaussian assumption, can be analytically
expressed, for any given set $(\mathbf{F}_{q},\mathbf{G}_{q})_{q\in \Omega }$
and $(\mathcal{C}_{k,q})_{k=1}^{N},$ as%
\begin{equation}
P_{e,q}(k)=\alpha _{k,q}\mathcal{Q}\left( \sqrt{\beta _{k,q}\limfunc{SINR}%
\nolimits_{k,q}}\right) ,  \label{average_P_e}
\end{equation}%
where $\alpha _{k,q}$ and $\beta _{k,q}$ are constants that depend on the
signal constellation, $\mathcal{Q}\left( \cdot \right) $ is the $\mathcal{Q}$%
-function \cite{Proakis}, and $\limfunc{SINR}\nolimits_{k,q}$ is defined as%
%
%\vspace{-0.3cm}
\begin{equation}
\limfunc{SINR}\nolimits_{k,q}%
%TCIMACRO{\TeXButton{def}{\triangleq}}%
%BeginExpansion
\triangleq%
%EndExpansion
\frac{\left\vert \left[ \mathbf{G}_{q}^{H}\mathbf{H}_{qq}\mathbf{F}_{q}%
\right] _{kk}\right\vert ^{2}}{\left[ \mathbf{G}_{q}^{H}\mathbf{R}_{_{q}}%
\mathbf{G}_{q}\right] _{kk}},  \label{SINR_qk}
\end{equation}%
with $\mathbf{R}_{_{q}}%
%TCIMACRO{\TeXButton{def}{\triangleq}}%
%BeginExpansion
\triangleq%
%EndExpansion
\mathbf{H}_{qq}\mathbf{F}_{q}\mathbf{F}_{q}^{H}\mathbf{H}_{qq}^{H}-\mathbf{H}%
_{qq}\mathbf{f}_{k,q}\mathbf{f}_{k,q}^{H}\mathbf{H}_{qq}^{H}+\mathbf{R}%
_{-q}, $ where $\mathbf{f}_{k,q}$ denotes the $k$-th column of $\mathbf{F}%
_{q},$ and $\mathbf{R}_{-q}=\sigma _{q}^{2}\mathbf{I}+\sum\limits_{r\neq
q\hfill =1\hfill }^{Q}\mathbf{H}_{rq}\mathbf{F}_{r}\mathbf{F}_{r}^{H}\mathbf{%
H}_{rq}^{H}$ (see, e.g., \cite{Barbarossa, Palomar-convex}).

According to the constraints \textbf{Co.3} in (\ref{Pe-constraint}), because
of (\ref{average_P_e}), the optimal linear receiver for each user $q$ can be
computed as the matrix $\mathbf{G}_{q}$ maximizing simultaneously all the $(%
\limfunc{SINR}\nolimits_{k,q})_{k=1}^{N}$ in (\ref{SINR_qk}), while keeping the set of precoding matrices $\left( \mathbf{F}_{q}\right) _{q\in
\Omega }$ and the constellations $(\mathcal{C}_{k,q})_{k=1,q\in \Omega }^{N}$%
fixed. This leads to the well-known Wiener filter for $\mathbf{G}_{q},$ as given
in (\ref{MMSE_q}) \cite{Barbarossa, Palomar-convex, Palomar-Barbarossa}, and
the following expression for the $\limfunc{SINR}$s in (\ref{SINR_qk}):
%\vspace{-0.3cm}
\begin{equation}
\limfunc{SINR}\nolimits_{k,q}=\limfunc{SINR}\nolimits_{k,q}(\mathbf{F}_{q},%
\mathbf{F}_{-q})=\frac{1}{\left[ (\mathbf{I}+\mathbf{F}_{q}^{H}\mathbf{H}%
_{qq}^{H}\mathbf{R}_{-q}^{-1}\mathbf{H}_{qq}\mathbf{F}_{q})^{-1}\right] _{kk}%
}-1,\quad k\in \{1,\ldots ,N\}.  \label{SINR_kq_MSE}
\end{equation}

Under the previous setup, each player has to choose the precoder $\mathbf{F}%
_{q}$ and the constellations $(\mathcal{C}_{k,q})_{k=1}^{N}$ that maximize
the transmission rate in (\ref{Tx-rate}), under constraints \textbf{Co.1} $%
\div $ \textbf{Co.3}. Since, for any given rate, the optimal combination of
the constellations $(\mathcal{C}_{k,q})_{k=1}^{N}$ would require an
exhaustive search over all the combinations that provide the desired rate, in
the following we adopt, as in \cite{Palomar-Barbarossa}, the classical
method to choose quasi-optimal combinations, based on the gap approximation
\cite{Forney, Goldsmith-Chua}.\footnote{%
In our optimization we will use, as optimal solution, the continuous bit
distribution obtained by the gap approximation, without considering the
effect on the optimality of the granularity and the bit cap. The performance
loss induced by these sources of distortion can be quantified using the
approach given in \cite{Palomar-Barbarossa}.} As a result, the number of
bits that can be transmitted over the $N$ substreams from the $q$-th
source,\ for a given family of constellations and a given error probability $%
P_{e,q}^{\star }$, is approximatively given by%
\begin{equation}
\limfunc{r}\nolimits_{q}(\mathbf{F}_{q},\mathbf{F}_{-q})=\frac{1}{N}%
\dsum\limits_{k=1}^{N}\log _{2}\left( 1\mathcal{+}\frac{\limfunc{SINR}%
\nolimits_{k,q}(\mathbf{F}_{q},\mathbf{F}_{-q})}{\Gamma _{q}}\right)
\label{Rate-gap}
\end{equation}%
where $\limfunc{SINR}\nolimits_{k,q}(\mathbf{F}_{q},\mathbf{F}_{-q})$ is
defined in (\ref{SINR_kq_MSE}), and $\Gamma _{q}\geq 1$ is the gap which
depends only on the constellations and on $P_{e,q}^{\star }.$ For $M$-QAM
constellations, e.g., if the error probability in (\ref{average_P_e})\ is
approximated by $P_{e,q}(k)\approx 4\mathcal{Q}\left( \sqrt{(3/(M-1))%
\limfunc{SINR}\nolimits_{k,q}}\right) ,$ the resulting gap is $\Gamma _{q}=(%
\mathcal{Q}^{-1}(P_{e,q}^{\star }/4))^{2}/3$ \cite{Palomar-Barbarossa}.

In summary, the structure of the game is
\begin{equation}
\left(
%TCIMACRO{\TeXButton{G}{\mathscr{G}}}%
%BeginExpansion
\mathscr{G}%
%EndExpansion
_{2}\right) :\qquad \qquad \qquad
\begin{array}{ll}
\limfunc{maximize}\limits_{\mathbf{F}_{q}} & \limfunc{r}\nolimits_{q}(%
\mathbf{F}_{q},\mathbf{F}_{-q}) \\
\limfunc{subject}\text{ }\limfunc{to} & \mathbf{F}_{q}\in {\mathscr{F}}_{q},%
\end{array}%
\qquad \forall q\in \Omega ,\qquad \qquad \qquad  \label{Rate-Game-gap}
\end{equation}%
where ${\mathscr{F}}_{q}$ and $\limfunc{r}\nolimits_{q}(\mathbf{F}_{q},%
\mathbf{F}_{-q})$ are defined in (\ref{P_q_Q_q}) and (\ref{Rate-gap}),
respectively. A{s in (\ref{Rate-matrix-game}), in the following we focus on
pure strategies only}.

\section{Optimality of the Channel-Diagonalizing Structure\label{Unified
formulation}}

We derive now the optimal set of precoding matrices $(\mathbf{F}_{q})_{q\in
\Omega }$ for both games $%
%TCIMACRO{\TeXButton{G}{\mathscr{G}}}%
%BeginExpansion
\mathscr{G}%
%EndExpansion
_{1}$ and $%
%TCIMACRO{\TeXButton{G}{\mathscr{G}}}%
%BeginExpansion
\mathscr{G}%
%EndExpansion
_{2},$ and provide a unified reformulation of the original complicated games
in a simpler equivalent form. The main result is summarized in the following
theorem.

\begin{theorem}
\label{Theo_multiccarrier_Info_rate_} An optimal solution to the
matrix-valued games $%
%TCIMACRO{\TeXButton{G}{\mathscr{G}}}%
%BeginExpansion
\mathscr{G}%
%EndExpansion
_{1}$ and $%
%TCIMACRO{\TeXButton{G}{\mathscr{G}}}%
%BeginExpansion
\mathscr{G}%
%EndExpansion
_{2}$ is
\begin{equation}
\mathbf{F}_{q}=\mathbf{W}\sqrt{\limfunc{diag}(\mathbf{p}_{q})},\quad \forall
q\in \Omega ,  \label{Optimal_F_q}
\end{equation}%
where $\mathbf{W}$ is the IFFT matrix, and $\mathbf{p}\triangleq (\mathbf{p}%
_{q})_{q\in \Omega },$ with $\mathbf{p}_{q}\triangleq (p_{q}(k))_{k=1}^{N},$
is the solution to the vector-valued game $%
%TCIMACRO{\TeXButton{G}{\mathscr{G}}}%
%BeginExpansion
\mathscr{G}%
%EndExpansion
,$ defined as%
\begin{equation}
\left(
%TCIMACRO{\TeXButton{G}{\mathscr{G}}}%
%BeginExpansion
\mathscr{G}%
%EndExpansion
\right) :\qquad \qquad \qquad
\begin{array}{l}
\limfunc{maximize}\limits_{\mathbf{p}_{q}}\quad \ R_{q}(\mathbf{p}_{q},%
\mathbf{p}_{-q}) \\
\limfunc{subject}\text{ }\limfunc{to}\text{\ \ \ }\mathbf{p}_{q}\in {%
%TCIMACRO{\TeXButton{P}{{\mathscr{P}}}}%
%BeginExpansion
{\mathscr{P}}%
%EndExpansion
}_{q}%
\end{array}%
,\qquad \forall q\in \Omega ,  \label{Rate Game}
\end{equation}%
where $R_{q}(\mathbf{p}_{q},\mathbf{p}_{-q})$ and ${\mathscr{P}}_{q}$ are
the payoff function and the set of admissible strategies of user $q,$
respectively, defined as%
\begin{equation}
R_{q}(\mathbf{p}_{q},\mathbf{p}_{-q})=\dfrac{1}{N}\dsum\limits_{k=1}^{N}\log
\left( 1+\dfrac{1}{\Gamma _{q}}\text{ }%
%TCIMACRO{\TeXButton{sinr}{\mathsf{sinr}}}%
%BeginExpansion
\mathsf{sinr}%
%EndExpansion
_{q}(k)\right) ,  \label{Rate}
\end{equation}%
{and}
\begin{equation}
{\mathscr{P}}_{q}\triangleq \left\{ \mathbf{p}_{q}\in \mathcal{\ \mathbb{R}}%
^{N}:\dfrac{1}{N}\ \sum_{k=1}^{N}p_{q}(k)\leq 1,\text{ }0\leq p_{q}(k)\leq
p_{q}^{\max }(k),\text{ \ }\forall k\in \{1,\ldots ,N\}\right\} ,
\label{admissible strategy set}
\end{equation}%
with $p_{q}^{\max }(k)\triangleq \overline{p}_{q}^{\max }(k)/P_{q},$
\begin{equation}
%TCIMACRO{\TeXButton{sinr}{\mathsf{sinr}}}%
%BeginExpansion
\mathsf{sinr}%
%EndExpansion
_{q}(k)=\frac{P_{q}\left\vert \bar{H}_{qq}(k)\right\vert
^{2}p_{q}(k)/d_{qq}^{\gamma }}{\sigma _{q}^{2}+\sum_{\,r\neq
q}P_{r}\left\vert \bar{H}_{rq}(k)\right\vert ^{2}p_{r}(k)/d_{rq}^{\gamma }}%
\triangleq \frac{\left\vert H_{qq}(k)\right\vert ^{2}p_{q}(k)}{%
1+\sum_{\,r\neq q}\left\vert H_{rq}(k)\right\vert ^{2}p_{r}(k)},
\label{SINR}
\end{equation}%
where $H_{rq}(k)\triangleq \bar{H}_{rq}(k)\sqrt{P_{r}/\left( \sigma _{q}^{2}%
\text{ }d_{rq}^{\gamma }\right) },$ and $\Gamma _{q}=1$ if $\
%TCIMACRO{\TeXButton{G}{\mathscr{G}}}%
%BeginExpansion
\mathscr{G}%
%EndExpansion
_{1}$ is considered.
\end{theorem}

\begin{proof}
See Appendix \ref{proof_Theo_multiccarrier_Info_rate}.
\end{proof}

\medskip

\noindent \textbf{Remark 1 $-$ Optimality of the diagonal transmission.}
According to Theorem \ref{Theo_multiccarrier_Info_rate_}, a NE of both games
${{\mathscr{G}}}_{1}$ and ${{\mathscr{G}}}_{2}$ is reached using, for each
user, a diagonal transmission strategy through the channel eigenmodes (i.e.,
the frequency bins), irrespective of the channel realizations, power budget,
spectral mask constraints and MUI. This result simplifies the original
matrix-valued optimization problems (\ref{Rate-matrix-game}) and (\ref%
{Rate-Game-gap}), as the number of unknowns for each user reduces from $%
N^{2} $ (the original matrix $\mathbf{F}_{q})$ to $N$ (the power allocation
vector $\mathbf{p}_{q}=(p_{q}(k))_{k=1}^{N}$, with no performance loss.

Observe that the optimality of the diagonalizing structure was well known in
the single-user case, when the optimization criterion is the maximization of
mutual information and the constraint is the average transmit power \cite%
{Barbarossa}-\cite{Palomar-Barbarossa}, \cite{Dani-Ottersten}. However,
under the additional constraint on the spectral emission masks, the
optimality of the diagonal transmission has never been proved, neither in a
single-user nor in a multi-user competitive scenario. But, most
interestingly, Theorem \ref{Theo_multiccarrier_Info_rate_} proves the
optimality of the diagonal transmission also for game ${{\mathscr{G}}}_{2},~$%
\ where each player maximizes the transmission rate, using finite order
constellations, and under constraints on the spectral emission mask,
transmit power, and average error probability. In such a case, the
optimality of the channel-diagonalizing scheme was not at all clear.
Previous works on this subject adopted the typical approach used in
single-user MIMO systems \cite{Yu}-\cite{Scutari-Barbarossa-SPAWC03}: They
first imposed the diagonal transmission and then employed the gap
approximation solution over the set of parallel subchannels. However, such a
combination of channel diagonalization and gap approximation was not proved
to be optimal. Conversely, Theorem \ref%
{Theo_multiccarrier_Info_rate_} proves the optimality of this approach and
it subsumes, as particular cases, the results of \cite{Yu}-\cite%
{Scutari-Barbarossa-SPAWC03}, corresponding to the simple case where there are no
mask constraints.

It is also worth noticing that the optimality of the diagonalizing structure is
 a consequence of the property that all channel matrices, under
assumptions \textbf{A.2} and \textbf{A.3}, are diagonalized by the \textit{%
same} matrix, i.e., the IFFT matrix $\mathbf{W}$. There is another
interesting scenario where this property holds true: The case where all the
channels are time-varying flat fading and the constraints are on the
transmit power and on the maximum power that can be emitted over some
specified time intervals (this is the dual version of the spectral mask
constraint). In such a case, all channel matrices are diagonal and then
it is trivial to see that they have a common diagonalizing matrix, i.e., the
identity matrix. Applying duality arguments to Theorem 1, the optimal
transmission strategy for each user is a sort of TDMA over a frame of $N$
time slots, where each user optimizes the power allocation across the $N$
time slots (possibly sharing time slots with the other users). Clearly, as
opposed to the case considered in Theorem 1, in the time-selective case, the
transmitter needs to have a non-causal knowledge of the channel variation.
In practice, this kind of knowledge would require some sort of channel
prediction.

\medskip

According to Theorem \ref{Theo_multiccarrier_Info_rate_}, instead of
considering the matrix-valued games ${%
%TCIMACRO{\TeXButton{G}{{\mathscr{G}}}}%
%BeginExpansion
{\mathscr{G}}%
%EndExpansion
}_{1}$ and ${%
%TCIMACRO{\TeXButton{G}{{\mathscr{G}}}}%
%BeginExpansion
{\mathscr{G}}%
%EndExpansion
}_{2},$ we may focus on the simpler vector game ${%
%TCIMACRO{\TeXButton{G}{{\mathscr{G}}}}%
%BeginExpansion
{\mathscr{G}}%
%EndExpansion
}$, with no performance loss. It is straightforward to see that a NE of both
matrix-valued games exists if the solution set of ${%
%TCIMACRO{\TeXButton{G}{{\mathscr{G}}}}%
%BeginExpansion
{\mathscr{G}}%
%EndExpansion
}$ is non empty. Moreover, the Nash equilibria of ${%
%TCIMACRO{\TeXButton{G}{{\mathscr{G}}}}%
%BeginExpansion
{\mathscr{G}}%
%EndExpansion
}$, if they exist, must satisfy the waterfilling solution \emph{for each}
user, i.e., the following system of \emph{nonlinear} equations:
\begin{equation}
\begin{array}{c}
\mathbf{p}_{q}^{\star }=%
%TCIMACRO{\TeXButton{WaterFill}{\mathsf{WF}}}%
%BeginExpansion
\mathsf{WF}%
%EndExpansion
_{q}\left( \mathbf{p}_{1}^{\star },\ldots ,\mathbf{p}_{q-1}^{\star },\mathbf{%
p}_{q+1}^{\star },\ldots ,\mathbf{p}_{Q}^{\star }\right) =%
%TCIMACRO{\TeXButton{WaterFill}{\mathsf{WF}}}%
%BeginExpansion
\mathsf{WF}%
%EndExpansion
_{q}(\mathbf{p}_{-q}^{\star })%
\end{array}%
,\quad \forall q\in \Omega ,  \label{sym_WF-sistem}
\end{equation}%
with the waterfilling operator $%
%TCIMACRO{\TeXButton{WaterFill}{\mathsf{WF}}}%
%BeginExpansion
\mathsf{WF}%
%EndExpansion
_{q}\left( \mathbf{\cdot }\right) $ defined as
\begin{equation}
\left[
%TCIMACRO{\TeXButton{WaterFill}{\mathsf{WF}}}%
%BeginExpansion
\mathsf{WF}%
%EndExpansion
_{q}\left( \mathbf{p}_{-q}\right) \right] _{k}%
%TCIMACRO{\TeXButton{def}{\triangleq}}%
%BeginExpansion
\triangleq%
%EndExpansion
\left[ \mu _{q}-\Gamma_q\,\dfrac{1+\sum_{\,r\neq q}\left\vert
H_{rq}(k)\right\vert ^{2}p_{r}(k)}{\left\vert H_{qq}(k)\right\vert ^{2}}%
\right] _{0}^{p_{q}^{\max }(k)},\quad k\in \{1,\ldots ,N\},  \label{WF_mask}
\end{equation}%
where $\left[ x\right] _{a}^{b}$ denotes the Euclidean projection of $x$
onto the interval $[a,b]$\footnote{%
The Euclidean projection $\left[ x\right] _{a}^{b}$ \ is defined as follows:
$\left[ x\right] _{a}^{b}=a$, if $x\leq a$, $\left[ x\right] _{a}^{b}=x$, if
$a<x<b$, and $\left[ x\right] _{a}^{b}=b$, if $x\geq b$.} and the water-level $%
\mu _{q}$ is chosen to satisfy the power constraint $(1/N)%
\sum_{k=1}^{N}p_{q}^{\star }(k)=1.$

Given the nonlinear system of equations (\ref{sym_WF-sistem}), the
fundamental questions are: i)\ \emph{Does a solution exist}? ii)\emph{\ If a
solution exists, is it unique}? iii) \emph{How can such a solution be
reached in a distributed way}?

The answer to the first two questions is given in the forthcoming sections,
whereas the study of distributed algorithms is addressed in Part II of this
paper \cite{Scutari-Part II}.

\section{Existence and Uniqueness of NE\label{Sec:Existence-Uniqueness}}

Before providing the conditions for the uniqueness of the NE of game ${%
%TCIMACRO{\TeXButton{G}{{\mathscr{G}}}}%
%BeginExpansion
{\mathscr{G}}%
%EndExpansion
,}$ we introduce the following intermediate definitions. Given game ${%
%TCIMACRO{\TeXButton{G}{{\mathscr{G}}}}%
%BeginExpansion
{\mathscr{G}}%
%EndExpansion
,}$ define $\mathbf{H}(k)\in
%TCIMACRO{\U{211d} }%
%BeginExpansion
\mathbb{R}
%EndExpansion
^{Q\times Q}$ as

\begin{equation}
\left[ \mathbf{H}(k)\right] _{qr}%
%TCIMACRO{\TeXButton{def}{\triangleq}}%
%BeginExpansion
\triangleq%
%EndExpansion
\left\{
\begin{array}{ll}
\Gamma_q\,\dfrac{|\bar{H}_{rq}(k)|^{2}}{|\bar{H}_{qq}(k)|^{2}}\dfrac{%
d_{qq}^{\alpha }}{d_{rq}^{\alpha }}\dfrac{P_{r}}{P_{q}}, & \text{if }\ k\in
\mathcal{D}_{q}\cap \mathcal{D}_{r}\text{ and }r\neq q, \\
0, & \text{otherwise,}%
\end{array}%
\right.  \label{def:H_matrix}
\end{equation}%
where $\mathcal{D}_{q}$ denotes the set $\{1,\ldots ,N\}$ (possibly)
deprived of the carrier indices that user $q$ would never use as the best
response set to any strategy used by the other users, for the given set of
transmit power and propagation channels:%
\begin{equation}
\mathcal{D}_{q}%
%TCIMACRO{\TeXButton{def}{\triangleq}}%
%BeginExpansion
\triangleq%
%EndExpansion
\left\{ k\in \{1,\ldots ,N\}:\exists \text{ }\mathbf{p}_{-q}\in {%
%TCIMACRO{\TeXButton{P}{{\mathscr{P}}}}%
%BeginExpansion
{\mathscr{P}}%
%EndExpansion
}_{-q}\text{ such that }\left[ {\mathsf{WF}}_{q}\left( \mathbf{p}%
_{-q}\right) \right] _{k}\neq 0\right\} ,  \label{D_q}
\end{equation}%
with ${\mathsf{WF}}_{q}\left( \mathbf{\cdot }\right) $ defined in (\ref%
{WF_mask}) and ${%
%TCIMACRO{\TeXButton{P}{{\mathscr{P}}}}%
%BeginExpansion
{\mathscr{P}}%
%EndExpansion
}_{-q}%
%TCIMACRO{\TeXButton{def}{\triangleq}}%
%BeginExpansion
\triangleq%
%EndExpansion
{%
%TCIMACRO{\TeXButton{P}{{\mathscr{P}}}}%
%BeginExpansion
{\mathscr{P}}%
%EndExpansion
}_{1}\times \cdots \times {%
%TCIMACRO{\TeXButton{P}{{\mathscr{P}}}}%
%BeginExpansion
{\mathscr{P}}%
%EndExpansion
}_{q-1}\times {%
%TCIMACRO{\TeXButton{P}{{\mathscr{P}}}}%
%BeginExpansion
{\mathscr{P}}%
%EndExpansion
}_{q+1}\times \cdots \times {%
%TCIMACRO{\TeXButton{P}{{\mathscr{P}}}}%
%BeginExpansion
{\mathscr{P}}%
%EndExpansion
}_{Q}$.

The study of game ${%
%TCIMACRO{\TeXButton{G}{{\mathscr{G}}}}%
%BeginExpansion
{\mathscr{G}}%
%EndExpansion
}$ is addressed in the following theorem.

\begin{theorem}
\label{th:existence_uniqueness_NE}Game ${%
%TCIMACRO{\TeXButton{G}{{\mathscr{G}}}}%
%BeginExpansion
{\mathscr{G}}%
%EndExpansion
}${\ admits a nonempty solution set for any set of channels, spectral
mask constraints and transmit power of the users. Furthermore, the NE is
unique if}\vspace{-0.3cm}%
\begin{equation}
\rho \left( \mathbf{H}(k)\right) <1,\qquad \forall k\in \{1,\ldots ,N\},
\tag{C1}  \label{SF}
\end{equation}%
where $\mathbf{H}(k)${\ is defined in (\ref{def:H_matrix}) and }$\rho \left(
\mathbf{H}(k)\right) ${\ denotes the spectral radius\footnote{%
The spectral radius $\rho \left( \mathbf{A}\right) $ of the matrix $\mathbf{A%
}$ is defined as $\rho \left( \mathbf{A}\right)
%TCIMACRO{\TeXButton{def}{\triangleq}}%
%BeginExpansion
\triangleq%
%EndExpansion
\max \{|\lambda |:\lambda \in \sigma (\mathbf{A})\},$ with $\sigma (\mathbf{A%
})$ denoting the spectrum of $\mathbf{A}$ \cite{Horn85}.} of }$\mathbf{H}%
(k). $
\end{theorem}

\begin{proof}
See Appendix \ref{proof of th Existence_uniqueness_NE}.
\end{proof}

\noindent

We provide now alternative sufficient conditions for Theorem \ref%
{th:existence_uniqueness_NE}. To this end, we first introduce the matrix $%
\mathbf{H}^{\max }\in
%TCIMACRO{\U{211d} }%
%BeginExpansion
\mathbb{R}
%EndExpansion
^{Q\times Q}$, defined as
\begin{equation}
\left[ \mathbf{H}^{\max }\right] _{qr}%
%TCIMACRO{\TeXButton{def}{\triangleq}}%
%BeginExpansion
\triangleq%
%EndExpansion
\left\{
\begin{array}{ll}
\Gamma_q\,\max\limits_{k\in \mathcal{D}_{q}\cap \mathcal{D}_{r}}\dfrac{|\bar{%
H}_{rq}(k)|^{2}}{|\bar{H}_{qq}(k)|^{2}}\dfrac{d_{qq}^{\alpha }}{%
d_{rq}^{\alpha }}\dfrac{P_{r}}{P_{q}}, & \text{if }\ r\neq q, \\
0, & \text{otherwise,}%
\end{array}%
\right.  \label{H_max}
\end{equation}%
with the convention that the maximum in (\ref{H_max}) is zero if $\ \mathcal{%
D}_{q}\cap \mathcal{D}_{r}$ is empty. Then, we have the following corollary
of Theorem \ref{th:existence_uniqueness_NE}.

\begin{corollary}
\label{Corollary:SF_Uniqueness_H_max}A sufficient condition for (\ref{SF})
is:\vspace{-0.2cm}
\begin{equation}
\rho \left( \mathbf{H}^{\max }\right) <1,  \tag{C2}  \label{SF_H_max}
\end{equation}%
\vspace{-0.3cm} where $\mathbf{H}^{\max }$ is defined in (\ref{H_max}).
\end{corollary}

To give additional insight into the physical interpretation of the
conditions for the uniqueness of the NE, we introduce the following
corollary.

\begin{corollary}
\label{Corollary:SF_Uniqueness_DD}A sufficient condition for (\ref{SF}) is
given by one of the two following set of conditions:%
\begin{equation}
\hspace{-0.1cm}\dfrac{\Gamma_q}{w_{q}}\text{ }\!\!\dsum\limits_{r=1,r\neq
q}w_{r}\left[ \mathbf{H}(k)\right] _{qr}<1,\text{ }\forall q\in \Omega ,%
\hspace{-0.15cm}\quad \text{and\quad }\forall k\in \{1,\ldots ,N\},  \tag{C3}
\label{SF_H_a}
\end{equation}%
\begin{equation}
\hspace{-0.1cm}\dfrac{1}{w_{r}}\!\!\text{ }\dsum\limits_{q=1,q\neq
r}\Gamma_q\,w_{q}\left[ \mathbf{H}(k)\right] _{rq}<1,\text{ }\forall r\in
\Omega ,\hspace{-0.15cm}\hspace{-0.15cm}\quad \text{and}\quad \forall k\in
\{1,\ldots ,N\},  \tag{C4}  \label{SF_H_b}
\end{equation}%
where $\mathbf{w}\triangleq \lbrack w_{1},\ldots ,w_{Q}]^{T}$ is any
positive vector and $\mathbf{H}(k)$ is defined in (\ref{def:H_matrix}).
\end{corollary}

Note that, as a by-product of the proof of Theorem \ref%
{th:existence_uniqueness_NE}, one can always choose $\mathcal{D}%
_{q}=\{1,\ldots ,N\}$ in (\ref{SF})-(\ref{SF_H_b}), i.e., without excluding
any subcarrier. However, less stringent conditions are obtained by removing the
unnecessary carriers, i.e., those carriers that, for a given power budget
and interference levels, are never going to be used.\bigskip

\noindent \textbf{Remark 2 $-$ Physical interpretation of uniqueness
conditions. }As expected, the uniqueness of NE is ensured if the links are
sufficiently far apart from each other. In fact, from (\ref{SF_H_a})-(\ref%
{SF_H_b}) for example, one infers that there exists a minimum distance
beyond which the uniqueness of NE is guaranteed, corresponding to the
maximum level of interference that may be tolerated by the users.
Specifically, condition (\ref{SF_H_a}) imposes a constraint on the maximum
amount of interference that each receiver can tolerate; whereas (\ref{SF_H_b}%
) introduces an upper bound on the maximum level of interference that each
transmitter is allowed to generate. This result agrees with the fact that,
as the MUI becomes negligible, the rates $R_{q}$ in (\ref{Rate}) become
decoupled and then the rate-maximization problem in (\ref{Rate Game}) for
each user admits a unique solution. But, the most interesting result coming
from conditions (\ref{SF})-(\ref{SF_H_b}) is that the uniqueness of the
equilibrium is robust against the worst normalized channels $%
|H_{rq}(k)|^{2}/ $ $|H_{qq}(k)|^{2};$ in fact, the subchannels corresponding
to the highest ratios $|H_{rq}(k)|^{2}/|H_{qq}(k)|^{2}$ (and, in particular,
the subchannels where $|H_{qq}(k)|^{2}$ is vanishing) do not necessarily
affect the uniqueness condition, as their carrier indices may not belong to the set $%
\mathcal{D}_{q}$.\bigskip

\noindent \textbf{Remark 3 $-$ Uniqueness condition and distributed algorithms.
}Interestingly, condition (\ref{SF_H_max}), in addition to guarantee the
uniqueness of the NE, is also responsible for the convergence of both
simultaneous and sequential iterative waterfilling algorithms, proposed in
Part II of the paper \cite{Scutari-Part II}.

\bigskip

\noindent \textbf{Remark 4 $-$ Comparison with previous conditions. }Theorem %
\ref{th:existence_uniqueness_NE} unifies and generalizes many existence and
uniqueness results obtained in the literature \cite{Yu}-\cite%
{Scutari-Barbarossa-SPAWC03}, \cite{Tse} for the special cases of game ${%
%TCIMACRO{\TeXButton{G}{{\mathscr{G}}}}%
%BeginExpansion
{\mathscr{G}}%
%EndExpansion
}$ in (\ref{Rate Game}). Specifically, in \cite{Yu}-\cite%
{Scutari-Barbarossa-SPAWC03} a game as in (\ref{Rate Game}) is studied,
where all the players are assumed to have the same power budget and no
spectral mask constraints are considered [i.e., $p_{q}^{\max }(k)=+\infty
,\forall k,q$]. In \cite{Tse} instead, the channel is assumed to be flat
over the whole bandwidth. Interestingly, the conditions obtained in \cite{Yu}-%
\cite{Scutari-Barbarossa-SPAWC03}, \cite{Tse} are more restrictive than (\ref%
{SF})-(\ref{SF_H_b}), as shown in the following corollary of Theorem \ref%
{th:existence_uniqueness_NE}.\footnote{%
We summarize the main results of \cite{Yu}-\cite{Scutari-Barbarossa-SPAWC03}
using our notation to facilitate the comparison.}

\begin{corollary}
\label{Corollary:Conditions of the others}Sufficient conditions for (\ref%
{SF_H_a}) are \cite{Yu, ChungISIT03, Scutari-Barbarossa-SPAWC03}%
\begin{equation}
\Gamma _{q}\max\limits_{k\in \{1,\ldots ,N\}}\left\{ \dfrac{|\bar{H}%
_{rq}(k)|^{2}}{|\bar{H}_{qq}(k)|^{2}}\right\} \dfrac{d_{qq}^{\alpha }}{%
d_{rq}^{\alpha }}\dfrac{P_{r}}{P_{q}}<\dfrac{1}{Q-1},\hspace{1.5cm}\forall
\text{ }r,q\neq r\in \Omega ,  \tag{C5}  \label{Chung SF_}
\end{equation}%
or \cite{Yamashitay-Luo}%
\begin{equation}
\Gamma _{q}\max\limits_{k\in \{1,\ldots ,N\}}\left\{ \dfrac{|\bar{H}%
_{rq}(k)|^{2}}{|\bar{H}_{qq}(k)|^{2}}\right\} \dfrac{d_{qq}^{\alpha }}{%
d_{rq}^{\alpha }}\dfrac{P_{r}}{P_{q}}<\dfrac{1}{2Q-3},\hspace{1.5cm}\forall
\text{ }r,q\neq r\in \Omega .  \tag{C6}  \label{Luo-Yamashitay}
\end{equation}

In the case of flat-fading channels (i.e., $\bar{H}_{rq}(k)=\bar{H}_{rq}$, $%
\forall r,q)$, condition (\ref{SF_H_a}) becomes \cite{Tse}%
\begin{equation}
\Gamma _{q}\sum_{r=1,r\neq q}^{Q}\dfrac{|\bar{H}_{rq}|^{2}}{|\bar{H}%
_{qq}|^{2}}\dfrac{d_{qq}^{\alpha }}{d_{rq}^{\alpha }}\dfrac{P_{r}}{P_{q}}%
<1,\qquad \forall q\in \Omega .  \label{Tse-cond}
\end{equation}
\end{corollary}

Recently, alternative sufficient conditions for the uniqueness of the NE of
game ${%
%TCIMACRO{\TeXButton{G}{{\mathscr{G}}}}%
%BeginExpansion
{\mathscr{G}}%
%EndExpansion
}$ were given in \cite{Luo-Pang}.\footnote{%
We thank Prof. Facchinei, who kindly brought to our attention reference \cite%
{Luo-Pang}, after this paper was completed.} Among all, an easy condition to
be checked is the following:
\begin{equation}
\mathbf{I+H}(k)\text{ is positive definite for all }k\in \{1,\ldots ,N\},
\tag{C7}  \label{Luo-Cond}
\end{equation}%
where $\mathbf{H}(k)$ is defined as in (\ref{def:H_matrix}), with each $%
\mathcal{D}_{q}=\{1,\ldots ,N\}.$

All the conditions above depend on the channel realizations $\left\{ \bar{H}%
_{rq}(k)\right\} $ and on the network topology through the distances $%
\left\{ d_{rq}\right\} .$ Hence, there is a nonzero probability that they
are not satisfied for a given set of channel realizations, drawn from a
given probability space. In order to compare the goodness of the above
conditions, we tested them over a set of channel impulse responses generated
as vectors composed of i.i.d. complex Gaussian random variables with zero
mean and unit variance. We plot in Figure \ref{Fig:Comparison_SF} the
probability that conditions (\ref{SF}), (\ref{Chung SF_}) and (\ref{Luo-Cond}%
) are satisfied versus the ratio $d_{rq}/d_{qq}$, i.e., the normalized
interlink distance. For the sake of simplicity, we assumed $d_{rq}=d_{qr},$ $%
P_{q}=P_{r}$ and $\Gamma _{q}=1,$ $\forall q,r\in \Omega .$ We considered $%
Q=5$ [Figure \ref{Fig:Comparison_SF}(a)] and $Q=15$ [Figure \ref%
{Fig:Comparison_SF}(b)] active links. We tested our condition considering in
(\ref{SF}) a set $\mathcal{D}_{q}$ obtained using the following worst case
scenario. For each user $q$, we build the worst possible interferer as the
virtual node (denoted by $v$) that has a power budget equal to the sum of
the transmit powers of all the other users (i.e., $\sum_{r\neq q}P_{r}$) and
channel between its own transmitter and receiver $q$ as the highest channel
among all the interference channels with respect to receiver $q,$ i.e., $%
|H_{vq}(k)|^{2}=\max_{r\neq q}|H_{rq}(k)|^{2}.$ We build a set that includes the set $\mathcal{D}_{q}$ defined in (\ref{D_q}) using the following iterative procedure: For each subcarrier
$k,$ the virtual user distributes its own power ($\sum_{r\neq q}P_{r}$)
across the whole spectrum in order to facilitate user $q$ to use the
subcarrier $k,$ as much as possible. If, even under these circumstances,
user $q$ is not going to use subcarrier $k,$ because of its own power budget
$P_{q}$ and $|H_{qq}(k)|^{2},$ then we are sure that index $k$ can not
possibly belong to $\mathcal{D}_{q}.$

\begin{figure}[tbp]
\centering\vspace{-0.7cm}
\subfigure[\,] {\includegraphics[height=6cm,
width=7.6cm]{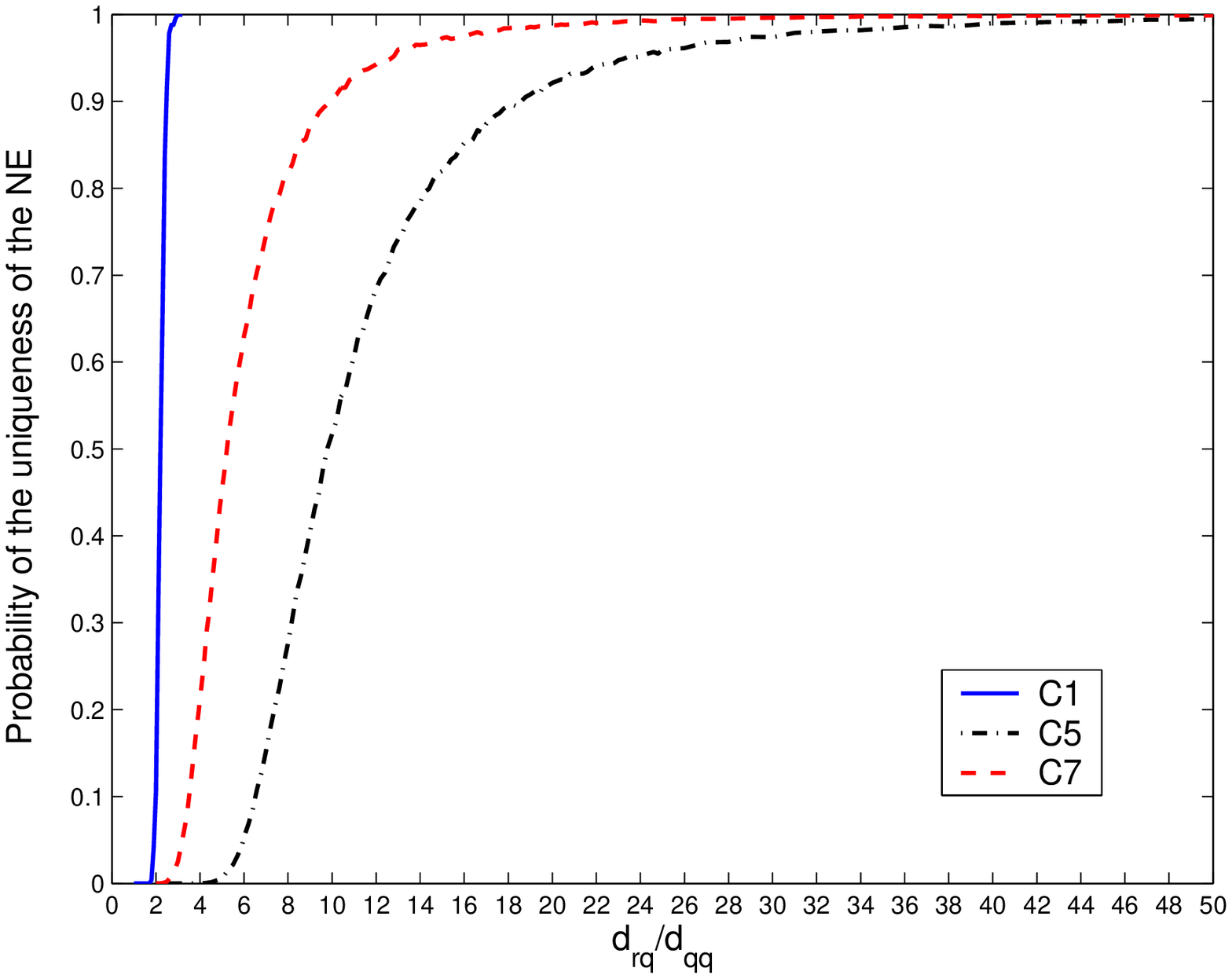}
\hfill}\hspace{0.2cm}
\subfigure[]{\includegraphics[height=6cm,
width=7.9cm]{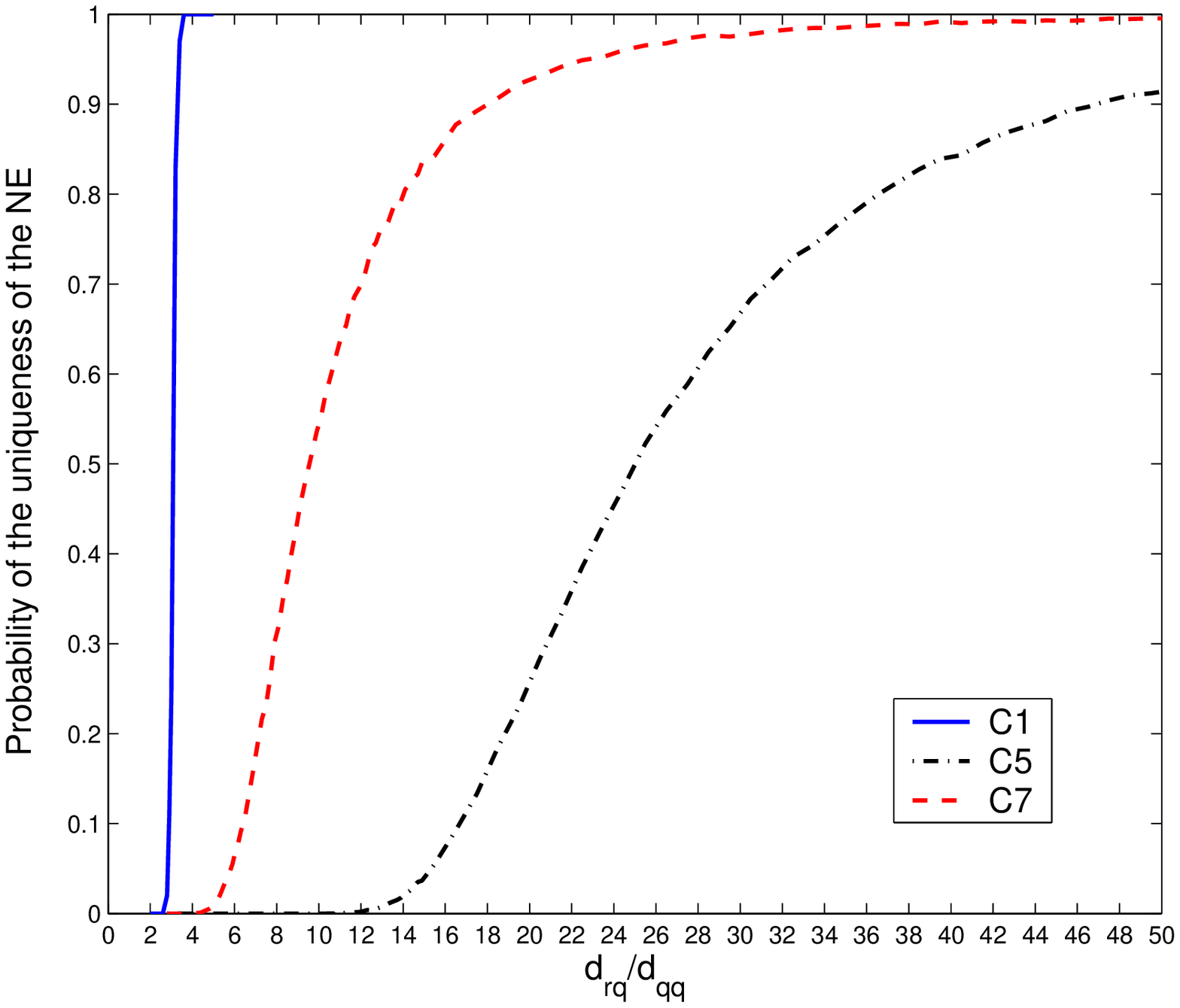}
\hfill}%\vspace{-0.4cm}
\vspace{-0.2cm}
\caption{{\protect\small Probability of (C1), (C5) and (C7) versus $%
d_{rq}/d_{qq}$; $Q=5$ [subplot (a)], $Q=15$ [subplot (b)], $\protect\gamma %
=2.5,$ $\Gamma _{q}=1,$ $d_{rq}/d_{qq}=d_{qr}/d_{rr},$ $\forall r,q\in
\Omega ,$ }$N=64.${\protect\small \ }}
\label{Fig:Comparison_SF}
\end{figure}
%\begin{figure}[tbph]
%\vspace{-1.0cm}
%\par
%\begin{center}
%\includegraphics[trim=0.000000in 0.000000in 0.000000in
%-0.212435in,height=7cm,
%width=9cm]{Figs/probability_NE_4.eps} \label{fig:OutageProbability}\vspace{%
%-1.0cm}
%\end{center}
%\par
%\vspace{-0.6cm}
%\caption{{\protect\footnotesize Probability of (C1), (C2) and (C3) versus $%
%d_{rq}/d_{qq}=d_{qr}/d_{rr}$;} {\protect\footnotesize $\protect\alpha =2.5,$
%$d_{rq}=d_{qr},$ $\Gamma _{1}=\Gamma _{2}=1,$ $\forall r,q.$}}
%\label{Fig:Comparison_SF}
%\end{figure}

We can see, from Figure \ref{Fig:Comparison_SF}, that the probability that
the NE is unique increases as the links become more and more separated of
each other (i.e., the ratio $d_{rq}/d_{qq}$ increases). Furthermore, we can
verify that, even having not considered the smallest possible set $\mathcal{D%
}_{q}$, as defined in (\ref{D_q}), our condition (\ref{SF}) has a much higher
probability of being satisfied than (\ref{Chung SF_}) and (\ref{Luo-Cond}).
The main difference between our condition (\ref{SF}) and (\ref{Chung SF_}), (%
\ref{Luo-Cond}) is that (\ref{SF}) exhibits a neat threshold behavior since
it transits very rapidly from the non-uniqueness guarantee to the almost
certain uniqueness, as the inter-user distance ratio $d_{rq}/d_{qq}$
increases by a small amount. This shows that the uniqueness condition (%
\ref{SF}) depends, ultimately, on the interlink distance rather than on
the channel realization. This represents the fundamental difference between
our uniqueness condition and those given in the literature. As an example,
for a system with $Q=5$ links and probability $0.99$ of guaranteeing the
uniqueness of the NE, condition (\ref{SF}) requires $d_{rq}/d_{qq}\simeq 2$
whereas conditions (\ref{Chung SF_}) and (\ref{Luo-Cond}) require $%
d_{rq}/d_{qq}>40$ and $d_{rq}/d_{qq}\simeq 24$, respectively. Furthermore,
this difference increases as the number $Q$ of links increases. \vspace{%
-0.5cm}

\section{Physical Interpretation of NE\label{Sec: OFDA/CDMA interpretation}}

\vspace{-0.2cm} In this section we provide a physical interpretation of the
optimal power allocation corresponding to the NE in the limiting cases of
low and high MUI.\footnote{%
For the sake of notation, in this section we consider only the case in which
$p_{q}^{\max }(k)=+\infty ,$ $\forall q,\forall k,$ but it is
straightforward to see that our derivations can be easily generalized to the
case of spectral mask constraints.} To quantify what low and high
interference mean, we introduce the SNR of link $q$ (denoted by $%
%TCIMACRO{\TeXButton{snr}{\mathsf{snr}}}%
%BeginExpansion
\mathsf{snr}%
%EndExpansion
_{q}$) and the Interference-to-Noise Ratio due to the interference received
by destination $q$ and generated by source $r,$ with $r\neq q$ (denoted by $%
%TCIMACRO{\TeXButton{inr}{\mathsf{inr}}}%
%BeginExpansion
\mathsf{inr}%
%EndExpansion
_{rq}$), defined as $%
%TCIMACRO{\TeXButton{snr}{\mathsf{snr}}}%
%BeginExpansion
\mathsf{snr}%
%EndExpansion
_{q}%
%TCIMACRO{\TeXButton{def}{\triangleq}}%
%BeginExpansion
\triangleq%
%EndExpansion
P_{q}/(\sigma _{q}^{2}d_{qq}^{\alpha })$ and $%
%TCIMACRO{\TeXButton{inr}{\mathsf{inr}}}%
%BeginExpansion
\mathsf{inr}%
%EndExpansion
_{rq}%
%TCIMACRO{\TeXButton{def}{\triangleq}}%
%BeginExpansion
\triangleq%
%EndExpansion
P_{r}/(\sigma _{q}^{2}d_{rq}^{\alpha }).$ Using $%
%TCIMACRO{\TeXButton{snr}{\mathsf{snr}}}%
%BeginExpansion
\mathsf{snr}%
%EndExpansion
_{q}$ and $%
%TCIMACRO{\TeXButton{inr}{\mathsf{inr}}}%
%BeginExpansion
\mathsf{inr}%
%EndExpansion
_{rq},$ the SINR $%
%TCIMACRO{\TeXButton{sinr}{\mathsf{sinr}}}%
%BeginExpansion
\mathsf{sinr}%
%EndExpansion
_{q}(k)$ in (\ref{SINR}) can be rewritten as
\begin{equation}
%TCIMACRO{\TeXButton{sinr}{\mathsf{sinr}}}%
%BeginExpansion
\mathsf{sinr}%
%EndExpansion
_{q}(k)=\frac{%
%TCIMACRO{\TeXButton{snr}{\mathsf{snr}}}%
%BeginExpansion
\mathsf{snr}%
%EndExpansion
_{q}\left\vert \bar{H}_{qq}(k)\right\vert ^{2}p_{q}(k)}{1+\sum_{\,r\neq q}%
%TCIMACRO{\TeXButton{inr}{\mathsf{inr}}}%
%BeginExpansion
\mathsf{inr}%
%EndExpansion
_{rq}\left\vert \bar{H}_{rq}(k)\right\vert ^{2}p_{r}(k)}.  \label{SINR_2}
\end{equation}%
\medskip

\noindent \textbf{Low interference case.}\textit{\ } Consider the low
interference case, i.e., the situation where the interference term in the
denominator in (\ref{SINR_2}) can be neglected. A sufficient condition to
satisfy this assumption is that the links are sufficiently far apart from
each other, i.e., $d_{rq}>>d_{qq}$ $\forall r\neq q$,$r,q\in \Omega $. For
sufficiently small $%
%TCIMACRO{\TeXButton{inr}{\mathsf{inr}}}%
%BeginExpansion
\mathsf{inr}%
%EndExpansion
_{rq}$ and sufficiently large $%
%TCIMACRO{\TeXButton{snr}{\mathsf{snr}}}%
%BeginExpansion
\mathsf{snr}%
%EndExpansion
_{q}$, condition (\ref{SF}) is satisfied and, hence, by Theorem 1, the NE is
unique. Also, by inspection of the waterfilling solution in (\ref%
{sym_WF-sistem}), it is clear that under those conditions, $p_{q}(k)>0$ for
all $q\in \Omega $ and $k\in \{1,\ldots ,N\}$. This means that each source
uses the whole bandwidth. Furthermore, it is well known that as the SNR
increases, the waterfilling solution tends to a flat power allocation. In
summary, we have the following result.

\begin{proposition}
\label{Proposition_CDMA Interpretation}Given game $%
%TCIMACRO{\TeXButton{G}{{\mathscr{G}}}}%
%BeginExpansion
{\mathscr{G}}%
%EndExpansion
,$ there exist sets of values $\{%
%TCIMACRO{\TeXButton{inr}{\mathsf{inr}}}%
%BeginExpansion
\mathsf{inr}%
%EndExpansion
_{rq}^{\star }\}_{r\neq q\in \Omega }$ and $\{%
%TCIMACRO{\TeXButton{snr}{\mathsf{snr}}}%
%BeginExpansion
\mathsf{snr}%
%EndExpansion
_{q}^{\star }\}_{q\in \Omega },$ with $%
%TCIMACRO{\TeXButton{inr}{\mathsf{inr}}}%
%BeginExpansion
\mathsf{inr}%
%EndExpansion
_{rq}^{\star }<<1$ and $%
%TCIMACRO{\TeXButton{snr}{\mathsf{snr}}}%
%BeginExpansion
\mathsf{snr}%
%EndExpansion
_{q}^{\star }>>1,$ such that, for all $%
%TCIMACRO{\TeXButton{inr}{\mathsf{inr}}}%
%BeginExpansion
\mathsf{inr}%
%EndExpansion
_{rq}\leq
%TCIMACRO{\TeXButton{inr}{\mathsf{inr}}}%
%BeginExpansion
\mathsf{inr}%
%EndExpansion
_{rq}^{\star }$ and $%
%TCIMACRO{\TeXButton{snr}{\mathsf{snr}}}%
%BeginExpansion
\mathsf{snr}%
%EndExpansion
_{q}\geq
%TCIMACRO{\TeXButton{snr}{\mathsf{snr}}}%
%BeginExpansion
\mathsf{snr}%
%EndExpansion
_{q}^{\star },$ the NE of $%
%TCIMACRO{\TeXButton{G}{{\mathscr{G}}}}%
%BeginExpansion
{\mathscr{G}}%
%EndExpansion
$ is \emph{unique }(cf. Theorem \ref{th:existence_uniqueness_NE}) and all
users share the \textit{whole} available bandwidth. In addition, if $%
%TCIMACRO{\TeXButton{snr}{\mathsf{snr}}}%
%BeginExpansion
\mathsf{snr}%
%EndExpansion
_{q}>>%
%TCIMACRO{\TeXButton{snr}{\mathsf{snr}}}%
%BeginExpansion
\mathsf{snr}%
%EndExpansion
_{q}^{\star },$ then the optimal power allocation of each user tends to be
flat over the whole bandwidth.
\end{proposition}

From Proposition \ref{Proposition_CDMA Interpretation}, it turns out, as it
could have been intuitively guessed, that when the interference is low, at
the (unique) NE, every user transmits over the entire available spectrum
(like a CDMA system), as in such a case nulling the interference would not
be worth of the bandwidth reduction. As a numerical example, in
Figure 2, we plot the optimal power spectral density (PSD) of a system composed of
three links, for different values of the ratio $d_{rq}/d_{qq}$. The
results shown in Figure 2 have been obtained using the distributed
algorithms described in Part II \cite{Scutari-Part II}. From Figure 2, we
can check that, as the ratio $d_{rq}/d_{qq}$ increases, the optimal PSD
tends to be flat, while satisfying the simultaneous waterfilling condition
in (\ref{sym_WF-sistem}).\bigskip

\noindent\textbf{High interference case.} When $\mathsf{inr}_{rq}>>1$ for
all $q$ and $r\neq q$, the interference is the dominant contribution in the
equivalent noise (thermal noise plus MUI) in the denominator of (\ref{SINR_2}%
). In this case, game $%
%TCIMACRO{\TeXButton{G}{{\mathscr{G}}}}%
%BeginExpansion
{\mathscr{G}}%
%EndExpansion
$ admits multiple Nash equilibria. An interesting class of these equilibria
includes the FDMA solutions (called orthogonal Nash equilibria), occurring
when the power spectra of different users are nonoverlapping.
The characterization of these equilibria is given in the following.

\begin{proposition}
\label{Prop_Subcarrier-allocation} Given game $%
%TCIMACRO{\TeXButton{G}{{\mathscr{G}}}}%
%BeginExpansion
{\mathscr{G}}%
%EndExpansion
,$ for each $q\in \Omega ,$ let $\mathcal{I}_{q}%
%TCIMACRO{\TeXButton{def}{\triangleq}}%
%BeginExpansion
\triangleq%
%EndExpansion
\left\{ k\in \{1,\ldots ,N\}:p_{q}(k)>0,\text{ }p_{r}(k)=0,\text{ }\forall
r\right. $ $\left. \neq q\in \Omega \right\} $ denote the set of subcarriers
over which only user $q$ transmits. For any given $\{%
%TCIMACRO{\TeXButton{snr}{\mathsf{snr}}}%
%BeginExpansion
\mathsf{snr}%
%EndExpansion
_{q}\}_{q\in \Omega },$ there exists $\{%
%TCIMACRO{\TeXButton{inr}{\mathsf{inr}}}%
%BeginExpansion
\mathsf{inr}%
%EndExpansion
_{rq}^{\star }\}_{r\neq q\in \Omega },$ with each $%
%TCIMACRO{\TeXButton{inr}{\mathsf{inr}}}%
%BeginExpansion
\mathsf{inr}%
%EndExpansion
_{rq}^{\star }>>1,$ such that for all $%
%TCIMACRO{\TeXButton{inr}{\mathsf{inr}}}%
%BeginExpansion
\mathsf{inr}%
%EndExpansion
_{rq}\geq
%TCIMACRO{\TeXButton{inr}{\mathsf{inr}}}%
%BeginExpansion
\mathsf{inr}%
%EndExpansion
_{rq}^{\star}$, game $%
%TCIMACRO{\TeXButton{G}{{\mathscr{G}}}}%
%BeginExpansion
{\mathscr{G}}%
%EndExpansion
$ admits multiple orthogonal Nash equilibria. If, in addition, $\{%
%TCIMACRO{\TeXButton{snr}{\mathsf{snr}}}%
%BeginExpansion
\mathsf{snr}%
%EndExpansion
_{q}\}_{q\in \Omega }$ and $\{%
%TCIMACRO{\TeXButton{inr}{\mathsf{inr}}}%
%BeginExpansion
\mathsf{inr}%
%EndExpansion
_{rq}\}_{r\neq q\in \Omega }$ are such that
\begin{equation}
\Gamma _{q}\frac{\left\vert \bar{H}_{qr}(k)\right\vert ^{2}}{\left\vert \bar{%
H}_{qq}(k)\right\vert ^{2}}\frac{%
%TCIMACRO{\TeXButton{inr}{\mathsf{inr}}}%
%BeginExpansion
\mathsf{inr}%
%EndExpansion
_{qr}}{%
%TCIMACRO{\TeXButton{snr}{\mathsf{snr}}}%
%BeginExpansion
\mathsf{snr}%
%EndExpansion
_{q}}\leq 1,\quad k\in \mathcal{I}_{q},\text{ }r\neq q\in \Omega ,
\label{medium_interference}
\end{equation}%
and an orthogonal NE still exists, the subcarriers are allocated among the
users according to%
\begin{equation}
\dfrac{\left\vert \bar{H}_{rr}(k_{r})\right\vert ^{2}}{\left\vert \bar{H}%
_{qq}(k_{r})\right\vert ^{2}}\geq \dfrac{\left\vert \bar{H}%
_{rr}(k_{q})\right\vert ^{2}}{\left\vert \bar{H}_{qq}(k_{q})\right\vert ^{2}}%
,\qquad \text{ }k_{r}\in \mathcal{I}_{r}\text{ and }k_{q}\in \mathcal{I}_{q}.
\label{Subcarrier-allocation}
\end{equation}

\begin{proof}
See Appendix \ref{Proof_Prop_Subcarrier-allocation}.
\end{proof}
\end{proposition}

The above proposition has an intuitive interpretation: When the interference
is very high, the users self-organize themselves in order to remove the
interference totally, i.e., using nonoverlapping bands. In this case, game $%
%TCIMACRO{\TeXButton{G}{{\mathscr{G}}}}%
%BeginExpansion
{\mathscr{G}}%
%EndExpansion
$ may have multiple orthogonal Nash equilibria. For example, in the simple
case of $Q=N,$ there are $Q!$ different orthogonal Nash equilibria,
corresponding to all the permutations where each transmitter uses only one
carrier. As the interference level decreases (i.e., the $%
%TCIMACRO{\TeXButton{inr}{\mathsf{inr}}}%
%BeginExpansion
\mathsf{inr}%
%EndExpansion
_{rq}$'s), the NE becomes unique and in such a case, if an orthogonal
equilibrium still exists, then the distribution of the subcarriers among the
users must satisfy the rule given by (\ref{Subcarrier-allocation}). This
strategy is similar, in principle to FDMA, but differently from standard
FDMA, here each user is getting the \textquotedblleft best" portion of the
spectrum for itself. Interestingly, (\ref{Subcarrier-allocation}) is the
generalization of the condition satisfied by the subcarrier allocation in
the multiple access frequency-selective channel, where the optimization
problem is the sum-rate maximization under a transmit power constraint \cite%
{Cheng-Verdu'}.
In Figure \ref{OFDMA-PSD}, we show a numerical example of the optimal power
allocation at NE, for a system with two active links, in the case of high interference.
\begin{figure}[tbp]
\vspace{-0.9cm} \centering
\subfigure[Optimal PSDs at NE in the low-interference
case.] {\includegraphics[height=6.3cm, width=7.8cm]{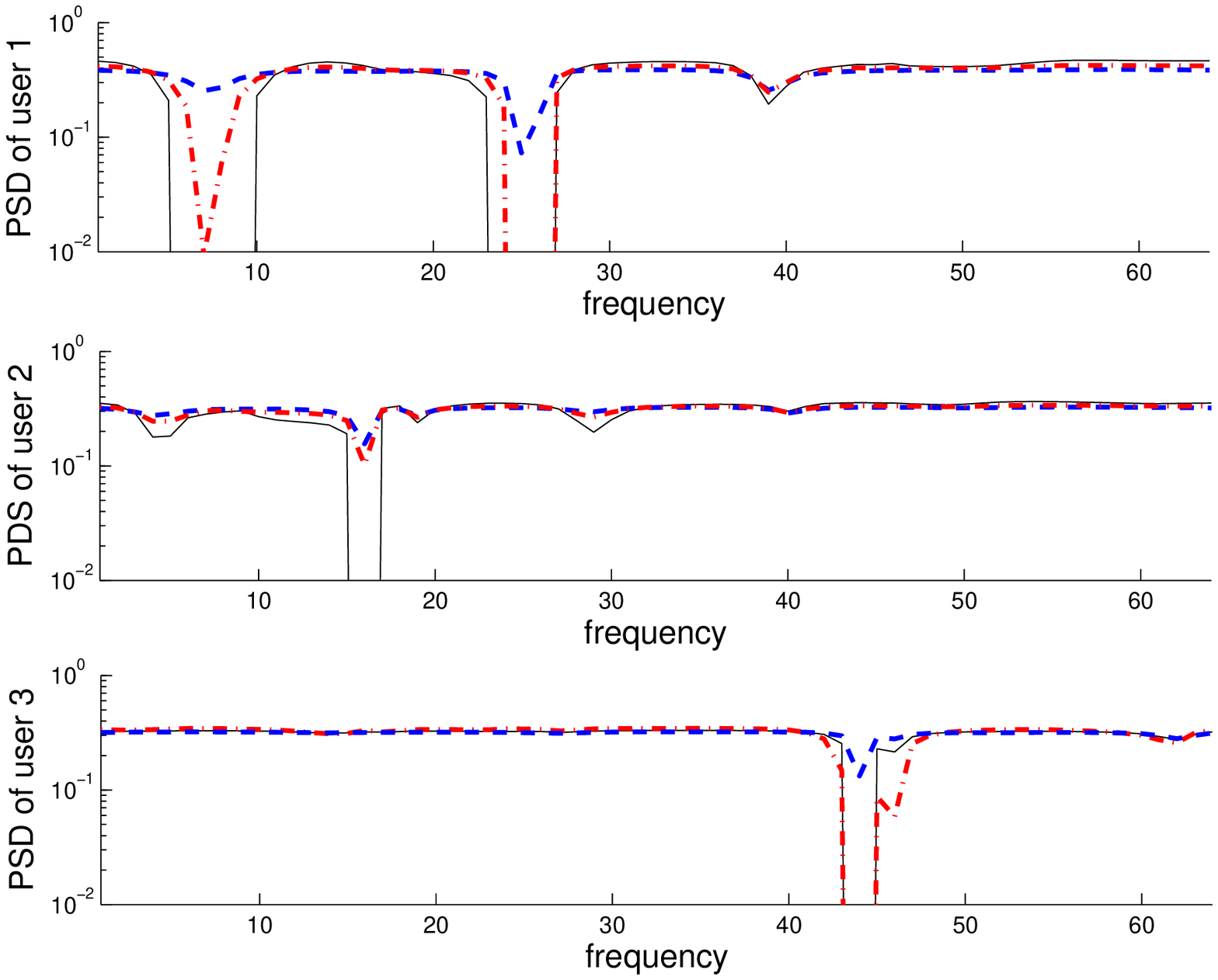}
\label{low-interference}\hspace{0.3cm}}
\subfigure[Optimal PSDs at NE in the
high-interference case.]{\includegraphics[height=6.3cm,
width=7.6cm]{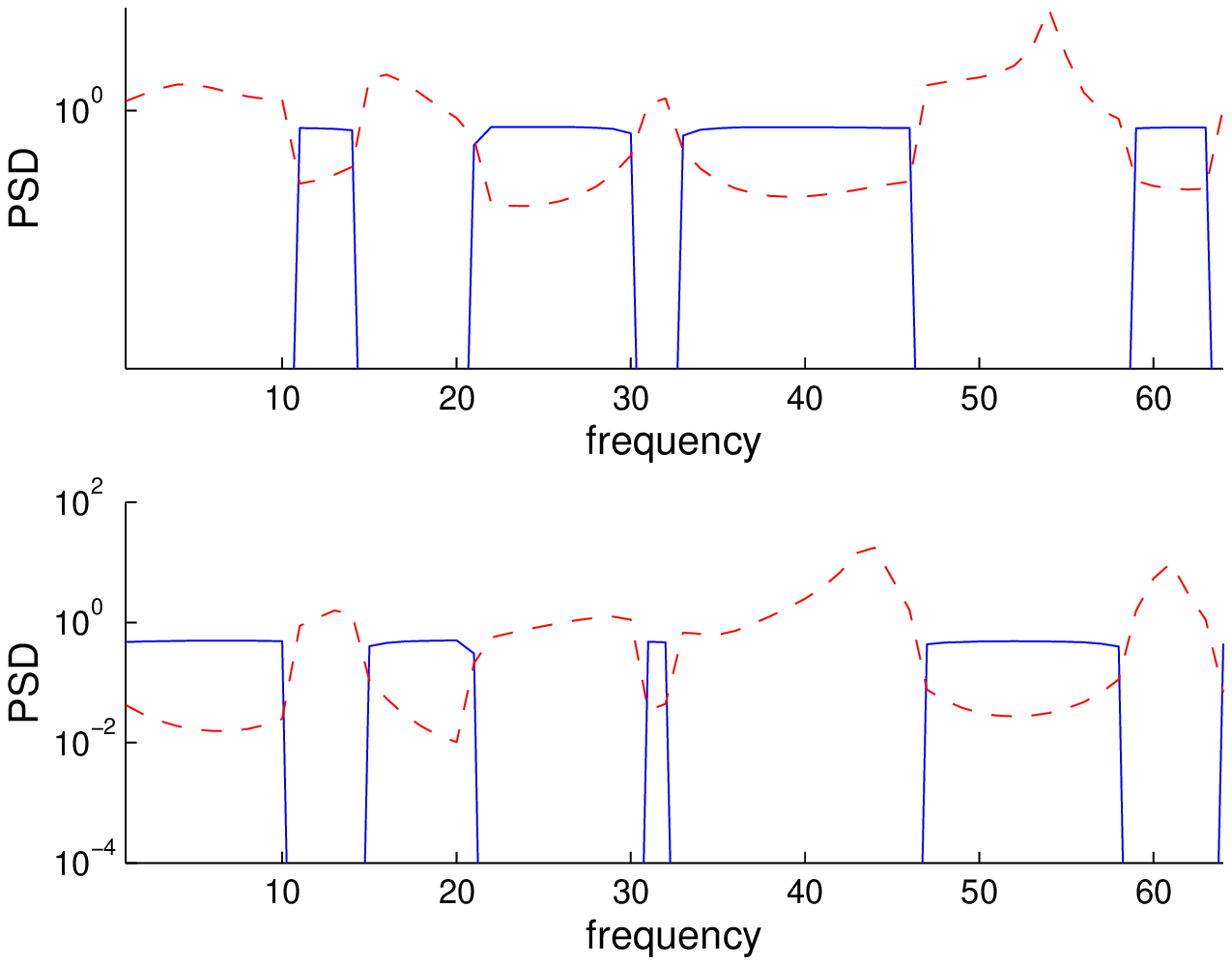} \label{OFDMA-PSD}\hfill} \vspace{-0.1cm}
\caption{{\protect\footnotesize Optimal PSDs at NE in the low-interference
and high-interference cases: a) Solid, dashed, and dashed-dot line curves
refer on the PSD obtained for $d_{rq}/d_{qq}=5,8$,$12$, respectively, and $%
\mathsf{snr}_{q}=15$ dB; b) Solid and dashed lines refer to PSD of each link
and PSD of the MUI plus thermal noise, normalized by the channel transfer
function square modulus of the link, respectively; $%
d_{12}/d_{22}=d_{21}/d_{11}=1,\Gamma _{1}=\Gamma _{2}=1,$ $\mathsf{snr}_{1}=%
\mathsf{snr}_{2}=5$ dB.}}
\label{PSD-OFDMA-shared}
\end{figure}
From Figure \ref{OFDMA-PSD} we observe that,
as predicted by Proposition \ref{Prop_Subcarrier-allocation}, different
users tend to transmit over non-overlapping bands.

In general, with intermediate interference levels, the optimal solution
consists in allowing partial superposition of the PSDs of the users. An
example of optimal PSD distribution for an intermediate level of
interference is plotted in Figure \ref{low-interference} (see the curves
obtained with $d_{rq}/d_{qq}=5$). %\begin{figure}[tbph]

\vspace{-0.5cm}

\section{How good is a Nash Equilibrium ?}

\label{Sec:How_good_NE} The optimality criterion used in this paper, based
on the achievement of a NE, is useful to derive decentralized strategies.
However, as the Nash equilibria need not to be Pareto-efficient \cite{Dubey}%
, this criterion, in general, does not provide any a priori guarantee about
the goodness of the equilibrium. Even when the equilibrium is unique, it
is important to know how far is the performance from the optimal
solution provided by a centralized strategy \cite{Cendrillon sub04,
Scutari-Barbarossa-ICASSP}. The scope of this section is thus to quantify
the performance loss resulting from the use of a decentralized approach,
with respect to the optimal centralized solution. To this end, we compare
the rates of the users corresponding to the Nash equilibria of game ${%
%TCIMACRO{\TeXButton{G}{{\mathscr{G}}}}%
%BeginExpansion
{\mathscr{G}}%
%EndExpansion
}$ in (\ref{Rate Game}) with the rates lying on the boundary of the largest
region of the achievable rates obtained as the Pareto-optimal trade-off surface
solving the multi-objective optimization based on (\ref{Rate}).

The rate region is defined as the set $\mathcal{R}$ of all feasible data
rate combinations $\mathbf{r}%
%TCIMACRO{\TeXButton{def}{\triangleq}}%
%BeginExpansion
\triangleq%
%EndExpansion
[r_{1},\ldots ,r_{Q}]^{T}\in \mathcal{%
%TCIMACRO{\U{211d} }%
%BeginExpansion
\mathbb{R}
%EndExpansion
}_{+}^{Q}$ among the active links. A given vector of rates $\mathbf{r}\in
\mathcal{%
%TCIMACRO{\U{211d} }%
%BeginExpansion
\mathbb{R}
%EndExpansion
}_{+}^{Q}$ is said to be \textit{feasible} if it is possible to transfer
information in the network at these rates, \textit{reliably}, for \emph{some}
power allocation $\mathbf{p}%
%TCIMACRO{\TeXButton{def}{\triangleq}}%
%BeginExpansion
\triangleq%
%EndExpansion
(\mathbf{p}_{q})_{q\in \Omega }$ satisfying the power constraints\ $\mathbf{p%
}_{q}\in
%TCIMACRO{\TeXButton{P}{{\mathscr{P}}}}%
%BeginExpansion
{\mathscr{P}}%
%EndExpansion
_{q}$, $\forall q\in \Omega .$ The rate region can be numerically computed
by considering all possible power allocations $\mathbf{p}$ such that $%
\mathbf{p}_{q}\in
%TCIMACRO{\TeXButton{P}{{\mathscr{P}}}}%
%BeginExpansion
{\mathscr{P}}%
%EndExpansion
_{q}$, $\forall q\in \Omega $. Specifically, we have%
\begin{equation}
\mathcal{R\ }\mathcal{=}\underset{\ \mathbf{p}_{q}\in
%TCIMACRO{\TeXButton{P}{{\mathscr{P}}}}%
%BeginExpansion
{\mathscr{P}}%
%EndExpansion
_{q},\text{ }q\in \Omega }{\dbigcup }\left\{ \mathbf{r}\in \mathcal{%
%TCIMACRO{\U{211d} }%
%BeginExpansion
\mathbb{R}
%EndExpansion
}_{+}^{Q}\ |\ r_{q}\leq R_{q}(\mathbf{p}),\forall q\in \Omega \right\}
=\left\{ \mathbf{r}\in \mathcal{%
%TCIMACRO{\U{211d} }%
%BeginExpansion
\mathbb{R}
%EndExpansion
}_{+}^{Q}\ |\ \exists \ \mathbf{p}:\mathbf{p}_{q}\in
%TCIMACRO{\TeXButton{P}{{\mathscr{P}}}}%
%BeginExpansion
{\mathscr{P}}%
%EndExpansion
_{q},\ r_{q}\leq R_{q}(\mathbf{p}),\ \forall q\in \Omega \right\} \vspace{%
-0.1cm}  \label{RR}
\end{equation}%
where $R_{q}(\mathbf{p})$ and $%
%TCIMACRO{\TeXButton{P}{{\mathscr{P}}}}%
%BeginExpansion
{\mathscr{P}}%
%EndExpansion
_{q}$ are given in (\ref{Rate}) and (\ref{admissible strategy set}),
respectively.

In general, the characterization of the rate region is a very difficult
nonconvex optimization problem. Nevertheless, in the case of low interference,
we have the following.

\begin{proposition}
\label{ConvexRateRegion} In the case of low interference $(%
%TCIMACRO{\TeXButton{inr}{\mathsf{inr}}}%
%BeginExpansion
\mathsf{inr}%
%EndExpansion
_{rq}<<1),$ and high SNR $(%
%TCIMACRO{\TeXButton{snr}{\mathsf{snr}}}%
%BeginExpansion
\mathsf{snr}%
%EndExpansion
_{q}>>1),$ the rate region $\mathcal{R}$ given in (\ref{RR}) is a convex set.
\end{proposition}

\vspace{-0.2cm}

\begin{proof}
See Appendix \ref{Proof ConvexityRR}.
\end{proof}

The best achievable trade-off among the rates is given by the Pareto optimal
points of the set\ $\mathcal{R}$ given in (\ref{RR}). In formulas, all
highest feasible rates can be found by solving the following multi-objective
optimization problem (MOP) \cite{Miettinen}\vspace{-0.2cm}:
\begin{equation}
\begin{array}{l}
\limfunc{maximize}\limits_{\mathbf{p}}\quad \text{\ }\left\{ R_{1}(\mathbf{p}%
),\ldots ,R_{Q}(\mathbf{p})\right\} \\
\text{subject to \quad\ }\mathbf{p}_{q}\in
%TCIMACRO{\TeXButton{P}{{\mathscr{P}}}}%
%BeginExpansion
{\mathscr{P}}%
%EndExpansion
_{q},\quad \forall q\in \Omega .\
\end{array}
\label{MOP}
\end{equation}%
Practical algorithms to solve the MOP in (\ref{MOP}) can be obtained, using
the approach proposed, e.g., in \cite{Yu_DSL, Das}. The Pareto optimal
solutions to the MOP can also be obtained by solving the following ordinary
scalar optimization% problem
:%
\begin{equation}
\begin{array}{l}
\limfunc{maximize}\limits_{\mathbf{p}}\quad \ \ \
\sum\limits_{q=1}^{Q}\lambda _{q}R_{q}(\mathbf{p}) \\
\text{subject to}\hspace{0.6cm}\ \mathbf{p}_{q}\in
%TCIMACRO{\TeXButton{P}{{\mathscr{P}}}}%
%BeginExpansion
{\mathscr{P}}%
%EndExpansion
_{q},\quad \forall q\in \Omega .%
\end{array}
\label{Scalarized_problem}
\end{equation}%
where $%
%TCIMACRO{\TeXButton{lmd}{\boldsymbol{\lambda}}}%
%BeginExpansion
\boldsymbol{\lambda}%
%EndExpansion
$ is a set of positive weights. For any given $%
%TCIMACRO{\TeXButton{lmd}{\boldsymbol{\lambda}}}%
%BeginExpansion
\boldsymbol{\lambda}%
%EndExpansion
>\mathbf{0},$ the (globally) optimal solution to (\ref{Scalarized_problem})
is a point on the trade-off surface of MOP\ (\ref{MOP}) (cf. Appendix \ref%
{Proof Prop_NE-PO}). However, as the rate region (\ref{Rate}) is in
principle nonconvex, by varying $%
%TCIMACRO{\TeXButton{lmd}{\boldsymbol{\lambda}}}%
%BeginExpansion
\boldsymbol{\lambda}%
%EndExpansion
>\mathbf{0,}$ only a portion of the trade-off curve of (\ref{MOP}) can be
explored. Specifically, all the Pareto optimal points lying on the nonconvex
part of the rate region cannot be obtained as solutions to (\ref%
{Scalarized_problem}). %(see Figure \ref{RR and CH}a)).
We will refer to (\ref{Scalarized_problem}) as the \emph{scalarized} MOP.

Comparing (\ref{Rate Game}) with (\ref{MOP}), we infer that, in general, the
Nash equilibria are not solutions to (\ref{MOP}), and thus they are not
Pareto-efficient. An interesting question is whether one can modify the
payoff function of every player so that some Nash equilibria of the modified
game coincide with the Pareto optimal solutions. The answer is given by the
following proposition.\vspace{-0.1cm}

\begin{proposition}
\label{Prop_NE-PO} Let $\widetilde{{%
%TCIMACRO{\TeXButton{G}{\mathscr{G}}}%
%BeginExpansion
\mathscr{G}%
%EndExpansion
}}\left(
%TCIMACRO{\TeXButton{lmd}{\boldsymbol{\lambda}}}%
%BeginExpansion
\boldsymbol{\lambda}%
%EndExpansion
\right) $ be the game defined as
\begin{equation}
\widetilde{{%
%TCIMACRO{\TeXButton{G}{\mathscr{G}}}%
%BeginExpansion
\mathscr{G}%
%EndExpansion
}}\left(
%TCIMACRO{\TeXButton{lmd}{\boldsymbol{\lambda}}}%
%BeginExpansion
\boldsymbol{\lambda}%
%EndExpansion
\right) =\left\{ \Omega ,\left\{
%TCIMACRO{\TeXButton{P}{{\mathscr{P}}}}%
%BeginExpansion
{\mathscr{P}}%
%EndExpansion
_{q}\right\} _{q\in \Omega },\{\widetilde{{\Phi }}_{q}\}_{q\in \Omega
}\right\} ,  \label{G_tilde}
\end{equation}%
where the payoff functions are
\begin{equation}
\widetilde{{\Phi }}_{q}(\mathbf{p})={R}_{q}(\mathbf{p})+\frac{1}{\lambda _{q}%
}\sum_{r\neq q}\lambda _{r}{R}_{r}(\mathbf{p})\vspace{-0.2cm}
\label{Rae_g_tilde}
\end{equation}%
with $R_{q}(\mathbf{p})$ defined in (\ref{Rate}) and $%
%TCIMACRO{\TeXButton{lmd}{\boldsymbol{\lambda}}}%
%BeginExpansion
\boldsymbol{\lambda}%
%EndExpansion
%TCIMACRO{\TeXButton{def}{\triangleq}}%
%BeginExpansion
\triangleq%
%EndExpansion
\mathbf{[}\lambda _{1},\ldots ,\lambda _{Q}\mathbf{]}^{T}$ is a set of
positive weights. Then, for any given $%
%TCIMACRO{\TeXButton{lmd}{\boldsymbol{\lambda}}}%
%BeginExpansion
\boldsymbol{\lambda}%
%EndExpansion
>\mathbf{0}$, the solution set of $\widetilde{{%
%TCIMACRO{\TeXButton{G}{\mathscr{G}}}%
%BeginExpansion
\mathscr{G}%
%EndExpansion
}}(%
%TCIMACRO{\TeXButton{lmd}{\boldsymbol{\lambda}}}%
%BeginExpansion
\boldsymbol{\lambda}%
%EndExpansion
)$ is not empty and contains the (globally) optimal solution to the
scalarized MOP (\ref{Scalarized_problem}), which is the Pareto optimal
solution to the MOP (\ref{MOP}) corresponding to the point where the
hyperplane with normal vector $%
%TCIMACRO{\TeXButton{lmd}{\boldsymbol{\lambda}}}%
%BeginExpansion
\boldsymbol{\lambda}%
%EndExpansion
$ is tangent to the boundary of the rate region (\ref{RR}).

Moreover, if the conditions of Proposition \ref{ConvexRateRegion} are
satisfied, then:\vspace{-0.2cm}

\begin{enumerate}
\item $\widetilde{{%
%TCIMACRO{\TeXButton{G}{\mathscr{G}}}%
%BeginExpansion
\mathscr{G}%
%EndExpansion
}}(%
%TCIMACRO{\TeXButton{lmd}{\boldsymbol{\lambda}}}%
%BeginExpansion
\boldsymbol{\lambda}%
%EndExpansion
)$ admits a unique NE, for any given $%
%TCIMACRO{\TeXButton{lmd}{\boldsymbol{\lambda}}}%
%BeginExpansion
\boldsymbol{\lambda}%
%EndExpansion
>\mathbf{0},$ and

\item \emph{Any}\footnote{%
We do not consider the rate profiles where the rate of some user is zero,
w.l.o.g.. The corner points of the rate region can be achieved solving a
lower dimensional problem.} Pareto optimal solution to the MOP (\ref{MOP})
can be obtained as the unique NE of $\widetilde{{%
%TCIMACRO{\TeXButton{G}{\mathscr{G}}}%
%BeginExpansion
\mathscr{G}%
%EndExpansion
}}(%
%TCIMACRO{\TeXButton{lmd}{\boldsymbol{\lambda}}}%
%BeginExpansion
\boldsymbol{\lambda}%
%EndExpansion
)$, with a proper choice of $%
%TCIMACRO{\TeXButton{lmd}{\boldsymbol{\lambda}}}%
%BeginExpansion
\boldsymbol{\lambda}%
%EndExpansion
>\mathbf{0}.$
\end{enumerate}
\end{proposition}

\begin{proof}
See Appendix \ref{Proof Prop_NE-PO}.
\end{proof}

Comparing (\ref{Rate}) with (\ref{Rae_g_tilde}), an interesting
interpretation arises: The Pareto-optimal solutions to the MOP (\ref{MOP})
can be achieved as a NE of the modified game where each player incorporates,
in its utility function, a linear combination, through positive coefficient,
of the utilities of the other players.\footnote{%
An alternative approach to move toward Pareto optimality is to allow that
the game could be played more than once, i.e., to consider the so called
\emph{repeated} games \cite{Osborne}, with a proper punishment strategy \cite%
{Supergames-paper}. %The study of dynamic
%games
This study goes beyond the scope of this paper.} The NE of the modified game
$\widetilde{{\ \mathscr{G}}}(\boldsymbol{\lambda })$ in (\ref{G_tilde}) can
be obtained, for any given $\boldsymbol{\lambda },$ using, e.g., the
gradient projection based iterative algorithm, proposed in \cite{Rosen}.
However, this algorithm requires coordination among the players, since, at
each iteration, each user has to know the channels and the strategies
adopted by all the other users. Thus, Pareto-efficiency can be achieved only
at the price of a significant increase of signalling and coordination among
the users and this goes against our search for distributed, independent,
coding/decoding among the users. Note that the structure of (\ref%
{Rae_g_tilde}) generalizes the pricing techniques widely used in the game
theory field to obtain a Pareto improvement of the system performance with
respect to the Nash equilibria of a noncooperative game (see, e.g., \cite%
{Altman-Winter} and references therein).

Using Proposition \ref{Prop_NE-PO}, we can now characterize and quantify the
Nash equilibria of the game ${%
%TCIMACRO{\TeXButton{G}{\mathscr{G}}}%
%BeginExpansion
\mathscr{G}%
%EndExpansion
}$ given in (\ref{Rate Game}) by providing upper and lower bounds.

\begin{proposition}
\label{min_max_Proposition}All the Nash equilibria $\mathbf{p}^{\star }$ of
 game $%
%TCIMACRO{\TeXButton{G}{{\mathscr{G}}}}%
%BeginExpansion
{\mathscr{G}}%
%EndExpansion
$ in (\ref{Rate Game}) satisfy the following inequality%
\begin{equation}
\max\limits_{\mathbf{p}_{q}\in
%TCIMACRO{\TeXButton{P}{{\mathscr{P}}}}%
%BeginExpansion
{\mathscr{P}}%
%EndExpansion
_{q}}\min\limits_{\mathbf{p}_{-q}\in
%TCIMACRO{\TeXButton{P}{{\mathscr{P}}}}%
%BeginExpansion
{\mathscr{P}}%
%EndExpansion
_{-q}}R_{q}(\mathbf{p}_{q},\mathbf{p}_{-q})\leq R_{q}(\mathbf{p}_{q}^{\star
},\mathbf{p}_{-q}^{\star })\leq \widetilde{{\Phi }}_{q}^{\star },\quad
\forall q\in \Omega ,  \label{max_min_ineq}
\end{equation}%
where $\widetilde{{\Phi }}_{q}^{\star }$ is a NE of $\widetilde{{%
%TCIMACRO{\TeXButton{G}{\mathscr{G}}}%
%BeginExpansion
\mathscr{G}%
%EndExpansion
}}(%
%TCIMACRO{\TeXButton{lmd}{\boldsymbol{\lambda}}}%
%BeginExpansion
\boldsymbol{\lambda}%
%EndExpansion
)$ defined in (\ref{G_tilde}) with a proper choice of $%
%TCIMACRO{\TeXButton{lmd}{\boldsymbol{\lambda}}}%
%BeginExpansion
\boldsymbol{\lambda}%
%EndExpansion
.$
\end{proposition}

\begin{proof}
See Appendix \ref{Proof_min_max_Proposition}.
\end{proof}

Thus, at any NE of game ${%
%TCIMACRO{\TeXButton{G}{\mathscr{G}}}%
%BeginExpansion
\mathscr{G}%
%EndExpansion
}$ in (\ref{Rate Game}), the rate of each user is always better than that
obtained by optimizing the worst case (which in general is too pessimistic).
However, in general, the Nash equilibria of ${%
%TCIMACRO{\TeXButton{G}{\mathscr{G}}}%
%BeginExpansion
\mathscr{G}%
%EndExpansion
}$ may be Pareto dominated. We quantify numerically this loss in the next
section.%\vspace{-0.3cm}

\section{Numerical Results}

\label{Sec:Numerical Results}It is interesting to compare the Nash
equilibria of game ${%
%TCIMACRO{\TeXButton{G}{\mathscr{G}}}%
%BeginExpansion
\mathscr{G}%
%EndExpansion
}$ with the rates lying on the boundary of the rate region, in order to
quantify the loss of Nash equilibria with respect to the Pareto optimal
solutions. To this end, we consider the following two different topologies.
In the first example, we assume that the system operates in a symmetric
situation, whereas in the second example we consider an asymmetric scenario.

\noindent \textbf{Example 1}: \emph{Symmetric Case}. In this scenario, the
system is assumed to be symmetric, i.e., the transmitters have the same
power budget and the interference links are at the same distance (i.e., $%
d_{rq}=d_{qr},\,\,\forall q,r$), so that the cross channel gains are
comparable in average sense. In Figure \ref{RR_PO_vs_NE}, we plot the Pareto
optimal points of$\mathcal{\ }$(\ref{MOP}) and the Nash equilibria of game $%
%TCIMACRO{\TeXButton{G}{{\mathscr{G}}}}%
%BeginExpansion
{\mathscr{G}}%
%EndExpansion
$ defined in (\ref{Rate Game}), for a two-users symmetric system. The two
axes represent the rates, in bits/symbol, for the two links. The two pairs
of nodes are placed at different distances, to test situations with
different level of interference. In the picture, we plot: i) the
Pareto-optimal boundary (\ref{RR}), referred to as $\mathcal{R}_{MOP}$
(solid lines); ii) the NE points ($\ast $) of game ${%
%TCIMACRO{\TeXButton{G}{\mathscr{G}}}%
%BeginExpansion
\mathscr{G}%
%EndExpansion
,}$ given in (\ref{Rate Game}); iii) the NE points of the modified game $%
\widetilde{{%
%TCIMACRO{\TeXButton{G}{\mathscr{G}}}%
%BeginExpansion
\mathscr{G}%
%EndExpansion
}}(%
%TCIMACRO{\TeXButton{lmd}{\boldsymbol{\lambda}}}%
%BeginExpansion
\boldsymbol{\lambda}%
%EndExpansion
),$ given in (\ref{G_tilde}), for different values of the vector $%
%TCIMACRO{\TeXButton{lmd}{\boldsymbol{\lambda}}}%
%BeginExpansion
\boldsymbol{\lambda}%
%EndExpansion
$ (squares); and iv) the rate region (referred as $\mathcal{R}_{NE}$)
corresponding to the Nash equilibria achieved by varying the transmit power
of each link, under the constraint that the overall transmit power is fixed
(dashed lines).\footnote{%
Note that the comparison between dashed and solid lines is not totally fair
because all the rates on the boundary of $\mathcal{R}_{MOP}$ are achieved
with the same power constraint $P_{q}$ for each transmitter, whereas the NEs
reported in the dashed lines are obtained assuming only a total power
constraint.} All the Nash equilibria of game ${%
%TCIMACRO{\TeXButton{G}{\mathscr{G}}}%
%BeginExpansion
\mathscr{G}%
%EndExpansion
}$ are reached using the algorithms introduced in Part II of this paper
\cite{Scutari-Part II}, whereas the Nash equilibria of game $\widetilde{{%
%TCIMACRO{\TeXButton{G}{\mathscr{G}}}%
%BeginExpansion
\mathscr{G}%
%EndExpansion
}}(%
%TCIMACRO{\TeXButton{lmd}{\boldsymbol{\lambda}}}%
%BeginExpansion
\boldsymbol{\lambda}%
%EndExpansion
)$\ are reached using the gradient projection algorithm, proposed in \cite%
{Rosen}. %We have set the same power constraint $P_{q}$ for all the
%users, and $%
%%TCIMACRO{\TeXButton{snr}{\mathsf{snr}}}%
%%BeginExpansion
%\mathsf{snr}%
%%EndExpansion
%_{q}=P_{q}/d_{qq}^{2}\sigma _{w}^{2}=20$ dB.
From Figure \ref{RR_PO_vs_NE}, we infer that the Nash equilibria approach
the optimal Pareto curve as the interference level decreases (i.e., the
ratio $d_{rq}/d_{qq}$ increases) at least in the two user case. This is not
surprising as, in the case in which the interference is sufficiently low, the
interaction (interference) among users becomes negligible and the
performance is limited by noise only, not by the interference. On the contrary,
at small interpair distances (i.e., small ratios $d_{rq}/d_{qq}$),
interference becomes the dominant performance limiting factor and the loss
resulting from using the decentralized approach becomes progressively
larger. But the most interesting result is that this loss is limited also in
the case where the links are rather close to each other (we have observed
this result for several independent symmetric channel realizations). This
suggests that, for symmetric systems, the decentralized approach, based on a
game-theoretic formulation, is indeed a viable choice, considering its
greater simplicity with respect to the centralized optimal solution. From
Figure \ref{RR_PO_vs_NE}, we also see that the solutions to the MOP can be
alternatively reached as Nash equilibria of the modified game $\widetilde{{%
%TCIMACRO{\TeXButton{G}{\mathscr{G}}}%
%BeginExpansion
\mathscr{G}%
%EndExpansion
}}$ in (\ref{G_tilde}) (Proposition \ref{Prop_NE-PO}), using the gradient
projection algorithm of \cite{Rosen}. However, this algorithm cannot be
implemented in a distributed way, as it requires the knowledge from each
user of the channels and the power allocations of all the other links.
%This is the price to be paid if one wants to
%achieve the globally optimal solution.
\begin{figure}[tbph]
\vspace{-0.6cm}
\par
\begin{center}
\includegraphics[trim=0.000000in 0.000000in 0.000000in
-0.212435in,width=13cm]{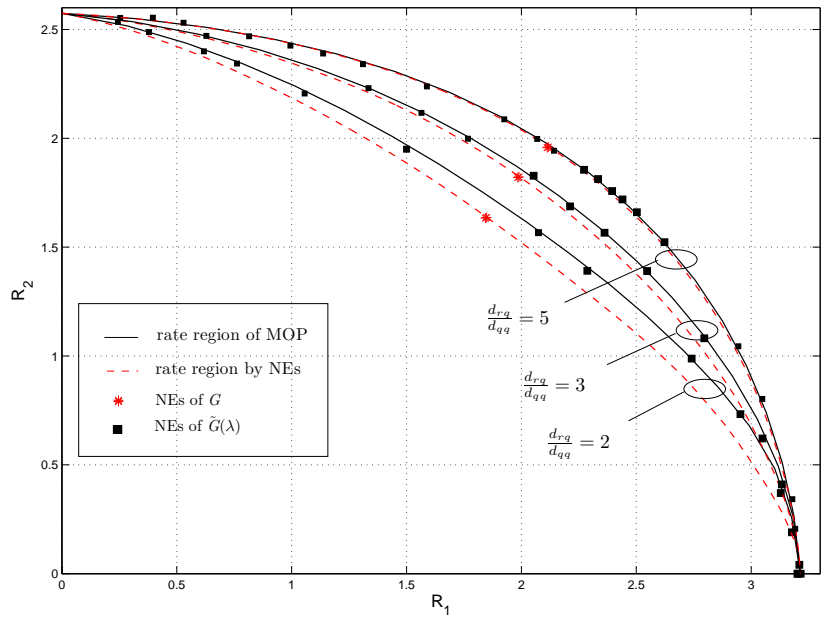}
\end{center}
\par
\vspace{-1.2cm}
\caption{{\protect\footnotesize Rate region achieved for $%
d_{12}/d_{22}=d_{12}/d_{22}=2,3,5$, $d_{11}=d_{22}=1$, $P_{1}=P_{2}$, $\text{%
snr}_{1}=\text{snr}_{2}=10$ dB. Solid line, dashed line curves and asterisks
refer on rate region given in (\protect\ref{RR}), rate region obtained by
NEs of $\widetilde{{%
%TCIMACRO{\TeXButton{G}{\mathscr{G}}}%
%BeginExpansion
\mathscr{G}%
%EndExpansion
}}(%
%TCIMACRO{\TeXButton{lmd}{\boldsymbol{\lambda}}}%
%BeginExpansion
\boldsymbol{\lambda}%
%EndExpansion
)$ given in (\protect\ref{G_tilde}), and NE points of game $\mathscr{G}$
given in (\protect\ref{Rate Game}), respectively.}}
\label{RR_PO_vs_NE}
\end{figure}

\noindent \textbf{Example 2}: \emph{Asymmetric Case}. We consider now a
two-users system operating in an asymmetric situation, where one link
receives a large interference whereas the other does not. This asymmetry can
be due to many reasons, such as different transmission powers and/or
unbalanced cross gains among the users (i.e., $d_{12}<<d_{21}$ or $%
d_{12}>>d_{21}$). For the sake of brevity, in the following we consider only
the latter case, i.e., the situation in which both transmitters have the
same power budget but, because of the location of transmitters and
receivers, one link receives much more interference than the other.
%Similar
%conclusions can be drawn also for the case when the source of
%asymmetry is the transmission power instead of the cross gains.
In Figure \ref{fig:RR_asym} we plot the Pareto optimal points of$\mathcal{\ }
$(\ref{MOP}) and the Nash equilibria of game $%
%TCIMACRO{\TeXButton{G}{{\mathscr{G}}}}%
%BeginExpansion
{\mathscr{G}}%
%EndExpansion
$ defined in (\ref{Rate Game}), for a two-users asymmetric system, for
different values of the ratio $d_{12}/d_{21}$ and a given channel
realization. Low values of $d_{12}/d_{21}$ correspond to high asymmetric
situations. From the figure we infer that as the asymmetry of the system
increases (i.e., $d_{12}/d_{21}$ decreases) the loss of Nash equilibria with
respect to the corresponding Pareto optimal points becomes more significant.
As an example, for the setup considered in the figure, the performance loss
in terms of sum-rate can be as large as $30\%$ of the globally optimal
solution. The same qualitative behavior has been observed changing the
channel realizations and the number of users.
\begin{figure}[tbp]
\centering\vspace{-0.5cm}
\subfigure[\,] {\includegraphics[height=5.5cm,
width=7cm]{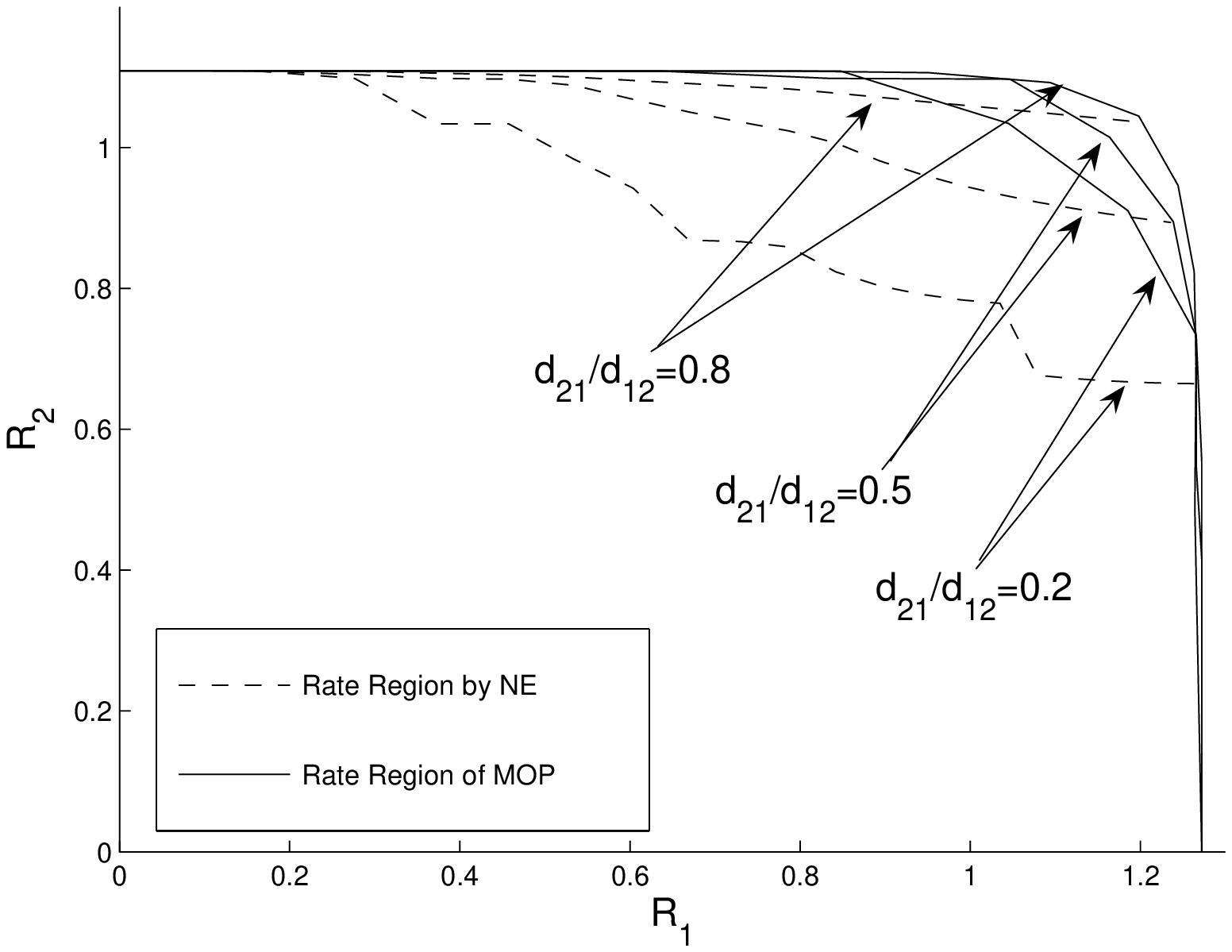}
\label{fig:RR_asym}}\hspace{0.5cm}
\subfigure[]{\includegraphics[height=5cm,
width=6.8cm]{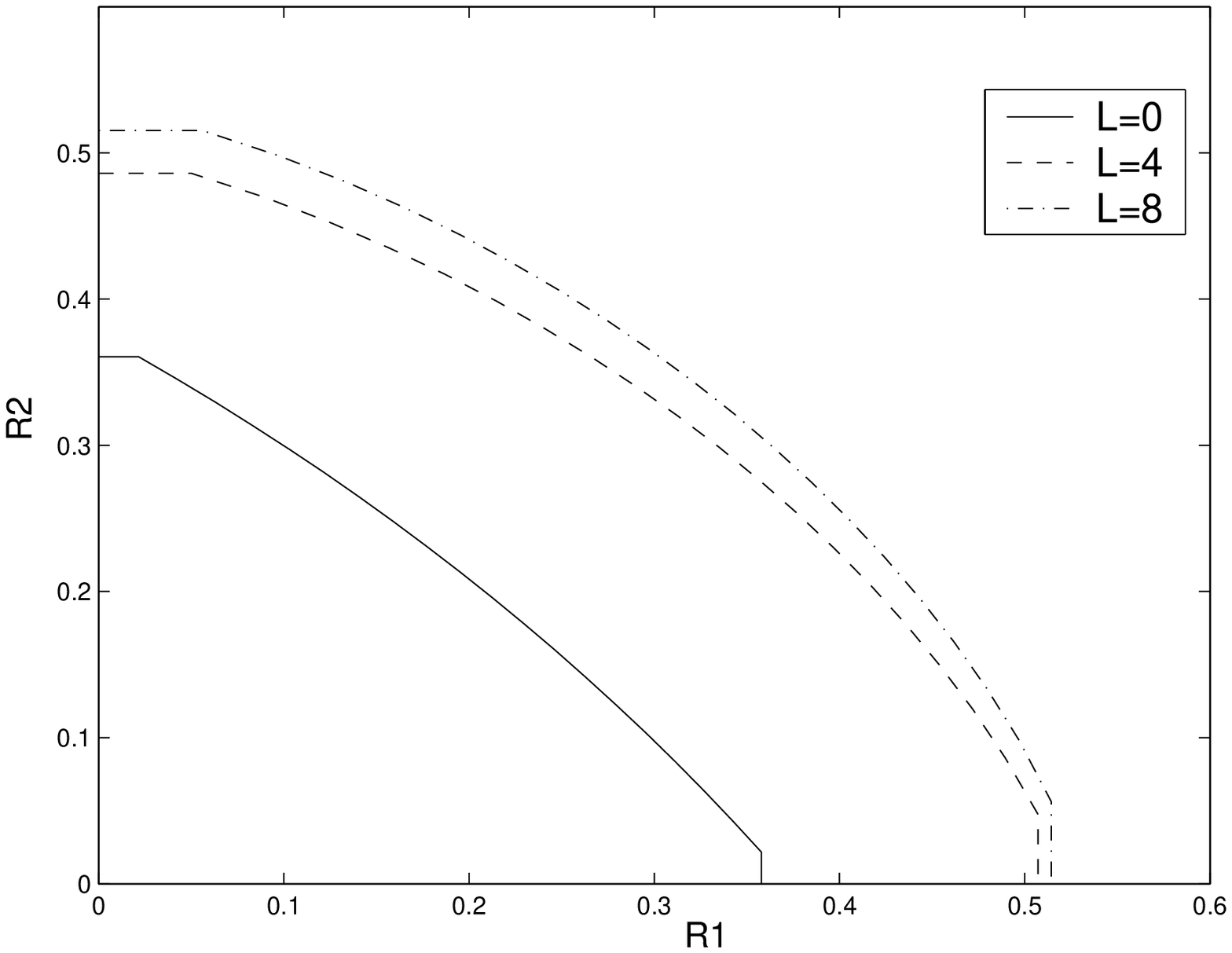}\label{Rate_vs_L}
\hfill}%\vspace{-0.4cm}
\vspace{-0.5cm}% \vspace{-0.4cm}
\caption{{\protect\footnotesize a) Rate region achieved in the asymmetric
case, for $d_{12}/d_{21}=0.2,0.5,0.8$, $d_{11}=d_{22}=1$, $P_{1}=P_{2}$, $%
\text{snr}_{1}=\text{snr}_{2}=5$ dB. Solid line and dashed line curves refer
on rate region given in (\protect\ref{RR}) and rate region obtained by NEs
of $\mathscr{G}$ given in (\protect\ref{Rate Game}), respectively; b) Rate
region achieved for $L=0,4,8$, $N=64$, $d_{12}/d_{22}=d_{21}/d_{11}=1$, $%
d_{11}=d_{22}=1$, $P_{1}=P_{2}$, $\text{snr}_{1}=\text{snr}_{2}=5$ dB.}}
\label{fig1_main}
\end{figure}

\noindent \textbf{Example 3}: \emph{Rate region versus channel order}.
Finally, we study the effect of channel frequency selectivity on the Nash
equilibria of game ${%
%TCIMACRO{\TeXButton{G}{\mathscr{G}}}%
%BeginExpansion
\mathscr{G}%
%EndExpansion
}$ in (\ref{Rate Game}). To this end, we plot in Figure \ref{Rate_vs_L} the
average rate region obtained by the Nash equilibria of ${%
%TCIMACRO{\TeXButton{G}{\mathscr{G}}}%
%BeginExpansion
\mathscr{G}%
%EndExpansion
}$ for different values of the channel order $L_{h}$. The channel taps are
simulated as i.i.d. Gaussian random variables with zero mean and variance $%
1/(L_{h}+1)$. From the figure, we infer that the performance of
decentralized system improves with the channel order. This is due to the
fact that larger frequency fluctuations of the channel provide more degrees
of freedom for each user to find out the best spectral partition for itself,
given its own channel and interference. %In fact, as the
%channel order increases, one can see that the optimal power
%allocations at NE tend to share larger portions of bandwidth, while
%increasing the performance.
\vspace{-0.5cm}

\section{Conclusions}

\vspace{-0.1cm} \label{Conclusions} In this paper we formulated the problem
of finding the optimal precoding/multiplexing strategy in an
infrastructureless multiuser scenario as a noncooperative game. We first
considered the theoretical problem of maximizing mutual information on each
link, given constraints on the spectral mask and transmit power. Then, to
accommodate practical implementation aspects, we focused on the competitive
maximization of the transmission rate on each link, using finite order
constellations, under the same constraints as above plus a constraint on the
average error probability. We proved that in both cases a NE always exists
and the optimal precoding/multiplexing strategy leads to a (pure strategy)
diagonal transmission for all the users. This result strongly simplifies the
optimization, as we can reduce both original complicated matrix-valued games
to a simpler unified vector power control game. Thus, we studied such a
game, and derived sufficient conditions for the uniqueness of the NE, that
were proved to have a broader validity than the conditions known in the
literature for special cases of our game. Then, we compared the totally
decentralized strategy with the Pareto-optimal centralized solution and we
observed that the decentralized strategy has rather low performance loss
with respect to the Pareto-optimal solution, especially for symmetric
systems. Larger losses were observed in the case of very asymmetric systems.
\ In the effort to approach the Pareto-optimal performance, we then showed
how to modify the payoff of each user in order to create a modified game,
whose Nash equilibria are Pareto-optimal. However, this comes at the cost
of extra signaling among the users and breaks the noncooperative feature
of the original games.

Once proved that a Nash equilibrium exists and under which conditions is
unique, the problem boils down to how to reach such an equilibrium. This
problem is addressed in Part II of this paper \cite{Scutari-Part II},
where we provide a variety of distributed algorithms along with their
convergence conditions.\vspace{-0.2cm}

\appendix

\section{Proof of Theorem \protect\ref{Theo_multiccarrier_Info_rate_} \label%
{proof_Theo_multiccarrier_Info_rate}}

\vspace{-0.3cm} We first prove Theorem \ref{Theo_multiccarrier_Info_rate_} for game $%
%TCIMACRO{\TeXButton{G}{\mathscr{G}}}%
%BeginExpansion
\mathscr{G}%
%EndExpansion
_{1}.$ Then, we show that the same result holds true also for game $%
%TCIMACRO{\TeXButton{G}{\mathscr{G}}}%
%BeginExpansion
\mathscr{G}%
%EndExpansion
_{2}.$ \vspace{-0.4cm}

\subsection{Game $%
%TCIMACRO{\TeXButton{G}{\mathscr{G}}}%
%BeginExpansion
\mathscr{G}%
%EndExpansion
_{1}$\label{proof_Theo_1_G_1}}

Given $%
%TCIMACRO{\TeXButton{G}{\mathscr{G}}}%
%BeginExpansion
\mathscr{G}%
%EndExpansion
_{1}$ in (\ref{Rate-matrix-game}), according to the definition of NE (cf.
Definition \ref{NE def}), the proof consists in showing that for each player
$q,$ given the optimal strategy profiles of the others at some NE
(corresponding to the optimal matrices $\mathbf{F}_{r}=\mathbf{W\Sigma }%
_{r}^{1/2},$ with $\mathbf{\Sigma }_{r}%
%TCIMACRO{\TeXButton{def}{\triangleq}}%
%BeginExpansion
\triangleq%
%EndExpansion
\limfunc{diag}(\mathbf{p}_{r}),$ $\forall r\neq q$) i.e., $\mathbf{R}%
_{-q}=\sigma _{q}^{2}\mathbf{I}+\sum\limits_{r\neq q\hfill \hfill }\mathbf{H}%
_{rq}\mathbf{W\Sigma }_{r}\mathbf{W}^{H}\mathbf{H}_{rq}^{H}$, the maximum of
mutual information ${\limfunc{I}\nolimits_{q}}(\mathbf{F}_{q},\mathbf{F}%
_{-q})$ defined in (\ref{R_qq}), under constraints (\ref{P_q_Q_q}), is
achieved by $\mathbf{F}_{q}\mathbf{F}_{q}^{H}=\mathbf{W\Sigma }_{q}\mathbf{W}%
^{H},$ where $\mathbf{\Sigma }_{q}%
%TCIMACRO{\TeXButton{def}{\triangleq}}%
%BeginExpansion
\triangleq%
%EndExpansion
\limfunc{diag}(\mathbf{p}_{q})$ and $\mathbf{p}_{q}$ is solution to (\ref%
{Rate Game}), for fixed $\mathbf{p}_{-q}.$

In the absence of spectral mask constraints, the statement of the theorem
comes directly from the well-known diagonality result of the
capacity-achieving solution to the single user vector Gaussian\ channel \cite%
{Cover}, and the fact that all channel matrices in (\ref{R_qq}) are
diagonalized by the IFFT matrix $\mathbf{W}$:
\begin{equation}
\log \left( \left\vert \mathbf{I}+\mathbf{F}_{q}^{H}\mathbf{H}_{qq}^{H}%
\mathbf{R}_{\mathbf{-}q}^{-1}\mathbf{H}_{qq}\mathbf{F}_{q}\right\vert
\right) =\log \left( \left\vert \mathbf{I}+\mathbf{\Lambda }_{q}\mathbf{Q}%
_{q}\right\vert \right) \leq \sum_{k}\log \left( 1+\left[ \mathbf{\Lambda }%
_{q}\right] _{kk}\left[ \mathbf{Q}_{q}\right] _{kk}\right) ,\vspace{-0.2cm}
\label{Hadamard_ineq}
\end{equation}%
where $\mathbf{Q}_{q}%
%TCIMACRO{\TeXButton{def}{\triangleq}}%
%BeginExpansion
\triangleq%
%EndExpansion
\mathbf{W}^{H}\mathbf{F}_{q}\mathbf{F}_{q}^{H}\mathbf{W,}$ we have used the
eigen-decomposition $\mathbf{H}_{qq}^{H}\mathbf{R}_{\mathbf{-}q}^{-1}\mathbf{%
H}_{qq}=\mathbf{W\Lambda }_{q}\mathbf{W}^{H},$ with
\begin{equation}
\left[ \mathbf{\Lambda }_{q}\right] _{kk}=\frac{\left\vert
H_{qq}(k)\right\vert ^{2}}{1+\sum_{\,r\neq q}\left\vert H_{rq}(k)\right\vert
^{2}p_{r}(k)},\quad k\in \{1,\ldots ,N\},  \label{Lambda_q}
\end{equation}
and the last inequality follows from the Hadamard's inequality \cite{Cover}.
Since the equality in (\ref{Hadamard_ineq}) is reached if and only if $%
\mathbf{Q}_{q}$ is diagonal and the power constraint $\mathsf{Tr}\left(
\mathbf{F}_{q}\mathbf{F}_{q}^{H}\right) =\mathsf{Tr}\left( \mathbf{Q}%
_{q}\right) \leq P_{T}$ depends only on the diagonal elements of $\mathbf{Q}%
_{q},$ we may set $\mathbf{Q}_{q}$ diagonal w.l.o.g., i.e., $\mathbf{Q}_{q}=%
\mathbf{\Sigma }_{q};$ which leads to the desired optimal structure for $%
\mathbf{F}_{q}$.

Interestingly, in the presence of spectral mask constraints, we can still
use the previous result since the additional constraints $\left[ \mathbf{W}%
^{H}\mathbf{F}_{q}\mathbf{F}_{q}^{H}\mathbf{W}\right] _{kk}=\left[ \mathbf{Q}%
_{q}\right] _{kk}\leq \overline{p}_{q}^{\max }(k)$ depend only on the
diagonal elements of $\mathbf{Q}_{q}.$ Introducing the optimal structure $%
\mathbf{F}_{q}=\mathbf{W\Sigma }_{q}^{1/2}$ in $%
%TCIMACRO{\TeXButton{G}{\mathscr{G}}}%
%BeginExpansion
\mathscr{G}%
%EndExpansion
_{1}$, we obtain the simpler game $%
%TCIMACRO{\TeXButton{G}{\mathscr{G}}}%
%BeginExpansion
\mathscr{G}%
%EndExpansion
$ in (\ref{Rate Game}). \vspace{-0.2cm}

\subsection{Game $%
%TCIMACRO{\TeXButton{G}{\mathscr{G}}}%
%BeginExpansion
\mathscr{G}%
%EndExpansion
_{2}$\label{proof_Theo_1_G_2}}

The proof hinges on majorization theory for which the reader is referred to
\cite{Marshall-book} or \cite{Palomar-convex, Palomar_QoS}. The key
definitions and results on which the proof is based are outlined next for
convenience.

\begin{definition}[{\protect\cite[$1.A.1$]{Marshall-book}}]
\label{Def:majorization}For any two vectors $\mathbf{x},\mathbf{y\in
%TCIMACRO{\U{211d} }%
%BeginExpansion
\mathbb{R}
%EndExpansion
}^{n}$, we say $\mathbf{x}$ is majorized by $\mathbf{y}$ or \ $\mathbf{y}$
majorizes $\mathbf{x}$ (denoted by $\mathbf{x}\prec \mathbf{y}$\ or $\mathbf{%
y\succ x}$) if \vspace{-0.4cm}
\begin{eqnarray*}
\sum_{k=1}^{i}x_{[k]} &\leq &\sum_{k=1}^{i}y_{[k]},\quad 1\leq k<n, \\
\sum_{k=1}^{n}x_{[k]} &=&\sum_{k=1}^{n}y_{[k]},
\end{eqnarray*}%
where $x_{[k]}$ and $y_{[k]}$ denote the elements of $\mathbf{x}$ and $%
\mathbf{y}$, respectively, in decreasing order.
\end{definition}

\begin{definition}
\label{Def:Schur convexity/concavity}A real valued function $\phi $ defined
on a set $\mathcal{A}\subseteq
%TCIMACRO{\U{211d} }%
%BeginExpansion
\mathbb{R}
%EndExpansion
^{n}$ is said to be Schur-convex on $\mathcal{A}$ if%
\begin{equation}
\mathbf{x}\prec \mathbf{y\quad }\text{on }\mathcal{A}\text{ }\mathbf{\quad
\Rightarrow \quad }\phi (\mathbf{x})\leq \phi (\mathbf{y}).
\label{def:Schur-Convex}
\end{equation}%
Similarly, $\phi $ is said to be Schur-concave on $\mathcal{A}$ if%
\begin{equation}
\mathbf{x}\prec \mathbf{y\quad }\text{on }\mathcal{A}\text{ }\mathbf{\quad
\Rightarrow \quad }\phi (\mathbf{x})\geq \phi (\mathbf{y}).
\label{def:Schur-Concave}
\end{equation}
\end{definition}

\begin{lemma}[{\protect\cite[p.7]{Marshall-book}, \protect\cite[9.B.1]%
{Marshall-book}}]
\label{Lemma_majorization_ineqs}For a Hermitian matrix $\mathbf{A}$ and a
unitary matrix $\mathbf{U,}$ it follows that%
\begin{equation}
\mathbf{1}(\mathbf{A})\prec \mathbf{d}(\mathbf{U}^{H}\mathbf{AU})\prec
\boldsymbol{\lambda }(\mathbf{A})\mathbf{,}  \label{eq:first_majoriz_ineq}
\end{equation}%
where $\mathbf{d}(\mathbf{A})$ and $\boldsymbol{\lambda }(\mathbf{A})$
denote the diagonal elements and eigenvalues of $\mathbf{A},$ respectively,
and $\mathbf{1}$ denotes the vector with identical components equal to the
average of the diagonal elements of $\mathbf{A}$.
\end{lemma}

Interestingly, matrix $\mathbf{U}$ in (\ref{eq:first_majoriz_ineq}) can
always be chosen such that the diagonal elements are equal to one extreme or
the other. To achieve $\mathbf{d}(\mathbf{U}^{H}\mathbf{AU})=\mathbf{1}(%
\mathbf{A}),$ $\mathbf{U}$ has to be chosen such that $\mathbf{U}^{H}\mathbf{%
AU}$ has equal diagonal elements; to achieve $\mathbf{d}(\mathbf{U}^{H}%
\mathbf{AU})=\boldsymbol{\lambda }(\mathbf{A}),$ $\mathbf{U}$ has to be
chosen to diagonalize matrix $\mathbf{A},$ i.e., equal to the eigenvectors
of $\mathbf{A}.$

As in Appendix \ref{proof_Theo_1_G_1}, the proof consists in showing that,
given the optimal strategy profiles of the others at some NE, the optimal
precoding matrix of user $q$, solution to (\ref{Rate-Game-gap}), is in the
form of (\ref{Optimal_F_q}). \ Thus, hereafter we can focus on a generic
user $q,$ and assume that the strategy profiles of the others are fixed and
equal to $\mathbf{F}_{r}=\mathbf{W\Sigma }_{r}^{1/2},$ with $\mathbf{\Sigma }%
_{r}%
%TCIMACRO{\TeXButton{def}{\triangleq}}%
%BeginExpansion
\triangleq%
%EndExpansion
\limfunc{diag}(\mathbf{p}_{r}),$ $\forall r\neq q.$ We prove the theorem in
two steps. Given $q,$ first, we show the equivalence of the original
complicated problem (\ref{Rate-Game-gap}) and a simpler problem, and then,
we solve the simple problem, implying the optimality of $\mathbf{F}_{q}$ in
the form (\ref{Optimal_F_q}).

Defining $\mathbf{P}_{q}%
%TCIMACRO{\TeXButton{def}{\triangleq}}%
%BeginExpansion
\triangleq%
%EndExpansion
\mathbf{W}^{H}\mathbf{F}_{q},$\ the MSE matrix of user $q$ can be written as%
\begin{equation}
\mathbf{E}_{q}\mathbf{(P}_{q}\mathbf{)}%
%TCIMACRO{\TeXButton{def}{\triangleq}}%
%BeginExpansion
\triangleq%
%EndExpansion
(\mathbf{I+F}_{q}^{H}\mathbf{H}_{qq}^{H}\mathbf{R}_{q}^{-1}\mathbf{H}_{qq}%
\mathbf{F}_{q})^{-1}=(\mathbf{I+F}_{q}^{H}\mathbf{W\Lambda }_{q}\mathbf{W}%
^{H}\mathbf{F}_{q})^{-1}=(\mathbf{I+P}_{q}^{H}\mathbf{\Lambda }_{q}\mathbf{P}%
_{q})^{-1},  \label{def:MSE_matrix}
\end{equation}%
where in the second equality we used $\mathbf{R}_{-q}=\sigma _{q}^{2}\mathbf{%
I}+\sum\limits_{r\neq q\hfill \hfill }\mathbf{H}_{rq}\mathbf{W\Sigma }_{r}%
\mathbf{W}^{H}\mathbf{H}_{rq}^{H}$\ and the eigen-decomposition $\mathbf{H}%
_{qq}^{H}\mathbf{R}_{\mathbf{-}q}^{-1}\mathbf{H}_{qq}=\mathbf{W\Lambda }_{q}%
\mathbf{W}^{H},$ with $\mathbf{\Lambda }_{q}$ defined in (\ref{Lambda_q}).
The matrix $\mathbf{E}_{q}\mathbf{(P}_{q}\mathbf{)}$ has the following
properties: i) $\mathbf{E}_{q}\mathbf{(P}_{q}\mathbf{)}$ is a continuous
function\footnote{%
This result can be proved using \cite[Theorem 10.7.1]{Campell-book}.} of $%
\mathbf{P}_{q}\in
%TCIMACRO{\U{2102} }%
%BeginExpansion
\mathbb{C}
%EndExpansion
^{N\times N}$; ii) For any unitary matrix $\mathbf{U},$ $\mathbf{E}_{q}%
\mathbf{(P}_{q}\mathbf{)}$ satisfies%
\begin{equation}
\mathbf{E}_{q}\mathbf{(P}_{q}\mathbf{U)=U}^{H}\mathbf{E}_{q}\mathbf{(P}_{q}%
\mathbf{)U.}  \label{E_q_property}
\end{equation}%
Using (\ref{def:MSE_matrix}), the payoff function of user $q$ (to be \emph{%
minimized}) in (\ref{Rate-Game-gap}) can be written as a function\ of $%
\mathbf{d}(\mathbf{E}_{q}\mathbf{(P}_{q}\mathbf{))}$: \vspace{-0.4cm}
\begin{equation}
f_{q}(\mathbf{d}(\mathbf{E}_{q}\mathbf{(P}_{q}\mathbf{)}))%
%TCIMACRO{\TeXButton{def}{\triangleq}}%
%BeginExpansion
\triangleq%
%EndExpansion
-\limfunc{r}\nolimits_{q}(\mathbf{F}_{q})=-\frac{1}{N}\dsum\limits_{k=1}^{N}%
\log _{2}\left( 1\mathcal{+}\frac{\left( \left[ \mathbf{E}_{q}\mathbf{(P}_{q}%
\mathbf{)}\right] _{kk}\right) ^{-1}-1}{\Gamma _{q}}\right) ,\vspace{-0.3cm}
\label{-R_q}
\end{equation}%
where $\limfunc{r}\nolimits_{q}(\mathbf{F}_{q})$ is defined in (\ref%
{Rate-gap}) and the last equality follows from the relation $\limfunc{SINR}%
\nolimits_{k,q}(\mathbf{F}_{q})=\left( \left[ \mathbf{E}_{q}\mathbf{(P}_{q}%
\mathbf{)}\right] _{kk}\right) ^{-1}$ $-1,$ with $\limfunc{SINR}%
\nolimits_{k,q}(\mathbf{F}_{q})$ defined in (\ref{SINR_kq_MSE}). \ The
function $f_{q}$ has the following properties: i) $f_{q}(\mathbf{d}(\mathbf{E%
}_{q}\mathbf{(P}_{q}\mathbf{)}))$ is a continuous function of $\mathbf{P}%
_{q}\in
%TCIMACRO{\U{2102} }%
%BeginExpansion
\mathbb{C}
%EndExpansion
^{N\times N}$ (implied from the continuity of $\mathbf{E}_{q}\mathbf{(P}_{q}%
\mathbf{),}$ as pointed out before\textbf{); }ii) $f_{q}(\mathbf{d}(\mathbf{E%
}_{q}\mathbf{(P}_{q}\mathbf{)}))$ is a Schur-concave function on $%
%TCIMACRO{\U{211d} }%
%BeginExpansion
\mathbb{R}
%EndExpansion
_{+}^{N}$ (cf. Definition \ref{Def:Schur convexity/concavity}) \cite%
{Palomar-Barbarossa}.

Using (\ref{def:MSE_matrix}) and (\ref{-R_q}), the optimization problem of
user $q$ in (\ref{Rate-Game-gap}) becomes%
\begin{equation}
\begin{array}{ll}
\limfunc{minimize}\limits_{\mathbf{P}_{q}} & f_{q}(\mathbf{d}(\mathbf{E}_{q}%
\mathbf{(P}_{q}\mathbf{)})) \\
\limfunc{subject}\text{ }\limfunc{to} &
\begin{array}{l}
\dfrac{1}{N}%
%TCIMACRO{\TeXButton{Tr}{\mathsf{Tr}}}%
%BeginExpansion
\mathsf{Tr}%
%EndExpansion
\left( \mathbf{P}_{q}\mathbf{P}_{q}^{H}\right) \leq P_{q}, \\
\mathbf{d}\left( \mathbf{P}_{q}\mathbf{P}_{q}^{H}\right) \leq \overline{%
\mathbf{p}}_{q}^{\max },%
\end{array}%
\end{array}
\tag{P1}  \label{G_2_eq_1}
\end{equation}%
where $\overline{\mathbf{p}}_{q}^{\max }%
%TCIMACRO{\TeXButton{def}{\triangleq}}%
%BeginExpansion
\triangleq%
%EndExpansion
(\overline{p}_{q}^{\max }(k))_{k=1}^{N}.$ Observe that problem P1 always
admits a solution, since the feasible set is closed and bounded (thus
compact) and the objective function is continuous in its interior (as
pointed out before). A priori the solution is unknown, but we assume it
given by an oracle and denoted by\footnote{%
In the case of multiple solutions, we may choose one of them, w.l.o.g..} $%
\mathbf{P}_{q}^{\star }$. As it will be shown next, we do not need to know
explicitly such a solution to complete the proof of the theorem. Then,
problem P1 is equivalent to%
\begin{equation}
\begin{array}{ll}
\limfunc{minimize}\limits_{\mathbf{P}_{q}} & f_{q}(\mathbf{d}(\mathbf{E}_{q}%
\mathbf{(P}_{q}\mathbf{)})) \\
\limfunc{subject}\text{ }\limfunc{to} & \mathbf{d}\left( \mathbf{P}_{q}%
\mathbf{P}_{q}^{H}\right) =\mathbf{d}_{q}^{\star },%
\end{array}
\tag{P2}  \label{G_2_eq_2}
\end{equation}%
where $\mathbf{d}_{q}^{\star }%
%TCIMACRO{\TeXButton{def}{\triangleq}}%
%BeginExpansion
\triangleq%
%EndExpansion
\mathbf{d}\left( \mathbf{P}_{q}^{\star }\mathbf{P}_{q}^{\star H}\right) .$
The equivalence of both problems is in the following sense: 1) The optimal
(and thus feasible) point $\mathbf{P}_{q}^{\star }$ of P1 is feasible in P2
with the same value of the objective function; 2) \ A feasible point in P2
(not necessarily optimal) is feasible also in P1 with the same value of the
objective function.

Thus, for any given $\mathbf{d}_{q}^{\star },$ we can focus on solving
problem P2 w.l.o.g., and show that the optimal solution to P2 is diagonal;
which leads to the desired structure for the original matrix $\mathbf{F}_{q}$
in (\ref{Rate-Game-gap}). Since $f_{q}\left( \mathbf{d}\left( \mathbf{E}_{q}(%
\mathbf{P}_{q})\right) \right) $ is Schur-concave, it follows from
Definition \ref{Def:Schur convexity/concavity} and Lemma \ref%
{Lemma_majorization_ineqs} that $f_{q}\left( \mathbf{d}\left( \mathbf{E}_{q}(%
\mathbf{P}_{q})\right) \right) \geq f_{q}\left( \boldsymbol{\lambda }\left(
\mathbf{E}_{q}(\mathbf{P}_{q})\right) \right) .$ Interestingly, for any
given $\mathbf{P}_{q}$, it is always possible to achieve the lower bound $%
f_{q}\left( \boldsymbol{\lambda }\left( \mathbf{E}_{q}(\mathbf{P}%
_{q})\right) \right) $ of $f_{q}\left( \mathbf{d}\left( \mathbf{E}_{q}(%
\mathbf{P}_{q})\right) \right) $ by using instead the matrix $\widetilde{%
\mathbf{P}}_{q}=\mathbf{P}_{q}\mathbf{U}$ such that $\mathbf{E}_{q}(%
\widetilde{\mathbf{P}}_{q})=\mathbf{U}^{H}\mathbf{E}_{q}\mathbf{(P}_{q}%
\mathbf{)U}$ is diagonal (see (\ref{E_q_property})), without affecting the
constraints, since $\mathbf{d}(\widetilde{\mathbf{P}}_{q}\widetilde{\mathbf{P%
}}_{q}^{H})=\mathbf{d}\left( \mathbf{P}_{q}\mathbf{P}_{q}^{H}\right) .$ In
such a case, $f_{q}(\mathbf{d}(\mathbf{E}_{q}(\widetilde{\mathbf{P}}%
_{q})))=f_{q}(\boldsymbol{\lambda }(\mathbf{E}_{q}(\widetilde{\mathbf{P}}%
_{q})))=f_{q}\left( \mathbf{d}\left( \mathbf{E}_{q}(\mathbf{P}_{q})\right)
\right) .$ This implies that there is an optimal $\mathbf{P}_{q}$ for
problem P2 such that $\mathbf{E}_{q}(\mathbf{P}_{q})$ is diagonal or,
equivalently from (\ref{def:MSE_matrix}), such that $\mathbf{P}_{q}^{H}%
\mathbf{\Lambda }_{q}\mathbf{P}_{q}$ is diagonal.

Now, given that $\mathbf{P}_{q}^{H}\mathbf{\Lambda }_{q}\mathbf{P}_{q}$ is a
diagonal matrix, let us say $\mathbf{\Sigma }_{q},$ i.e., $\mathbf{P}_{q}^{H}%
\mathbf{\Lambda }_{q}\mathbf{P}_{q}=$ $\mathbf{\Sigma }_{q}$, we can write $%
\mathbf{P}_{q}$ as $\mathbf{P}_{q}=\mathbf{\Lambda }_{q}^{-1/2}\mathbf{U}_{q}%
\mathbf{\Sigma }_{q}^{1/2}$ w.l.o.g. (see, e.g., \cite{Palomar-convex,
Palomar_QoS}), where $\mathbf{U}_{q}$ is any unitary matrix. Using such a
structure of $\mathbf{P}_{q},$ problem P1 can be equivalently written as%
\vspace{-0.3cm}
\begin{equation}
\begin{array}{ll}
\limfunc{minimize}\limits_{\mathbf{\Sigma }_{q}\mathbf{,U_{q}}} &
-\dsum\limits_{k=1}^{N}\log _{2}\left( 1\mathcal{+}\dfrac{\left[ \mathbf{%
\Sigma }_{q}\right] _{kk}}{\Gamma _{q}}\right) \\
\limfunc{subject}\text{ }\limfunc{to} & \mathbf{d}\left( \mathbf{U}_{q}%
\mathbf{\Sigma }_{q}\mathbf{U}_{q}^{H}\right) =\overline{\mathbf{d}}%
_{q}^{\star },%
\end{array}
\label{G_2_eq_3}
\end{equation}%
where $\overline{\mathbf{d}}_{q}^{\star }%
%TCIMACRO{\TeXButton{def}{\triangleq}}%
%BeginExpansion
\triangleq%
%EndExpansion
\mathbf{\Lambda }_{q}\mathbf{d}^{\star }$.

From majorization theory (see \cite[5.A.9.A and 9.B.2]{Marshall-book} or
\cite[Lemma 4 and Lemma 3]{Palomar_QoS}) we know that, for any given $%
\mathbf{\Sigma }_{q},$ we can always find a unitary matrix $\mathbf{U}_{q}$
satisfying the constraint in (\ref{G_2_eq_3}) if and only if\vspace{-0.2cm}
\begin{equation}
\boldsymbol{\lambda }(\mathbf{\Sigma }_{q})\succ \overline{\mathbf{d}}%
_{q}^{\star }.  \label{majorization_constraint}
\end{equation}%
Therefore, we can first find the optimal $\mathbf{\Sigma }_{q}$ in (\ref%
{G_2_eq_3}) replacing the original constraint with (\ref%
{majorization_constraint}) and then choose the matrix $\mathbf{U}_{q}$ to
satisfy the constraint in (\ref{G_2_eq_3}) with that optimal $\mathbf{\Sigma
}_{q}.$ Therefore, we have\vspace{-0.2cm}
\begin{equation}
\begin{array}{ll}
\limfunc{minimize}\limits_{\mathbf{\Sigma }_{q}} & -\dsum\limits_{k=1}^{N}%
\log _{2}\left( 1\mathcal{+}\dfrac{\left[ \mathbf{\Sigma }_{q}\right] _{kk}}{%
\Gamma _{q}}\right) \\
\limfunc{subject}\text{ }\limfunc{to} & \mathbf{d}\left( \mathbf{\Sigma }%
_{q}\right) \succ \overline{\mathbf{d}}_{q}^{\star }.%
\end{array}
\label{G_2_eq_4}
\end{equation}

Since the objective function in (\ref{G_2_eq_4}) is Schur-convex, it follows
from Definition \ref{Def:Schur convexity/concavity} that the optimal
solution $\mathbf{\Sigma }_{q}^{\star }$ to (\ref{G_2_eq_4}) is $\mathbf{d}%
\left( \mathbf{\Sigma }_{q}^{\star }\right) =\overline{\mathbf{d}}%
_{q}^{\star }.$ Given such a $\mathbf{\Sigma }_{q}^{\star },$ the constraint
in (\ref{G_2_eq_3}) is satisfied if the matrix $\mathbf{U}_{q}\mathbf{=I}$,
implying for $\mathbf{P}_{q}$ in problem P2 the optimal structure $\mathbf{P}%
_{q}=$ $\mathbf{\Lambda }_{q}^{-1/2}\mathbf{\Sigma }_{q}^{1/2};$ which leads
to the desired expression for $\mathbf{F}_{q}=\mathbf{WP}_{q}=\mathbf{%
W\Lambda }_{q}^{-1/2}\mathbf{\Sigma }_{q}^{1/2}%
%TCIMACRO{\TeXButton{def}{\triangleq}}%
%BeginExpansion
\triangleq%
%EndExpansion
\mathbf{W}\sqrt{\limfunc{diag}(\mathbf{p}_{q})}.$\hspace{\fill}This proves
that all the solutions to problem P1 can be obtained using $\mathbf{F}_{q}=%
\mathbf{W}\sqrt{\limfunc{diag}(\mathbf{p}_{q})}.$%
%TCIMACRO{\TeXButton{QED}{\hspace{\fill}\rule{1.5ex}{1.5ex}}}%
%BeginExpansion
\hspace{\fill}\rule{1.5ex}{1.5ex}%
%EndExpansion
\vspace{-0.4cm}

\section{Proof of Theorem\label{proof of th Existence_uniqueness_NE} \protect
\ref{th:existence_uniqueness_NE}}

\subsection{Existence of a NE}

First, we show the existence of NE for the game ${\
%TCIMACRO{\TeXButton{G}{\mathscr{G}}}%
%BeginExpansion
\mathscr{G}%
%EndExpansion
}$ in {(\ref{Rate Game})} using the following fundamental game theory result:%
\vspace{-0.2cm}

\begin{theorem}
\cite{Osborne, Rosen}\label{Rosen's Theorem} The strategic noncooperative
game ${%
%TCIMACRO{\TeXButton{G}{\mathscr{G}}}%
%BeginExpansion
\mathscr{G}%
%EndExpansion
}=\left\{ \Omega ,\{\mathcal{X}_{q}\}_{q\in \Omega },\{{\Phi }_{q}\}_{q\in
\Omega }\right\} $ admits at least one NE\ if, for all $q\in \Omega $: $1)$
The set $\mathcal{X}_{q}$ is a nonempty compact convex subset of a Euclidean
space\footnote{%
A subset $\mathcal{X}_{q}$ of a Euclidean space is compact if and only if it
is closed and bounded. The set $\mathcal{X}_{q}$ is convex if $\mathbf{x}%
=\theta \mathbf{x}^{(0)}+(1-\theta )\mathbf{x}^{(1)}\in \mathcal{X}_{q}$, $%
\forall \mathbf{x}^{(0)},\mathbf{x}^{(1)}\in \mathcal{X}_{q}$ and $\theta
\in \lbrack 0,1].$}, and $2)$ The payoff function $\Phi _{q}(\mathbf{x})$ is
continuous on $\mathcal{X}$ and quasi-concave\footnote{%
A function $f:{%
%TCIMACRO{\TeXButton{C}{\mathscr{C}}}%
%BeginExpansion
\mathscr{C}%
%EndExpansion
\subseteq \mathcal{%
%TCIMACRO{\U{211d} }%
%BeginExpansion
\mathbb{R}
%EndExpansion
}}^{n}{\mapsto }\mathcal{%
%TCIMACRO{\U{211d} }%
%BeginExpansion
\mathbb{R}
%EndExpansion
}$ is called quasiconvex if its domain ${%
%TCIMACRO{\TeXButton{C}{\mathscr{C}}}%
%BeginExpansion
\mathscr{C}%
%EndExpansion
}$ and all its sublevel sets ${%
%TCIMACRO{\TeXButton{S}{\mathscr{S}}}%
%BeginExpansion
\mathscr{S}%
%EndExpansion
}_{\alpha }=\left\{ \mathbf{x}\in {%
%TCIMACRO{\TeXButton{C}{\mathscr{C}}}%
%BeginExpansion
\mathscr{C}%
%EndExpansion
\ |\ }f\left( \mathbf{x}\right) \leq \alpha \right\} $, for $\alpha \in
\mathcal{%
%TCIMACRO{\U{211d} }%
%BeginExpansion
\mathbb{R}
%EndExpansion
}$, are convex. A function $f$ is quasi-concave if $-f$ is quasi-convex.} on
$\mathcal{X}_{q}$.\vspace{-0.1cm}
\end{theorem}

\noindent The game $\mathscr{G}$ in (\ref{Rate Game}) always admits at least
one NE, because it satisfies conditions required by Theorem \ref{Rosen's
Theorem}: $1)$ The set $%
%TCIMACRO{\TeXButton{P}{{\mathscr{P}}}}%
%BeginExpansion
{\mathscr{P}}%
%EndExpansion
_{q}$ of feasible strategy profiles, given by (\ref{admissible strategy set}%
), is convex and compact (because it is closed and bounded); $2)$ The payoff
function of each player $q$ in (\ref{Rate Game}) is continuous in $\mathbf{p}
$ and (strict) concave in $\mathbf{p}_{q}\in
%TCIMACRO{\TeXButton{P}{{\mathscr{P}}}}%
%BeginExpansion
{\mathscr{P}}%
%EndExpansion
_{q}$ for any given $\mathbf{p}_{-q}$ (this follows from the concavity of
the log function). Hence, it is also quasi-concave \cite{Boyd}.\vspace{-0.2cm%
}

\subsection{Uniqueness of the NE}

\vspace{-0.2cm} To prove the uniqueness we need the following intermediate
results.

\begin{definition}
\label{Def:K matrix}A matrix $\mathbf{A}\in
%TCIMACRO{\U{211d} }%
%BeginExpansion
\mathbb{R}
%EndExpansion
^{n\times n}$ is said to be a $\mathbf{Z}$-matrix if its off-diagonal
entries are all non-positive. A matrix $\mathbf{A}\in
%TCIMACRO{\U{211d} }%
%BeginExpansion
\mathbb{R}
%EndExpansion
^{n\times n}$ is said to be a $\mathbf{P}$-matrix if all its principal
minors are positive. A $\mathbf{Z}$-matrix which is also $\mathbf{P}$ is
called a $\mathbf{K}$-matrix.
\end{definition}

Many equivalent characterizations for the above classes of matrices can be
given. The interested reader is referred to \cite{CPStone92, BPlemmons79}
for more details. Here we focus on the following two properties.

\begin{lemma}[{\protect\cite[Theorem $3.3.4$]{CPStone92}}]
\label{Lemma_P_matrix}A matrix $\mathbf{A}\in
%TCIMACRO{\U{211d} }%
%BeginExpansion
\mathbb{R}
%EndExpansion
^{n\times n}$ is a $\mathbf{P}$-matrix if and only if $\ \mathbf{A}$ does
not reverse the sign of any nonzero vector, i.e.,
\begin{equation}
x_{i}\left[ \mathbf{Ax}\right] _{i}\leq 0\text{ \ for all }i\text{ \ }%
\Rightarrow \text{ \ }\mathbf{x=0}\text{.}
\end{equation}
\end{lemma}

\begin{lemma}[{\protect\cite[Lemma $5.3.14$]{CPStone92}}]
\label{Lemma-comparison-matrix}Let $\mathbf{A}\in
%TCIMACRO{\U{211d} }%
%BeginExpansion
\mathbb{R}
%EndExpansion
^{n\times n}$ be a $\mathbf{K}$-matrix and $\mathbf{B}$ a nonnegative
matrix. Then $\rho (\mathbf{A}^{-1}\mathbf{B})<1$ if and only if $\mathbf{A-B%
}$ is a $\mathbf{K}$-matrix.
\end{lemma}

We provide now sufficient conditions for the uniqueness of the NE. First, we
derive a necessary condition for two admissible strategy profiles (whose
existence is guaranteed by the first part of the Theorem) to be different
Nash equilibria of the game ${%
%TCIMACRO{\TeXButton{G}{\mathscr{G}}}%
%BeginExpansion
\mathscr{G}%
%EndExpansion
}$ in (\ref{Rate Game}){. Then we obtain a simple sufficient condition that
guarantees the previous condition is not satisfied; hence guaranteeing that }%
there cannot be two different Nash equilibria.

Assume that the game ${%
%TCIMACRO{\TeXButton{G}{\mathscr{G}}}%
%BeginExpansion
\mathscr{G}%
%EndExpansion
}$ admits two \emph{distinct} NE points, denoted by $\mathbf{p}^{(0)},%
\mathbf{p}^{(1)},$ where $\mathbf{p}^{(r)}%
%TCIMACRO{\TeXButton{def}{\triangleq}}%
%BeginExpansion
\triangleq%
%EndExpansion
[\mathbf{p}_{1}^{(r)T},\ldots ,\mathbf{p}_{Q}^{(r)T}]^{T}$ and $\mathbf{p}%
_{q}^{(r)}%
%TCIMACRO{\TeXButton{def}{\triangleq}}%
%BeginExpansion
\triangleq%
%EndExpansion
[p_{q}^{(r)}(1),\ldots ,p_{q}^{(r)}(N)]^{T},$ with $r=0,1,$ $q\in \Omega .$
Then, they must satisfy the following (necessary and sufficient)
Karush-Kuhn-Tucker (KKT) conditions (see, e.g., \cite{Boyd}) for each $q\in
\Omega $:\vspace{-0.2cm}
\begin{equation}
\begin{array}{l}
\nabla _{\mathbf{p}_{q}}R_{q}\left( \mathbf{p}^{\left( r\right) }\right)
+\dsum\limits_{k=1}^{2N+1}\mu _{q,k}^{\left( r\right) }\nabla _{\mathbf{p}%
_{q}}f_{q,k}(\mathbf{p}_{q}^{\left( r\right) })=\mathbf{0},\medskip \\
\mu _{q,k}^{\left( r\right) }f_{q,k}(\mathbf{p}_{q}^{\left( r\right) })=0,\
\mathbf{%
%TCIMACRO{\TeXButton{mu}{\boldsymbol{\mu}}}%
%BeginExpansion
\boldsymbol{\mu}%
%EndExpansion
}_{q}^{r}\geq \mathbf{0},\ f_{q,k}(\mathbf{p}_{q}^{\left( r\right) })\geq 0,%
\end{array}%
\qquad \text{ }r=0,1,\vspace{-0.2cm}  \label{KKT_NE}
\end{equation}%
where $\mathbf{%
%TCIMACRO{\TeXButton{mu}{\boldsymbol{\mu}}}%
%BeginExpansion
\boldsymbol{\mu}%
%EndExpansion
}_{q}^{\left( r\right) }%
%TCIMACRO{\TeXButton{def}{\triangleq}}%
%BeginExpansion
\triangleq%
%EndExpansion
[\mu _{q,1}^{\left( r\right) },\ldots ,\mu _{q,2N+1}^{\left( r\right) }]^{T}$
is the nonnegative Lagrange multiplier vector, and
\begin{equation}
\begin{array}{l}
f_{q,k}(\mathbf{p}_{q}^{\left( r\right) })=\left\{
\begin{array}{ll}
p_{q}^{\left( r\right) }(k), & \text{if }1\leq k\leq N, \\
1-\mathbf{1}^{T}\mathbf{p}_{q}^{\left( r\right) }, & \text{if \ }k=N+1, \\
p_{q}^{\max }(k-N-1)-p_{q}^{(r)}(k-N-1),\quad & \text{if }N+1<k\leq 2N+1,%
\end{array}%
\right. \hspace{0.5cm} \\
\left[ \nabla _{\mathbf{p}_{q}}R_{q}\left( \mathbf{p}\right) \right] _{k}=%
\dfrac{\log e}{\Gamma _{q}N}\dfrac{\left\vert H_{qq}(k)\right\vert ^{2}}{%
1+\sum_{r\neq q,r=1}^{Q}\left\vert H_{rq}(k)\right\vert ^{2}p_{r}(k)+\Gamma
_{q}^{-1}\left\vert H_{qq}(k)\right\vert ^{2}p_{q}(k)}\smallskip \\
\hspace{2.25cm}=\dfrac{\log e}{\Gamma _{q}N}\dfrac{|H_{qq}(k)|^{2}}{%
1+\sum_{r=1}^{Q}\Gamma _{q}^{-\delta _{rq}}|H_{rq}(k)|^{2}p_{r}(k)},\medskip
\quad%
\end{array}
\label{KKT_def}
\end{equation}%
with $\delta _{rq}$ defined as $\delta _{rq}=1$ if $r=q$ and $\delta _{rq}=0$
otherwise$.$ Note that the functions $f_{q,k}(\mathbf{p}_{q}^{\left(
r\right) })$ are (linear) concave, thus it must be that, for each $\mathbf{p}%
_{q}^{(1)}$ and $\mathbf{p}_{q}^{(0)},$ \cite{Boyd}%
\begin{equation}
f_{q,k}(\mathbf{p}_{q}^{\left( 1\right) })\leq f_{q,k}(\mathbf{p}%
_{q}^{\left( 0\right) })+\left( \mathbf{p}_{q}^{(1)}-\mathbf{p}%
_{q}^{(0)}\right) ^{T}\nabla _{\mathbf{p}_{q}}\text{ }f_{q,k}(\mathbf{p}%
_{q}^{\left( 0\right) }),\quad \forall k\in \{1,\ldots ,2N+1\},\quad \forall
q\in \Omega .  \label{first_order_condition_concavity}
\end{equation}%
Multiplying the first equation of (\ref{KKT_NE}) by $\left( \mathbf{p}%
_{q}^{(1)}-\mathbf{p}_{q}^{(0)}\right) ^{T}$ for $r=0,$ by $\left( \mathbf{p}%
_{q}^{(0)}-\mathbf{p}_{q}^{(1)}\right) ^{T}$ for $r=1,$ and summing over $r$
we obtain, for each $q\in \Omega :$\vspace{-0.2cm}
\begin{equation}
\begin{array}{l}
\left( \mathbf{p}_{q}^{(1)}-\mathbf{p}_{q}^{(0)}\right) ^{T}\nabla _{\mathbf{%
p}_{q}}R_{q}\left( \mathbf{p}^{\left( 0\right) }\right) +\left( \mathbf{p}%
_{q}^{(0)}-\mathbf{p}_{q}^{(1)}\right) ^{T}\nabla _{\mathbf{p}%
_{q}}R_{q}\left( \mathbf{p}^{\left( 1\right) }\right) + \\
\dsum\limits_{k=1}^{2N+1}\left[ \mu _{q,k}^{\left( 0\right) }\left( \mathbf{p%
}_{q}^{(1)}-\mathbf{p}_{q}^{(0)}\right) ^{T}\nabla _{\mathbf{p}_{q}}f_{q,k}(%
\mathbf{p}_{q}^{\left( 0\right) })+\mu _{q,k}^{1}\left( \mathbf{p}_{q}^{(0)}-%
\mathbf{p}_{q}^{(1)}\right) ^{T}\nabla _{\mathbf{p}_{q}}f_{q,k}(\mathbf{p}%
_{q}^{\left( 1\right) })\right] =0.%
\end{array}%
\vspace{-0.1cm}  \label{KKT_NE2}
\end{equation}%
The second term in (\ref{KKT_NE2}) is clearly nonnegative since it is lower
bounded by\vspace{-0.2cm}
\begin{align}
& \sum\limits_{k=1}^{2N+1}\left[ \mu _{q,k}^{\left( 0\right) }\left( f_{q,k}(%
\mathbf{p}_{q}^{\left( 1\right) })-f_{q,k}(\mathbf{p}_{q}^{\left( 0\right)
})\right) +\mu _{q,k}^{(1)}\left( f_{q,k}(\mathbf{p}_{q}^{\left( 0\right)
})-f_{q,k}(\mathbf{p}_{q}^{\left( 1\right) })\right) \right]  \notag \\
& =\sum\limits_{k=1}^{2N+1}\left[ \mu _{q,k}^{\left( 0\right) }f_{q,k}(%
\mathbf{p}_{q}^{\left( 1\right) })+\mu _{q,k}^{\left( 1\right) }f_{q,k}(%
\mathbf{p}_{q}^{\left( 0\right) })\right] \geq 0,  \label{positive_gamma}
\end{align}%
where we have used the (linearity) concavity of $f_{q,k}$ and the
constraints $\mu _{q,k}^{\left( r\right) }f_{q,k}(\mathbf{p}_{q}^{\left(
r\right) })=0.$

It is now evident that if, for some $q,$ the first term in (\ref{KKT_NE2})
is strictly positive, then $\mathbf{p}^{(0)}$ and $\mathbf{p}^{(1)}$ cannot
be both NE points, as (\ref{KKT_NE2}) would not hold. We will next obtain a
simple sufficient condition for this term to be indeed positive for
different points $\mathbf{p}^{(0)}$ and $\mathbf{p}^{(1)}$; in other words,
to guarantee that there cannot be two different NE solutions.

The positivity condition for the first term in (\ref{KKT_NE2}) is, for some $%
q\in \Omega ,$
\begin{align}
& \left( \mathbf{p}_{q}^{(1)}-\mathbf{p}_{q}^{(0)}\right) ^{T}\left( \nabla
_{\mathbf{p}_{q}}R_{q}\left( \mathbf{p}^{(0)}\right) -\nabla _{\mathbf{p}%
_{q}}R_{q}\left( \mathbf{p}^{(1)}\right) \right)  \notag \\
& =\sum_{k\in \mathcal{D}_{q}}\left( \alpha _{q}\left( k,\mathbf{p}^{(0)},%
\mathbf{p}^{(1)}\right) |H_{qq}(k)|^{2}\left( p_{q}^{(1)}\left( k\right)
-p_{q}^{(0)}\left( k\right) \right) \sum_{r=1}^{Q}\Gamma _{q}^{-\delta
_{rq}}|H_{rq}(k)|^{2}\left( p_{r}^{(1)}(k)-p_{r}^{(0)}(k)\right) \right) >0,
\label{DSC_1}
\end{align}%
where $\alpha _{q}\left( k,\mathbf{p}^{(0)},\mathbf{p}^{(1)}\right)
%TCIMACRO{\TeXButton{def}{\triangleq}}%
%BeginExpansion
\triangleq%
%EndExpansion
\dfrac{\log e}{\Gamma _{q}N}\left( 1+\sum_{r}\Gamma _{q}^{-\delta
_{rq}}|H_{rq}(k)|^{2}p_{r}^{(0)}(k)\right) ^{-1}\left( 1+\sum_{r}\Gamma
_{q}^{-\delta _{rq}}|H_{rq}(k)|^{2}p_{r}^{(1)}(k)\right) ^{-1},$ with $%
\alpha _{q}\left( k,\mathbf{p}^{(0)},\right. $ $\left. \mathbf{p}%
^{(1)}\right) >0$, and $\mathbf{p}^{(0)}\neq \mathbf{p}^{(1)}.$ In (\ref%
{DSC_1}) we have used the fact that, outside $\mathcal{D}_{q},$ we have $%
p_{q}^{(0)}(k)=$ $p_{q}^{(1)}(k)=0$, where $\mathcal{D}_{q}$ is defined in (%
\ref{D_q}).

Define $\mathcal{K}_{q}$ as the set of carriers in $\mathcal{D}_{q}$ such
that the two solutions coincide, i.e., for user $q$:
\begin{equation}
\mathcal{K}_{q}%
%TCIMACRO{\TeXButton{def}{\triangleq}}%
%BeginExpansion
\triangleq%
%EndExpansion
\left\{ k\in \mathcal{D}_{q}\text{ }|\text{ }p_{q}^{(1)}(k)=p_{q}^{(0)}(k)%
\right\} .  \label{K_q}
\end{equation}%
Observe that it cannot be that $\mathcal{K}_{q}=\mathcal{D}_{q}$ for all $q$%
, since $\mathbf{p}^{(0)}$ and $\mathbf{p}^{(1)}$ are different solutions by
assumption. From (\ref{DSC_1}) it follows that $\mathbf{p}^{(0)}$ and $%
\mathbf{p}^{(1)}$ cannot be both NE points if the following sufficient
condition is satisfied:
\begin{equation}
|H_{qq}(k)|^{2}\left( p_{q}^{(1)}\left( k\right) -p_{q}^{(0)}\left( k\right)
\right) \sum_{r=1}^{Q}\Gamma _{q}^{-\delta _{rq}}|H_{rq}(k)|^{2}\left(
p_{r}^{(1)}(k)-p_{r}^{(0)}(k)\right) >0,\quad \forall k\in \mathcal{D}%
_{q}\diagdown \mathcal{K}_{q}\text{ and some }q\in \Omega .
\label{intermediate-condition_0}
\end{equation}%
Since in (\ref{intermediate-condition_0}) $k\in \mathcal{D}_{q}\diagdown
\mathcal{K}_{q},$ it follows that $p_{q}^{(1)}\left( k\right)
-p_{q}^{(0)}\left( k\right) \neq 0.$ The sufficient condition is then, $%
\forall k\in \mathcal{D}_{q}\diagdown \mathcal{K}_{q}$ and some $q\in \Omega
,$%
\begin{equation}
\dfrac{1}{\Gamma _{q}}|H_{qq}(k)|^{2}\left\vert
p_{q}^{(1)}(k)-p_{q}^{(0)}(k)\right\vert +\limfunc{sign}\left(
p_{q}^{(1)}(k)-p_{q}^{(0)}(k)\right) \sum_{r\neq
q,r=1}^{Q}|H_{rq}(k)|^{2}\left( p_{r}^{(1)}(k)-p_{r}^{(0)}(k)\right) >0,
\label{intermediate-condition}
\end{equation}%
where $\limfunc{sign}\left( \cdot \right) $ is the sign function, defined as
$\limfunc{sign}\left( x\right) =1$ if $x>0,$ $\limfunc{sign}\left( x\right)
=0$ if $x=0,$ and $\limfunc{sign}\left( x\right) =-1$ if $x<0.$ \ Using the
fact that $p_{r}^{(1)}\left( k\right) -p_{r}^{(0)}\left( k\right) =0$
whenever $k\notin \mathcal{D}_{r},$ condition (\ref{intermediate-condition})
can be equivalently rewritten as, $\forall k\in \mathcal{D}_{q}\diagdown
\mathcal{K}_{q}$ and some $q\in \Omega ,$%
\begin{equation}
\left\vert p_{q}^{(1)}(k)-p_{q}^{(0)}(k)\right\vert +\limfunc{sign}\left(
p_{q}^{(1)}(k)-p_{q}^{(0)}(k)\right) \sum_{r\neq q,r=1}^{Q}G_{rq}(k)\left(
p_{r}^{(1)}(k)-p_{r}^{(0)}(k)\right) >0,  \label{intermediate-condition_2}
\end{equation}%
with%
\begin{equation}
G_{rq}(k)%
%TCIMACRO{\TeXButton{def}{\triangleq}}%
%BeginExpansion
\triangleq%
%EndExpansion
\left\{
\begin{array}{ll}
\Gamma _{q}\dfrac{|\bar{H}_{rq}(k)|^{2}}{|\bar{H}_{qq}(k)|^{2}}\dfrac{%
d_{qq}^{\alpha }}{d_{rq}^{\alpha }}\dfrac{P_{r}}{P_{q}}, & \text{if }k\in
\mathcal{D}_{r}, \\
0, & \text{otherwise.}%
\end{array}%
\right.  \label{G_rq}
\end{equation}%
A more stringent sufficient condition than (\ref{intermediate-condition_2})
is obtained by considering the worst possible case, i.e., when the second
term in (\ref{intermediate-condition_2}) is as negative as possible:%
\begin{equation}
\Delta _{q}(k)>\sum_{r\neq q,r=1}^{Q}G_{rq}(k)\Delta _{r}(k),\qquad \forall
k\in \mathcal{D}_{q}\diagdown \mathcal{K}_{q}\text{ and some }q\in \Omega ,
\label{SF_on_subcarrier}
\end{equation}%
where $\Delta _{q}(k)$ is defined as\vspace{-0.2cm}
\begin{equation}
\Delta _{q}(k)%
%TCIMACRO{\TeXButton{def}{\triangleq}}%
%BeginExpansion
\triangleq%
%EndExpansion
\left\vert p_{q}^{(1)}(k)-p_{q}^{(0)}(k)\right\vert .  \label{abs_p_0_p1}
\end{equation}%
Note that, since $\mathbf{p}^{(0)}\neq \mathbf{p}^{(1)}$ by assumption, it
must be
\begin{equation}
\Delta _{q}(k)\neq 0,\quad \forall k\in \mathcal{D}_{q}\diagdown \mathcal{K}%
_{q}\text{ \ and some }q\in \Omega .  \label{non_zero_delta}
\end{equation}

Thus far, we have that condition (\ref{SF_on_subcarrier}) can not be
satisfied by the two different NE points $\mathbf{p}^{(0)}$ and $\mathbf{p}%
^{(1)}.$ This means that the following condition needs to be satisfied by
such $\mathbf{p}^{(0)}$ and $\mathbf{p}^{(1)}$:%
\begin{equation}
\Delta _{q}(k_{q})\leq \sum_{r\neq q,r=1}^{Q}G_{rq}(k_{q})\Delta
_{r}(k_{q}),\quad \text{for some }k_{q}\in \mathcal{D}_{q}\diagdown \mathcal{%
K}_{q}\text{ and }\forall q\in \Omega ,  \label{NC_cond}
\end{equation}%
where for the sake of notation we denoted by $k_{q}$ any subcarrier index in
the set $\mathcal{D}_{q}\diagdown \mathcal{K}_{q}.$ Introducing
\begin{equation}
\widetilde{G}_{rq}(k)%
%TCIMACRO{\TeXButton{def}{\triangleq}}%
%BeginExpansion
\triangleq%
%EndExpansion
\left\{
\begin{array}{ll}
G_{rq}(k), & \text{if }k\in \mathcal{D}_{q}, \\
0, & \text{otherwise,}%
\end{array}%
\right.  \label{def:G_tilede}
\end{equation}%
%
%
%
%
%
%and using the fact that $\Delta _{q}(k)=0$ whenever $k\notin \mathcal{D}_{q}$
%or $k\in \mathcal{K}_{q},$
and using the fact that each $\Delta _{q}(k)=0$ whenever $k\notin \mathcal{D}%
_{q}$ and $k\in \mathcal{K}_{q},$ condition (\ref{NC_cond}) becomes
%\begin{equation}
%\Delta _{q}(k)\leq \sum_{r\neq q,r=1}^{Q}\widetilde{G}_{rq}(k)\Delta
%_{r}(k),\quad \text{for some }k\in \{1,\ldots ,N\}\text{ and }\forall q\in
%\Omega .  \label{NC_cond_2}
%\end{equation}%
\begin{equation}
\Delta _{q}(k_{q})\leq \sum_{r\neq q,r=1}^{Q}\widetilde{G}_{rq}(k_{q})\Delta
_{r}(k_{q}),\quad \text{for some }k_{q}\in \mathcal{D}_{q}\diagdown \mathcal{%
K}_{q}\text{ and }\forall q\in \Omega .  \label{NC_cond_2}
\end{equation}%
We rewrite now condition (\ref{NC_cond_2}) in a vector form. To this end, \
we introduce the $NQ$-length vector $%
%TCIMACRO{\TeXButton{bDelta}{\boldsymbol{\Delta}}}%
%BeginExpansion
\boldsymbol{\Delta}%
%EndExpansion
%TCIMACRO{\TeXButton{def}{\triangleq}}%
%BeginExpansion
\triangleq%
%EndExpansion
[%
%TCIMACRO{\TeXButton{bDelta}{\boldsymbol{\Delta}}}%
%BeginExpansion
\boldsymbol{\Delta}%
%EndExpansion
^{T}(1),\ldots ,%
%TCIMACRO{\TeXButton{bDelta}{\boldsymbol{\Delta}}}%
%BeginExpansion
\boldsymbol{\Delta}%
%EndExpansion
^{T}(N)]^{T},$ with $%
%TCIMACRO{\TeXButton{bDelta}{\boldsymbol{\Delta}}}%
%BeginExpansion
\boldsymbol{\Delta}%
%EndExpansion
(k)%
%TCIMACRO{\TeXButton{def}{\triangleq}}%
%BeginExpansion
\triangleq%
%EndExpansion
[\Delta _{1}(k),\ldots ,\Delta _{Q}(k)]^{T}$and the matrices $\overline{%
\mathbf{H}}\in
%TCIMACRO{\U{211d} }%
%BeginExpansion
\mathbb{R}
%EndExpansion
_{+}^{NQ\times NQ}$ and $\overline{\mathbf{H}}(k)\in
%TCIMACRO{\U{211d} }%
%BeginExpansion
\mathbb{R}
%EndExpansion
_{+}^{Q\times Q}$ defined as%
\begin{equation}
\overline{\mathbf{H}}%
%TCIMACRO{\TeXButton{def}{\triangleq}}%
%BeginExpansion
\triangleq%
%EndExpansion
\limfunc{diag}\left( \left\{ \overline{\mathbf{H}}(k)\right\}
_{k=1}^{N}\right) ,\quad \text{and}\quad \left[ \overline{\mathbf{H}}(k)%
\right] _{qr}%
%TCIMACRO{\TeXButton{def}{\triangleq}}%
%BeginExpansion
\triangleq%
%EndExpansion
\left\{
\begin{array}{ll}
\widetilde{G}_{rq}(k), & \text{if }k=k_{q}\text{ and }r\neq q, \\
0, & \text{otherwise.}%
\end{array}%
\right.  \label{def:H_matrix_appendix}
\end{equation}%
Using (\ref{def:H_matrix_appendix}), condition (\ref{NC_cond_2}) can be
equivalently rewritten as%
\begin{equation}
\left( \mathbf{I}-\overline{\mathbf{H}}\right)
%TCIMACRO{\TeXButton{bDelta}{\boldsymbol{\Delta}}}%
%BeginExpansion
\boldsymbol{\Delta}%
%EndExpansion
\leq \mathbf{0}.  \label{NC_cond_vect}
\end{equation}

The proof will be completed by showing that condition (\ref{SF}) is enough
to guarantee that (\ref{NC_cond_vect}) cannot be satisfied by two different
solutions (implying then that (\ref{SF_on_subcarrier}) is satisfied and that
the two different solutions are not NE points). Since $%
%TCIMACRO{\TeXButton{bDelta}{\boldsymbol{\Delta}}}%
%BeginExpansion
\boldsymbol{\Delta}%
%EndExpansion
\geq \mathbf{0,}$ inequality (\ref{NC_cond_vect}) implies%
%\begin{equation}
%\Delta _{q}(k)\left[ \left( \mathbf{I}-\mathbf{H}(k)\right)
%%TCIMACRO{\TeXButton{bDelta}{\boldsymbol{\Delta}}}%
%%BeginExpansion
%\boldsymbol{\Delta}%
%%EndExpansion
%(k)\right] _{q}\leq 0,\quad \text{for some }k\in \{1,\ldots ,N\}\text{ and }%
%\forall q\in \Omega .  \label{NC_cond_vect_2}
%\end{equation}
\begin{equation}
\Delta _{i}\left[ \left( \mathbf{I}-\overline{\mathbf{H}}\right)
%TCIMACRO{\TeXButton{bDelta}{\boldsymbol{\Delta}}}%
%BeginExpansion
\boldsymbol{\Delta}%
%EndExpansion
\right] _{i}\leq 0,\quad \text{ }\forall i=1,\ldots ,NQ,
\label{NC_cond_vect_2}
\end{equation}%
where $\Delta _{i}$ denotes the $i$-th component of $%
%TCIMACRO{\TeXButton{bDelta}{\boldsymbol{\Delta}}}%
%BeginExpansion
\boldsymbol{\Delta}%
%EndExpansion
.$ It follows from Lemma \ref{Lemma_P_matrix} that, if%
\begin{equation}
\mathbf{I}-\overline{\mathbf{H}}\text{ is a }\mathbf{P}\text{-matrix},
\label{P-condition}
\end{equation}%
inequality (\ref{NC_cond_vect_2}) provides %$%
%%TCIMACRO{\TeXButton{bDelta}{\boldsymbol{\Delta}}}%
%%BeginExpansion
%\boldsymbol{\Delta}%
%%EndExpansion
%(k)=\mathbf{0,}$ for all %$k\in \{1,\ldots ,N\}\mathbf{;}$
$\Delta _{q}(k)=0,$ for all $q\in \Omega $ and $k\in \{1,\ldots ,N\}$; which
contradicts the initial assumption that $\mathbf{p}^{(0)}$ and $\mathbf{p}%
^{(1)}$ are two different points (see (\ref{non_zero_delta})). Therefore,
condition (\ref{P-condition}) is sufficient to guarantee the uniqueness of
the NE.

We complete the proof showing that condition (\ref{SF}) implies (\ref%
{P-condition}). We first introduce the matrices $\mathbf{H}\in
%TCIMACRO{\U{211d} }%
%BeginExpansion
\mathbb{R}
%EndExpansion
_{+}^{NQ\times NQ}$ and $\mathbf{H}(k)\in
%TCIMACRO{\U{211d} }%
%BeginExpansion
\mathbb{R}
%EndExpansion
_{+}^{Q\times Q}$ defined as%
\begin{equation}
\mathbf{H}%
%TCIMACRO{\TeXButton{def}{\triangleq}}%
%BeginExpansion
\triangleq%
%EndExpansion
\limfunc{diag}\left( \left\{ \mathbf{H}(k)\right\} _{k=1}^{N}\right) ,\quad
\text{and}\quad \left[ \mathbf{H}(k)\right] _{qr}%
%TCIMACRO{\TeXButton{def}{\triangleq}}%
%BeginExpansion
\triangleq%
%EndExpansion
\left\{
\begin{array}{ll}
\widetilde{G}_{rq}(k), & \text{if }r\neq q, \\
0, & \text{otherwise.}%
\end{array}%
\right.  \label{def:H_matrix_appendix_nobar}
\end{equation}

Observe that $\mathbf{H}(k)$ in (\ref{def:H_matrix_appendix_nobar})
coincides with the matrix defined in (\ref{def:H_matrix}). Since $\mathbf{I}-%
\overline{\mathbf{H}}$ and $\mathbf{I}-\mathbf{H}$ are $\mathbf{Z}$-matrices
(cf. Definition \ref{Def:K matrix}) and $\mathbf{I}-\overline{\mathbf{H}}%
\geq \mathbf{I}-\mathbf{H,}$ where \textquotedblleft $\geq $%
\textquotedblright\ is intended component-wise, a sufficient condition for (%
\ref{P-condition}) is the following \cite{CPStone92}
\begin{equation}
\mathbf{I}-\mathbf{H}\text{ is a }\mathbf{P}\text{-matrix}.
\label{P-condition_2}
\end{equation}

\ Given (\ref{P-condition_2}), $\mathbf{I}-\mathbf{H}$ is a $\mathbf{K}$%
-matrix (cf. Definition \ref{Def:K matrix}). It follows from Lemma \ref%
{Lemma-comparison-matrix} (using the correspondences: $\mathbf{A=I}$ and $%
\mathbf{B=H}$) that condition (\ref{P-condition_2}) is equivalent to $\rho
\left( \mathbf{H}\right) <1;$ which, exploring the block-diagonal structure
of $\mathbf{H}$ in (\ref{def:H_matrix_appendix_nobar}) leads to condition
(C1). This completes the proof.

Condition (\ref{SF_H_max}) in Corollary \ref{Corollary:SF_Uniqueness_H_max}
is obtained using the following result \cite[Corollary $8.1.19$]{Horn85}%
\begin{equation*}
\mathbf{0}\leq \mathbf{H}(k)\leq \mathbf{H}^{\max }\quad \Rightarrow \quad
\rho \left( \mathbf{H}(k)\right) \leq \rho \left( \mathbf{H}^{\max }\right)
\mathbf{.}
\end{equation*}

Conditions (\ref{SF_H_a})-(\ref{SF_H_b}) in Corollary \ref%
{Corollary:SF_Uniqueness_DD} can be obtained as follows. Using \cite[Theorem
$5.6.9$]{Horn85}%
\begin{equation}
\rho \left( \mathbf{H}(k)\right) =\rho \left( \mathbf{H}^{T}(k)\right) \leq
\left\Vert \mathbf{H}(k)\right\Vert ,
\end{equation}%
where $\left\Vert \mathbf{\cdot }\right\Vert $ is any matrix norm \cite%
{Horn85}, a sufficient condition for (\ref{SF}) is $\left\Vert \mathbf{H}%
(k)\right\Vert _{\infty }^{\mathbf{w}}<1,$ with $\left\Vert \mathbf{\cdot }%
\right\Vert _{\infty }^{\mathbf{w}}$ denoting the weighted block maximum
norm, defined as \cite{Horn85}
\begin{equation}
\left\Vert \mathbf{H}(k)\right\Vert _{\infty }^{\mathbf{w}%
}\triangleq\max_{q\in \Omega }\frac{1}{w_{q}}\sum_{r\neq q}w_{r}\left[
\mathbf{H}(k)\right] _{rq},
\end{equation}%
and $\mathbf{w}\triangleq\lbrack w_{1},\ldots ,w_{Q}]^{T}$ is any positive
vector.%
%TCIMACRO{\TeXButton{QED}{\hspace{\fill}\rule{1.5ex}{1.5ex}}}%
%BeginExpansion
\hspace{\fill}\rule{1.5ex}{1.5ex}%
%EndExpansion
\vspace{-0.4cm}

\section{Proof of Proposition \protect\ref{Prop_Subcarrier-allocation} \label%
{Proof_Prop_Subcarrier-allocation}}

We first show that, in the case of high interference, an orthogonal NE
always exists. Then, we prove that if an orthogonal NE exists and condition (%
\ref{medium_interference}) is satisfied, at the equilibrium, the available
subcarriers must be distributed among the users according to (\ref%
{Subcarrier-allocation}).

To prove that an orthogonal NE always exists in the high interference
environment, it is sufficient to show that, at such a point, all the users
do not get any rate increase in sharing some subcarriers. The power
distribution of each user $q$ at any orthogonal NE must satisfy the
single-user waterfilling solution over the subcarriers that are in $\mathcal{%
I}_{q}$, i.e.,
\begin{equation}
\begin{array}{c}
p_{q}^{\star }(k)=\left\{
\begin{array}{l}
\mu _{q}-\dfrac{\Gamma _{q}}{%
%TCIMACRO{\TeXButton{snr}{\mathsf{snr}}}%
%BeginExpansion
\mathsf{snr}%
%EndExpansion
_{q}\left\vert \bar{H}_{qq}(k)\right\vert ^{2}},\qquad \ k\in \mathcal{I}%
_{q}, \\
0,\hspace{3cm}\hspace{0.5cm}\ \text{otherwise.}%
\end{array}%
\right. \quad q\in \Omega .%
\end{array}
\label{FDMA_NE}
\end{equation}

It is straightforward to check that a power distribution as in (\ref{FDMA_NE}%
) always exists, provided that $Q\leq N.$ For example, in the simple case of
$Q=N,$ there are $Q!$ different partitions $\mathcal{I}_{1},$ $\mathcal{I}%
_{2},\ldots ,\mathcal{I}_{Q}$ of the set $\{1,\ldots ,N\}$, corresponding to
all the permutations where every user uses only one carrier; which
guarantees (\ref{FDMA_NE}) to be satisfied.

Observe that a strategy profile as in (\ref{FDMA_NE}) is not necessarily a
NE. However, for any given $\{%
%TCIMACRO{\TeXButton{snr}{\mathsf{snr}}}%
%BeginExpansion
\mathsf{snr}%
%EndExpansion
_{q}\}_{q\in \Omega },$ there exists a set $\{%
%TCIMACRO{\TeXButton{inr}{\mathsf{inr}}}%
%BeginExpansion
\mathsf{inr}%
%EndExpansion
_{rq}^{\star }\}_{r\neq q\in \Omega },$ with each $%
%TCIMACRO{\TeXButton{inr}{\mathsf{inr}}}%
%BeginExpansion
\mathsf{inr}%
%EndExpansion
_{rq}^{\star }>>1,$ such that, for all $%
%TCIMACRO{\TeXButton{inr}{\mathsf{inr}}}%
%BeginExpansion
\mathsf{inr}%
%EndExpansion
_{rq}\geq
%TCIMACRO{\TeXButton{inr}{\mathsf{inr}}}%
%BeginExpansion
\mathsf{inr}%
%EndExpansion
_{rq}^{\star },$
\begin{equation*}
\dfrac{\Gamma _{q}}{%
%TCIMACRO{\TeXButton{snr}{\mathsf{snr}}}%
%BeginExpansion
\mathsf{snr}%
%EndExpansion
_{q}\left\vert \bar{H}_{qq}(k_{q})\right\vert ^{2}}+P_{q}<\Gamma _{q}\frac{1+%
%TCIMACRO{\TeXButton{inr}{\mathsf{inr}}}%
%BeginExpansion
\mathsf{inr}%
%EndExpansion
_{rq}\left\vert \bar{H}_{rq}(k_{r})\right\vert ^{2}p_{r}(k_{r})}{%
%TCIMACRO{\TeXButton{snr}{\mathsf{snr}}}%
%BeginExpansion
\mathsf{snr}%
%EndExpansion
_{q}\left\vert \bar{H}_{qq}(k_{r})\right\vert ^{2}},\quad \forall k_{r}\in
\mathcal{I}_{r},\text{ }k_{q}\in \mathcal{I}_{q},\text{ }\forall r\neq q\in
\Omega ,
\end{equation*}%
which is sufficient for $q$-th user to allocate no power over any subcarrier
in $\mathcal{I}_{r}$. Thus, every user $q$, given the power distribution of
the others, does not have any incentive to change its own power allocation.
Hence, the strategy profile given by (\ref{FDMA_NE}) is a NE.

We assume now that an orthogonal NE exists and that condition (\ref%
{medium_interference}) is satisfied. At such a NE the KKT conditions must
hold for all users. Thus, for each $q$ and $k,$ it must be
\begin{equation}
\begin{array}{l}
\dfrac{\log e}{\Gamma _{q}N}\dfrac{|H_{qq}(k)|^{2}}{1+\sum_{t=1}^{Q}\Gamma
_{q}^{-\delta _{qt}}p_{t}(k)|H_{tq}(k)|^{2}}\leq \mu _{q},\smallskip \\
p_{q}(k)\geq 0,\ (1/N)\dsum_{k=1}^{N}p_{q}(k)=1,\ \mu _{q}>0,%
\end{array}%
\qquad \forall q\in \Omega ,\forall k\in \{1,\ldots ,N\},  \label{KKTConds}
\end{equation}%
where we used the fact that at the optimum the power constraint must be
satisfied with equality and $H_{rq}(k)=%
%TCIMACRO{\TeXButton{inr}{\mathsf{inr}}}%
%BeginExpansion
\mathsf{inr}%
%EndExpansion
_{rq}^{1/2}\bar{H}_{rq}(k),$ $\forall r\neq q$ and $H_{qq}(k)=%
%TCIMACRO{\TeXButton{snr}{\mathsf{snr}}}%
%BeginExpansion
\mathsf{snr}%
%EndExpansion
_{q}^{1/2}\bar{H}_{qq}(k),$ $\forall q\in \Omega .$ Note that in the first
expression equality is met if $p_{q}(k)>0$. Now, we characterize the set $%
\mathcal{I}_{q}.$ Assume that $k_{q}\in \mathcal{I}_{q}$. Using (\ref%
{KKTConds}) we can write%
\begin{equation}
\dfrac{|H_{qq}(k_{q})|^{2}}{1+\Gamma _{q}^{-1}p_{q}(k_{q})|H_{qq}(k_{q})|^{2}%
}=\mu _{q},\qquad \dfrac{|H_{rr}(k_{q})|^{2}}{%
1+p_{q}(k_{q})|H_{qr}(k_{q})|^{2}}\leq \mu _{r},\qquad \mu _{q}>0,\smallskip
\ \mu _{r}>0,\ r\neq q.  \label{kq}
\end{equation}%
It follows that %\begin{equation}
%\dfrac{\dfrac{|H_{rr}(k_{q})|^{2}}{1+p_{q}(k_{q})|H%
%_{qr}(k_{q})|^{2}}}{\dfrac{|H_{qq}(k_{q})|^{2}}{1+\Gamma
%_{q}^{-1}p_{q}(k_{q})|H_{qq}(k_{q})|^{2}}}=\frac{|H%
%_{rr}(k_{q})|^{2}}{|H_{qq}(k_{q})|^{2}}\frac{1+\Gamma
%_{q}^{-1}p_{q}(k_{q})|H_{qq}(k_{q})|^{2}}{1+p_{q}(k_{q})|%
%H_{qr}(k_{q})|^{2}}\leq \frac{\mu _{r}}{\mu _{q}}.
%\label{firstpair}
%\end{equation}%
%Similarly, for $k_{r}\in \mathcal{I}_{r},$ it must be%
%\begin{equation}
%\frac{\dfrac{|H_{rr}(k_{r})|^{2}}{1+\Gamma _{r}^{-1}p_{r}(k_{r})|%
%H_{rr}(k_{r})|^{2}}}{\dfrac{|H_{qq}(k_{r})|^{2}}{%
%1+p_{r}(k_{r})|H_{rq}(k_{r})|^{2}}}=\frac{|H%
%_{rr}(k_{r})|^{2}}{|H_{qq}(k_{r})|^{2}}\frac{1+p_{r}(k_{r})|%
%H_{rq}(k_{r})|^{2}}{1+\Gamma _{r}^{-1}p_{r}(k_{r})|H%
%_{rr}(k_{r})|^{2}}\geq \frac{\mu _{r}}{\mu _{q}},\qquad \mu _{r}>0,\ \mu
%_{q}>0,\smallskip \ r\neq q.  \label{secondpair}
%\end{equation}%
%From (\ref{firstpair}) and (\ref{secondpair}) $\mu _{r}/\mu _{q}$ has been
%upper and lower bounded by
\begin{equation}
\frac{|H_{rr}(k_{q})|^{2}}{|H_{qq}(k_{q})|^{2}}\frac{1+\Gamma
_{q}^{-1}p_{q}(k_{q})|H_{qq}(k_{q})|^{2}}{1+p_{q}(k_{q})|H_{qr}(k_{q})|^{2}}%
\leq \frac{\mu _{r}}{\mu _{q}}\leq \frac{|H_{rr}(k_{r})|^{2}}{%
|H_{qq}(k_{r})|^{2}}\frac{1+p_{r}(k_{q})|H_{rq}(k_{r})|^{2}}{1+\Gamma
_{r}^{-1}p_{r}(k_{r})|H_{rr}(k_{r})|^{2}}.  \label{INEQ_1}
\end{equation}

Using (\ref{medium_interference}) we have%
\begin{equation}
\dfrac{|H_{rr}(k_{q})|^{2}}{|H_{qq}(k_{q})|^{2}}\dfrac{1+\Gamma
_{q}^{-1}p_{q}(k_{q})|H_{qq}(k_{q})|^{2}}{1+p_{q}(k_{q})|H_{qr}(k_{q})|^{2}}%
\geq \dfrac{|H_{rr}(k_{q})|^{2}}{|H_{qq}(k_{q})|^{2}},\text{ and }\dfrac{%
|H_{rr}(k_{r})|^{2}}{|H_{qq}(k_{r})|^{2}}\dfrac{%
1+p_{r}(k_{r})|H_{rq}(k_{r})|^{2}}{1+\Gamma
_{r}^{-1}p_{r}(k_{r})|H_{rr}(k_{r})|^{2}}\leq \dfrac{|{H}_{rr}(k_{r})|^{2}} {%
|H_{qq}(k_{r})|^{2}}.  \label{INEQ_2}
\end{equation}%
Introducing (\ref{INEQ_2}) in (\ref{INEQ_1}) we obtain the desired relation.%
%TCIMACRO{\TeXButton{QED}{\hspace{\fill}\rule{1.5ex}{1.5ex}}}%
%BeginExpansion
\hspace{\fill}\rule{1.5ex}{1.5ex}%
%EndExpansion
\vspace{-0.4cm}

\section{\label{Proof ConvexityRR}Proof of Proposition \protect\ref%
{ConvexRateRegion}}

Under the conditions $%
%TCIMACRO{\TeXButton{inr}{\mathsf{inr}}}%
%BeginExpansion
\mathsf{inr}%
%EndExpansion
_{rq}<<1$ and $%
%TCIMACRO{\TeXButton{snr}{\mathsf{snr}}}%
%BeginExpansion
\mathsf{snr}%
%EndExpansion
_{q}>>1,$ $\forall q,r\neq q\in \Omega ,$ each rate $R_{q}(\mathbf{p})$ can
be approximated with a negligible error by\footnote{%
There always exist a set of $%
%TCIMACRO{\TeXButton{inr}{\mathsf{inr}}}%
%BeginExpansion
\mathsf{inr}%
%EndExpansion
_{rq}$ and $%
%TCIMACRO{\TeXButton{snr}{\mathsf{snr}}}%
%BeginExpansion
\mathsf{snr}%
%EndExpansion
_{q}$ such that all the links transmit over the whole bandwidth.}%
\begin{equation}
R_{q}\left( \mathbf{p}\right) \approx \dfrac{1}{N}\dsum_{k=1}^{N}\log \left(
\dfrac{\Gamma _{q}^{-1}%
%TCIMACRO{\TeXButton{snr}{\mathsf{snr}}}%
%BeginExpansion
\mathsf{snr}%
%EndExpansion
_{q}\left\vert \bar{H}_{qq}(k)\right\vert ^{2}p_{q}(k)}{1+\sum_{\,r\neq
q,r=1}^{Q}%
%TCIMACRO{\TeXButton{inr}{\mathsf{inr}}}%
%BeginExpansion
\mathsf{inr}%
%EndExpansion
_{rq}\left\vert \bar{H}_{rq}(k)\right\vert ^{2}p_{r}(k)}\right) =\dfrac{1}{N}%
\dsum_{k=1}^{N}\log \left( \dfrac{G_{qq}(k)p_{q}(k)}{1+\sum_{\,r\neq
q,r=1}^{Q}G_{rq}(k)p_{r}(k)}\right)  \label{R_low_int}
\end{equation}%
where $G_{qq}(k)%
%TCIMACRO{\TeXButton{def}{\triangleq}}%
%BeginExpansion
\triangleq%
%EndExpansion
\Gamma _{q}^{-1}%
%TCIMACRO{\TeXButton{snr}{\mathsf{snr}}}%
%BeginExpansion
\mathsf{snr}%
%EndExpansion
_{q}\left\vert \bar{H}_{qq}(k)\right\vert ^{2}$, $G_{rq}(k)%
%TCIMACRO{\TeXButton{def}{\triangleq}}%
%BeginExpansion
\triangleq%
%EndExpansion
%TCIMACRO{\TeXButton{inr}{\mathsf{inr}}}%
%BeginExpansion
\mathsf{inr}%
%EndExpansion
_{rq}\left\vert \bar{H}_{rq}(k)\right\vert ^{2}$ and $\mathbf{p}\in \mathcal{%
P}%
%TCIMACRO{\TeXButton{def}{\triangleq}}%
%BeginExpansion
\triangleq%
%EndExpansion
\left\{ \mathbf{x}\in \mathcal{%
%TCIMACRO{\U{211d} }%
%BeginExpansion
\mathbb{R}
%EndExpansion
}_{++}^{QN}:\mathbf{1}^{T}\mathbf{x}_{q}\leq 1,\ x_{q}(k)\leq \right. $ $%
\left. p_{q}^{\max }(k),\text{ }\forall k\in \{1,\ldots ,N\},\text{ }\forall
q\in \Omega \right\} .$ To prove the convexity of the rate region defined in
(\ref{RR}), we need to show that\vspace{-0.2cm}
\begin{equation}
\exists \text{ }\mathbf{p}\in \mathcal{P}\text{ }|\text{ }\alpha \bar{%
\mathbf{r}}+\beta \tilde{\mathbf{r}}\leq \mathbf{R}(\mathbf{p}),
\label{convex_comb}
\end{equation}%
whenever\vspace{-0.2cm}
\begin{equation}
\begin{array}{c}
\bar{\mathbf{r}}\leq \mathbf{R}(\bar{\mathbf{p}}),\text{ for some }\bar{%
\mathbf{p}}\in \mathcal{P}\text{;} \\
\tilde{\mathbf{r}}\leq \mathbf{R}(\tilde{\mathbf{p}}),\text{ for some }%
\tilde{\mathbf{p}}\in \mathcal{P}\text{;}%
\end{array}
\label{rates_in_RR}
\end{equation}%
and $\alpha \mathbf{,}$ $\beta \geq 0$, $\alpha +\beta =1$. Starting from (%
\ref{convex_comb}) and using (\ref{rates_in_RR}) we have, for each $q,$%
%\begin{align}
%\alpha \bar{r}_{q}+\beta \tilde{r}_{q}& \leq \alpha R_{q}(\bar{\mathbf{p}}%
%)+\beta R_{q}(\tilde{\mathbf{p}})  \notag \\
%& =\dfrac{1}{N}\dsum_{k=0}^{N-1}\left( \alpha \log \left( \dfrac{G_{qq}(k)e^{%
%\bar{x}_{q}(k)}}{1+\sum_{\,r\neq q}G_{rq}(k)e^{\bar{x}_{r}(k)}}\right)
%+\beta \log \left( \dfrac{G_{qq}(k)e^{\tilde{x}_{q}(k)}}{1+\sum_{\,r\neq
%q}G_{rq}(k)e^{\tilde{x}_{r}(k)}}\right) \right)  \label{convex_comb_rates}
%\end{align}%
\begin{equation}
\alpha \bar{r}_{q}+\beta \tilde{r}_{q}\leq \dfrac{1}{N}\dsum_{k=1}^{N}\left(
\alpha \log \left( \dfrac{G_{qq}(k)e^{\bar{x}_{q}(k)}}{1+\sum_{\,r\neq
q}G_{rq}(k)e^{\bar{x}_{r}(k)}}\right) +\beta \log \left( \dfrac{G_{qq}(k)e^{%
\tilde{x}_{q}(k)}}{1+\sum_{\,r\neq q}G_{rq}(k)e^{\tilde{x}_{r}(k)}}\right)
\right)  \label{convex_comb_rates}
\end{equation}%
where $\bar{x}_{q}(k)=\log _{e}\bar{p}_{q}(k)$ and $\tilde{x}_{q}(k)=\log
_{e}\tilde{p}_{q}(k),$ $\forall q\in \Omega .$ Since the function $\log
\left( \left( 1+\sum_{\,r\neq q}G_{rq}e^{x_{r}}\right)
G_{qq}^{-1}e^{-x_{q}}\right) ,$ with positive $\left\{ G_{rq}\right\} ,$ is
convex in $\mathbf{x}\in \mathcal{%
%TCIMACRO{\U{211d} }%
%BeginExpansion
\mathbb{R}
%EndExpansion
}^{Q}$ ( the function $\log \left( 1+\sum_{\,q=1}^{n}e^{x_{q}}\right) $ is
convex in $\mathcal{%
%TCIMACRO{\U{211d} }%
%BeginExpansion
\mathbb{R}
%EndExpansion
}^{n}$ \cite{Boyd})$,$ expression (\ref{convex_comb_rates}) can be
upper-bounded by %\begin{align*}
%\alpha \bar{r}_{q}+\beta \tilde{r}_{q}& \leq \alpha R_{q}(\bar{\mathbf{p}}%
%)+\beta R_{q}(\tilde{\mathbf{p}})\leq \dfrac{1}{N}\dsum_{k=0}^{N-1}\log
%\left( \dfrac{G_{qq}(k)e^{\alpha \bar{x}_{q}(k)+\beta \tilde{x}_{q}(k)}}{%
%1+\sum_{\,r\neq q}G_{rq}(k)e^{\alpha \bar{x}_{r}(k)+\beta \tilde{x}_{r}(k)}}%
%\right) \\
%&=\dfrac{1}{N}\dsum_{k=0}^{N-1}\log \left( \dfrac{G_{qq}(k)\bar{p}%
%_{q}^{\alpha }(k)\tilde{p}_{q}^{\beta }(k)}{1+\sum_{\,r\neq q}G_{rq}(k)\bar{p%
%}_{r}^{\alpha }(k)\tilde{p}_{r}^{\beta }(k)}\right)
%%TCIMACRO{\TeXButton{def}{\triangleq}}%
%%BeginExpansion
%\triangleq%
%%EndExpansion
%R_{q}(\mathbf{p})
%\end{align*}%
\begin{equation}
\alpha \bar{r}_{q}+\beta \tilde{r}_{q}\leq \dfrac{1}{N}\dsum_{k=1}^{N}\log
\left( \dfrac{G_{qq}(k)\bar{p}_{q}^{\alpha }(k)\tilde{p}_{q}^{\beta }(k)}{%
1+\sum_{\,r\neq q}G_{rq}(k)\bar{p}_{r}^{\alpha }(k)\tilde{p}_{r}^{\beta }(k)}%
\right)
%TCIMACRO{\TeXButton{def}{\triangleq}}%
%BeginExpansion
\triangleq%
%EndExpansion
R_{q}(\mathbf{p}),
\end{equation}%
where $\mathbf{p}$ is given by $p_{q}(k)=\bar{p}_{q}^{\alpha }(k)\tilde{p}%
_{q}^{\beta }(k).$ Observe that $\mathbf{p}$ is feasible as $\mathbf{p}\leq
\alpha \mathbf{\bar{p}}+\beta \mathbf{\tilde{p}}\in \mathcal{P}$, where we
have used the geometric-arithmetic inequality \cite{Rockafellar} and the
fact that $\mathcal{P}$ is convex.%
%TCIMACRO{\TeXButton{QED}{\hspace{\fill}\rule{1.5ex}{1.5ex}}}%
%BeginExpansion
\hspace{\fill}\rule{1.5ex}{1.5ex}%
%EndExpansion
\vspace{-0.3cm}

\section{Proof of Proposition \protect\ref{Prop_NE-PO}\label{Proof
Prop_NE-PO}}

After showing that the (globally) optimal solution to the scalarized MOP (%
\ref{Scalarized_problem}) corresponds, for any given $%
%TCIMACRO{\TeXButton{lmd}{\boldsymbol{\lambda}}}%
%BeginExpansion
\boldsymbol{\lambda}%
%EndExpansion
>\mathbf{0},$ to one of the Pareto optimal points of the MOP (\ref{MOP}), we
prove that such a solution is a NE of the game $\widetilde{{%
%TCIMACRO{\TeXButton{G}{\mathscr{G}}}%
%BeginExpansion
\mathscr{G}%
%EndExpansion
}}\left(
%TCIMACRO{\TeXButton{lmd}{\boldsymbol{\lambda}}}%
%BeginExpansion
\boldsymbol{\lambda}%
%EndExpansion
\right) $ defined in (\ref{G_tilde}). We then show that, under conditions of
Proposition \ref{ConvexRateRegion}, the game $\widetilde{{%
%TCIMACRO{\TeXButton{G}{\mathscr{G}}}%
%BeginExpansion
\mathscr{G}%
%EndExpansion
}}\left(
%TCIMACRO{\TeXButton{lmd}{\boldsymbol{\lambda}}}%
%BeginExpansion
\boldsymbol{\lambda}%
%EndExpansion
\right) $ admits a unique NE, for any given $%
%TCIMACRO{\TeXButton{lmd}{\boldsymbol{\lambda}}}%
%BeginExpansion
\boldsymbol{\lambda}%
%EndExpansion
>\mathbf{0}.$

Given the MOP (\ref{MOP}), consider the scalarized MOP defined in (\ref%
{Scalarized_problem}). For any given $%
%TCIMACRO{\TeXButton{lmd}{\boldsymbol{\lambda}}}%
%BeginExpansion
\boldsymbol{\lambda}%
%EndExpansion
>\mathbf{0},$ the (globally) optimal solution to (\ref{Scalarized_problem})
is a point on the trade-off surface of MOP\ (\ref{MOP}). This follows
directly from the definition of Pareto optimality.
%In fact, let $\mathbf{p}%
%^{\star }$ be the optimal solution of (\ref{Scalarized_problem}) for a given
%$%
%%TCIMACRO{\TeXButton{lmd}{\boldsymbol{\lambda}}}%
%%BeginExpansion
%\boldsymbol{\lambda}%
%%EndExpansion
%>\mathbf{0}$ and assume that $\mathbf{p}^{\star }$ is not Pareto optimal for
%(\ref{MOP}). This means that there must exist a feasible $\mathbf{p\neq p}%
%^{\star }$ such that $\mathbf{R}(\mathbf{p})\geq \mathbf{R}(\mathbf{p}%
%^{\star })$ and $\mathbf{R}(\mathbf{p})\neq \mathbf{R}(\mathbf{p}^{\star }),$
%where $\mathbf{R}(\mathbf{p})%
%%TCIMACRO{\TeXButton{def}{\triangleq}}%
%%BeginExpansion
%\triangleq%
%%EndExpansion
%[R_{1}(\mathbf{p}),\ldots ,R_{Q}(\mathbf{p})]^{T}.$ Therefore $%
%%TCIMACRO{\TeXButton{lmd}{\boldsymbol{\lambda}}}%
%%BeginExpansion
%\boldsymbol{\lambda}%
%%EndExpansion
%^{T}\mathbf{R}(\mathbf{p})>%
%%TCIMACRO{\TeXButton{lmd}{\boldsymbol{\lambda}}}%
%%BeginExpansion
%\boldsymbol{\lambda}%
%%EndExpansion
%^{T}\mathbf{R}(\mathbf{p}^{\star }),$ which contradicts the assumption that $%
%\mathbf{p}^{\star }$ is an optimal solution of (\ref{Scalarized_problem}).
Using this result, we now derive the relationship between the scalarized MOP
(\ref{Scalarized_problem}) (and thus the MOP (\ref{MOP})) and the game $%
\widetilde{{%
%TCIMACRO{\TeXButton{G}{\mathscr{G}}}%
%BeginExpansion
\mathscr{G}%
%EndExpansion
}}\left(
%TCIMACRO{\TeXButton{lmd}{\boldsymbol{\lambda}}}%
%BeginExpansion
\boldsymbol{\lambda}%
%EndExpansion
\right) $ defined in (\ref{G_tilde}). Let $\mathbf{p}^{\star }$ be the
solution to (\ref{Scalarized_problem}) for a fixed $%
%TCIMACRO{\TeXButton{lmd}{\boldsymbol{\lambda}}}%
%BeginExpansion
\boldsymbol{\lambda}%
%EndExpansion
>\mathbf{0}$. Then
\begin{equation}
\sum\limits_{q=1}^{Q}\lambda _{q}R_{q}(\mathbf{p}^{\star })\geqslant
\sum\limits_{q=1}^{Q}\lambda _{q}R_{q}(\mathbf{p}),\quad \forall \mathbf{p}%
\neq \mathbf{p}^{\star },\text{ }\mathbf{p}_{q},\text{ }\mathbf{p}%
_{q}^{\star }\in
%TCIMACRO{\TeXButton{P}{{\mathscr{P}}}}%
%BeginExpansion
{\mathscr{P}}%
%EndExpansion
_{q},\forall q\in \Omega .  \label{weighted_sum}
\end{equation}%
Hence, setting $\mathbf{p}=(\mathbf{p}_{q},\mathbf{p}_{-q}^{\star }),$ from
the above inequality it follows
\begin{equation}
{R}_{q}(\mathbf{p}_{q}^{\star },\mathbf{p}_{-q}^{\star })+\frac{1}{\lambda
_{q}}\sum\limits_{r\neq q}\lambda _{r}R_{r}(\mathbf{p}_{q}^{\star },\mathbf{p%
}_{-q}^{\star })\geq {R}_{q}(\mathbf{p}_{q},\mathbf{p}_{-q}^{\star })+\frac{1%
}{\lambda _{q}}\sum\limits_{r\neq q}\lambda _{r}R_{r}(\mathbf{p}_{q},\mathbf{%
p}_{-q}^{\star }),\quad \forall q\in \Omega ,  \label{Proposition1}
\end{equation}%
or, equivalently,%
\begin{equation}
\widetilde{{\Phi }}_{q}(\mathbf{p}_{q}^{\star },\mathbf{p}_{-q}^{\star
})\geq \widetilde{{\Phi }}_{q}(\mathbf{p}_{q},\mathbf{p}_{-q}^{\star
}),\quad \quad \forall q\in \Omega .  \label{NE_Gtilde}
\end{equation}%
with $\widetilde{{\Phi }}_{q}(\mathbf{p}_{q},\mathbf{p}_{-q})$ is given by (%
\ref{Rae_g_tilde}). But (\ref{NE_Gtilde}) is indeed the definition of NE for
the game $\widetilde{{%
%TCIMACRO{\TeXButton{G}{\mathscr{G}}}%
%BeginExpansion
\mathscr{G}%
%EndExpansion
}}\left(
%TCIMACRO{\TeXButton{lmd}{\boldsymbol{\lambda}}}%
%BeginExpansion
\boldsymbol{\lambda}%
%EndExpansion
\right) $ given by (\ref{G_tilde}). Hence, it follows that, for any given $%
%TCIMACRO{\TeXButton{lmd}{\boldsymbol{\lambda}}}%
%BeginExpansion
\boldsymbol{\lambda}%
%EndExpansion
>\mathbf{0},$ the set of Nash equilibria of $\widetilde{{%
%TCIMACRO{\TeXButton{G}{\mathscr{G}}}%
%BeginExpansion
\mathscr{G}%
%EndExpansion
}}\left(
%TCIMACRO{\TeXButton{lmd}{\boldsymbol{\lambda}}}%
%BeginExpansion
\boldsymbol{\lambda}%
%EndExpansion
\right) $ contains a Pareto optimal point of (\ref{MOP}). However, the
converse is, in general, not true. Observe that the payoff functions $%
\widetilde{{\Phi }}_{q}(\mathbf{p}_{q},\mathbf{p}_{-q})$ are not
quasiconcave in $\mathbf{p}_{q}$ and, hence, the classical results of game
theory providing sufficient conditions for the existence of a NE (as, e.g.,
Theorem \ref{Rosen's Theorem} in Appendix \ref{proof of th
Existence_uniqueness_NE}) cannot be used for the game $\widetilde{{%
%TCIMACRO{\TeXButton{G}{\mathscr{G}}}%
%BeginExpansion
\mathscr{G}%
%EndExpansion
}}\left(
%TCIMACRO{\TeXButton{lmd}{\boldsymbol{\lambda}}}%
%BeginExpansion
\boldsymbol{\lambda}%
%EndExpansion
\right) $. Still, the game $\widetilde{{%
%TCIMACRO{\TeXButton{G}{\mathscr{G}}}%
%BeginExpansion
\mathscr{G}%
%EndExpansion
}}\left(
%TCIMACRO{\TeXButton{lmd}{\boldsymbol{\lambda}}}%
%BeginExpansion
\boldsymbol{\lambda}%
%EndExpansion
\right) $ admits at least one NE, since a solution to the scalarized MOP (%
\ref{Scalarized_problem}) always exists, for any given $%
%TCIMACRO{\TeXButton{lmd}{\boldsymbol{\lambda}}}%
%BeginExpansion
\boldsymbol{\lambda}%
%EndExpansion
$ (the domain $%
%TCIMACRO{\TeXButton{P}{{\mathscr{P}}}}%
%BeginExpansion
{\mathscr{P}}%
%EndExpansion
_{1}\times \ldots \times
%TCIMACRO{\TeXButton{P}{{\mathscr{P}}}}%
%BeginExpansion
{\mathscr{P}}%
%EndExpansion
_{Q}$ of (\ref{Scalarized_problem}) is compact and the objective function is
continuous in the interior of the domain).

We show now, that under conditions of Proposition \ref{ConvexRateRegion},
the correspondence between the set of Nash equilibria and the Pareto optimal
points of (\ref{MOP}) becomes one-to-one, and, moreover, any Pareto optimal
point is, for a proper $%
%TCIMACRO{\TeXButton{lmd}{\boldsymbol{\lambda}}}%
%BeginExpansion
\boldsymbol{\lambda}%
%EndExpansion
,$ the unique NE of $\widetilde{{%
%TCIMACRO{\TeXButton{G}{\mathscr{G}}}%
%BeginExpansion
\mathscr{G}%
%EndExpansion
}}\left(
%TCIMACRO{\TeXButton{lmd}{\boldsymbol{\lambda}}}%
%BeginExpansion
\boldsymbol{\lambda}%
%EndExpansion
\right) $. Assume that conditions of Proposition \ref{ConvexRateRegion} hold
true. Then, the rate region (\ref{RR}) becomes convex, and thus all the
points on its boundary can be achieved solving, for any given $%
%TCIMACRO{\TeXButton{lmd}{\boldsymbol{\lambda}}}%
%BeginExpansion
\boldsymbol{\lambda}%
%EndExpansion
>\mathbf{0},$ the scalarized MOP in (\ref{Scalarized_problem}) \cite%
{Aubin-book} and thus, as shown above, using the game $\widetilde{{%
%TCIMACRO{\TeXButton{G}{\mathscr{G}}}%
%BeginExpansion
\mathscr{G}%
%EndExpansion
}}\left(
%TCIMACRO{\TeXButton{lmd}{\boldsymbol{\lambda}}}%
%BeginExpansion
\boldsymbol{\lambda}%
%EndExpansion
\right) .$ To complete the proof, we need to show that, in such a case, $%
\widetilde{{%
%TCIMACRO{\TeXButton{G}{\mathscr{G}}}%
%BeginExpansion
\mathscr{G}%
%EndExpansion
}}\left(
%TCIMACRO{\TeXButton{lmd}{\boldsymbol{\lambda}}}%
%BeginExpansion
\boldsymbol{\lambda}%
%EndExpansion
\right) $ admits a unique NE, for any given $%
%TCIMACRO{\TeXButton{lmd}{\boldsymbol{\lambda}}}%
%BeginExpansion
\boldsymbol{\lambda}%
%EndExpansion
>\mathbf{0}.$

Since the game $\widetilde{{%
%TCIMACRO{\TeXButton{G}{\mathscr{G}}}%
%BeginExpansion
\mathscr{G}%
%EndExpansion
}}\left(
%TCIMACRO{\TeXButton{lmd}{\boldsymbol{\lambda}}}%
%BeginExpansion
\boldsymbol{\lambda}%
%EndExpansion
\right) $ satisfies the conditions of Theorem \ref{Rosen's Theorem} (see
Appendix \ref{proof of th Existence_uniqueness_NE}), it is sufficient that
the (sufficient) condition for the uniqueness of the NE (in (\ref{KKT_NE2}))
holds true. Replacing in (\ref{KKT_NE2}) the payoff functions $R_{q}(\mathbf{%
p})$ of ${%
%TCIMACRO{\TeXButton{G}{\mathscr{G}}}%
%BeginExpansion
\mathscr{G}%
%EndExpansion
}$ with the payoff functions $\widetilde{{\Phi }}_{q}\left( \mathbf{p}%
\right) $ of $\widetilde{{%
%TCIMACRO{\TeXButton{G}{\mathscr{G}}}%
%BeginExpansion
\mathscr{G}%
%EndExpansion
}}\left(
%TCIMACRO{\TeXButton{lmd}{\boldsymbol{\lambda}}}%
%BeginExpansion
\boldsymbol{\lambda}%
%EndExpansion
\right) $ and summing over $q,$ we obtain the following sufficient condition
for the uniqueness of the NE of the game $\widetilde{{%
%TCIMACRO{\TeXButton{G}{\mathscr{G}}}%
%BeginExpansion
\mathscr{G}%
%EndExpansion
}}\left(
%TCIMACRO{\TeXButton{lmd}{\boldsymbol{\lambda}}}%
%BeginExpansion
\boldsymbol{\lambda}%
%EndExpansion
\right) $, $\forall \mathbf{p}^{(1)}\neq \mathbf{p}^{(0)},$ $\mathbf{p}%
_{q}^{(1)},\mathbf{p}_{q}^{(0)}\in
%TCIMACRO{\TeXButton{P}{{\mathscr{P}}}}%
%BeginExpansion
{\mathscr{P}}%
%EndExpansion
_{q},\forall q\in \Omega ,$\vspace{-0.3cm}
\begin{equation}
\sum\limits_{q=1}^{Q}\left[ \left( \mathbf{p}_{q}^{(1)}-\mathbf{p}%
_{q}^{(0)}\right) ^{T}\nabla _{\mathbf{p}_{q}}\widetilde{{\Phi }}_{q}\left(
\mathbf{p}^{\left( 0\right) }\right) +\left( \mathbf{p}_{q}^{(0)}-\mathbf{p}%
_{q}^{(1)}\right) ^{T}\nabla _{\mathbf{p}_{q}}\widetilde{{\Phi }}_{q}\left(
\mathbf{p}^{\left( 1\right) }\right) \right] >0.  \label{DSC_condition}
\end{equation}%
Under Proposition \ref{ConvexRateRegion}, $\widetilde{{\Phi }}_{q}\left(
\mathbf{p}_{q},\mathbf{p}_{-q}\right) $ is given by (\ref{Rae_g_tilde}),
where each $R_{q}\left( \mathbf{p}_{q},\mathbf{p}_{-q}\right) $ can be
approximated with a negligible error by (\ref{R_low_int}). Hence, each $%
\widetilde{{\Phi }}_{q}\left( \mathbf{p}_{q},\mathbf{p}_{-q}\right) $
becomes a strict concave function with respect to $\mathbf{p}_{q},$ for any
given $\mathbf{p}_{-q}$ and $%
%TCIMACRO{\TeXButton{lmd}{\boldsymbol{\lambda}}}%
%BeginExpansion
\boldsymbol{\lambda}%
%EndExpansion
>\mathbf{0}$, and thus it satisfies the following inequality \cite{Boyd}%
\begin{equation}
\left( \mathbf{p}_{q}^{(1)}-\mathbf{p}_{q}^{(0)}\right) ^{T}\nabla _{\mathbf{%
p}_{q}}\widetilde{{\Phi }}_{q}\left( \mathbf{p}_{q}^{\left( 0\right) },%
\mathbf{p}_{-q}^{\left( 0\right) }\right) >\widetilde{{\Phi }}_{q}(\mathbf{p}%
_{q}^{\left( 1\right) },\mathbf{p}_{-q}^{\left( 0\right) })-\widetilde{{\Phi
}}_{q}(\mathbf{p}_{q}^{\left( 0\right) },\mathbf{p}_{-q}^{\left( 0\right) }).
\label{strict_concavity_cond}
\end{equation}%
Using (\ref{R_low_int}) and (\ref{strict_concavity_cond}) we find that
condition (\ref{DSC_condition}) is satisfied. Hence the game $\widetilde{{%
%TCIMACRO{\TeXButton{G}{\mathscr{G}}}%
%BeginExpansion
\mathscr{G}%
%EndExpansion
}}\left(
%TCIMACRO{\TeXButton{lmd}{\boldsymbol{\lambda}}}%
%BeginExpansion
\boldsymbol{\lambda}%
%EndExpansion
\right) $ admits, for any $%
%TCIMACRO{\TeXButton{lmd}{\boldsymbol{\lambda}}}%
%BeginExpansion
\boldsymbol{\lambda}%
%EndExpansion
>\mathbf{0},$ a unique NE that corresponds to one of the Pareto optimal
points of the MOP (\ref{MOP}).%
%TCIMACRO{\TeXButton{QED}{\hspace{\fill}\rule{1.5ex}{1.5ex}}}%
%BeginExpansion
\hspace{\fill}\rule{1.5ex}{1.5ex}%
%EndExpansion
\vspace{-0.3cm}

\section{Proof of Proposition \protect\ref{min_max_Proposition}}

\label{Proof_min_max_Proposition}

The upper bound of $R_{q}(\mathbf{p}_{q}^{\star },\mathbf{p}_{-q}^{\star })$
in (\ref{max_min_ineq}) comes directly from (\ref{Rae_g_tilde}). Thus we
prove only the lower bound of $R_{q}(\mathbf{p}_{q}^{\star },\mathbf{p}%
_{-q}^{\star }).$ To this end, we use {\cite[Corollary 37.6.2]{Rockafellar}}
and {\cite[Lemma 36.2]{Rockafellar}}.

Starting from the definition of NE, at any NE $\mathbf{p}^{\star}$ we have
that
\begin{equation}
R_{q}(\mathbf{p}_{q}^{\star },\mathbf{p}_{-q}^{\star })=\max\limits_{\mathbf{%
p}_{q}\in
%TCIMACRO{\TeXButton{P}{{\mathscr{P}}}}%
%BeginExpansion
{\mathscr{P}}%
%EndExpansion
_{q}}\ R_{q}(\mathbf{p}_{q},\mathbf{p}_{-q}^{\star })\geq \min\limits_{%
\mathbf{p}_{-q}\in
%TCIMACRO{\TeXButton{P}{{\mathscr{P}}}}%
%BeginExpansion
{\mathscr{P}}%
%EndExpansion
_{-q}}\max\limits_{\mathbf{p}_{q}\in
%TCIMACRO{\TeXButton{P}{{\mathscr{P}}}}%
%BeginExpansion
{\mathscr{P}}%
%EndExpansion
_{q}}R_{q}(\mathbf{p}_{q},\mathbf{p}_{-q}).  \label{min_max}
\end{equation}

Inequality (\ref{max_min_ineq}) comes directly from (\ref{min_max})
interchanging the order of the supremum and the infimum operators. To this
end, it is sufficient that, for each $q$, the functions $R_{q}(\mathbf{p}%
_{q},\mathbf{p}_{-q})$ in (\ref{Rate}) and the sets $%
%TCIMACRO{\TeXButton{P}{{\mathscr{P}}}}%
%BeginExpansion
{\mathscr{P}}%
%EndExpansion
_{q},$ defined in (\ref{admissible strategy set}), and $%
%TCIMACRO{\TeXButton{P}{{\mathscr{P}}}}%
%BeginExpansion
{\mathscr{P}}%
%EndExpansion
_{-q}%
%TCIMACRO{\TeXButton{def}{\triangleq}}%
%BeginExpansion
\triangleq%
%EndExpansion
%TCIMACRO{\TeXButton{P}{{\mathscr{P}}}}%
%BeginExpansion
{\mathscr{P}}%
%EndExpansion
_{1}\times \ldots \times
%TCIMACRO{\TeXButton{P}{{\mathscr{P}}}}%
%BeginExpansion
{\mathscr{P}}%
%EndExpansion
_{q-1}\times
%TCIMACRO{\TeXButton{P}{{\mathscr{P}}}}%
%BeginExpansion
{\mathscr{P}}%
%EndExpansion
_{q+1}\times \ldots \times
%TCIMACRO{\TeXButton{P}{{\mathscr{P}}}}%
%BeginExpansion
{\mathscr{P}}%
%EndExpansion
_{Q},$ satisfy conditions required by {\cite[Corollary 37.6.2]{Rockafellar}}
and {\cite[Lemma 36.2]{Rockafellar}}. The set $%
%TCIMACRO{\TeXButton{P}{{\mathscr{P}}}}%
%BeginExpansion
{\mathscr{P}}%
%EndExpansion
_{1}\times \ldots \times
%TCIMACRO{\TeXButton{P}{{\mathscr{P}}}}%
%BeginExpansion
{\mathscr{P}}%
%EndExpansion
_{Q}$ is a closed bounded non-empty set and thus conditions in {\cite[%
Corollary 37.6.2]{Rockafellar}} are satisfied. Each function $R_{q}(\mathbf{p%
}_{q},\mathbf{p}_{-q})$ is strictly concave with respect to $\mathbf{p}_{q}$
for any given $\mathbf{p}_{-q}\in
%TCIMACRO{\TeXButton{P}{{\mathscr{P}}}}%
%BeginExpansion
{\mathscr{P}}%
%EndExpansion
_{-q}$ and convex with respect to $\mathbf{p}_{-q}$ for any given $\mathbf{p}%
_{q}\in
%TCIMACRO{\TeXButton{P}{{\mathscr{P}}}}%
%BeginExpansion
{\mathscr{P}}%
%EndExpansion
_{q}.$ Hence each $R_{q}(\mathbf{p}_{q},\mathbf{p}_{-q})$ is a
concave-convex function on $%
%TCIMACRO{\TeXButton{P}{{\mathscr{P}}}}%
%BeginExpansion
{\mathscr{P}}%
%EndExpansion
_{q}\times
%TCIMACRO{\TeXButton{P}{{\mathscr{P}}}}%
%BeginExpansion
{\mathscr{P}}%
%EndExpansion
_{-q}.$\hspace{\fill}\rule{1.5ex}{1.5ex}

\vspace{-0.2cm} %\input{biblio.tex}

\def\baselinestretch{0.3}

\end{document}